\documentclass[11pt]{article}
\usepackage{amsmath}
\usepackage{amssymb}
\usepackage{amsthm}
\usepackage{graphicx}
\usepackage{hyperref}
\usepackage{verbatim}
\usepackage{color}
\usepackage{subfig}
\usepackage{float}
\usepackage{caption}
\usepackage[section]{placeins}
\edef\restoreparindent{\parindent=\the\parindent\relax}
\usepackage{parskip}
\restoreparindent
\usepackage[normalem]{ulem}
\usepackage{fullpage}
\newtheorem{claim}{Claim}
\newcommand{\be}{\begin{equation}}
\newcommand{\ee}{\end{equation}}
\newcommand{\bea}{\begin{eqnarray}}
\newcommand{\eea}{\end{eqnarray}}
\newcommand{\e}{\epsilon}
\newcommand{\cof}{\Lambda}
\newcommand{\sgn}{\mathrm{sgn}}

\newcommand{\hD}{\hat \Delta} 
\newcommand{\hp}{\hat \phi}
\newcommand{\vac}{| 0 \rangle}
\newcommand{\cF}{\mathcal F}
 
\newcommand{\mL}{\mathcal L}
\newcommand{\li}{{\mathrm{Li}}}
\newcommand{\re}{{\mathrm{Re}}}
\newcommand{\im}{\mathrm{Image}}
\newcommand{\kker}{\mathrm{Ker}}
\newcommand{\ha}{\hat a}
\newcommand{\cS}{\mathcal S}
\newcommand{\mink}{\mathbb M}
\newcommand{\Bkg}{\Box_{\mathrm{KG}}}

\newcommand{\cK}{\mathcal K}
\newcommand{\cO}{\mathcal O} 
\newcommand{\intm}{\int_{-}}
\newcommand{\intp}{\int_{+}}
\newcommand{\gOjsl}{\Omega^{l}_{js}} 
\newcommand{\djsa}{\Delta^{a}_{js}}
\newcommand{\dja}{\Delta^{a}_{j(n-j)}} 
\newcommand{\UAS}{U^{A/S}}
\newcommand{\FAS}{F^{A/S}}
\newcommand{\GAS}{G^{A/S}}
\newcommand{\HAS}{H^{A/S}}
\newcommand{\QAS}{Q^{A/S}}
\newcommand{\PAS}{P^{A/S}}
\newcommand{\uAS}{u^{A/S}}
\newcommand{\kAS}{k^{A/S}}
\newcommand{\Du}{\Delta u}
\newcommand{\Dv}{\Delta v}

\newcommand{\mathsym}[1]{{}}
\newcommand{\unicode}[1]{{}}
\newcommand{\ka}{k_A}
\newcommand{\ks}{k_S}
\newcommand{\cc}{{\cal C}}
\newcommand{\cm}{{\cal M}}
\newcommand{\id}{{\mathbb I}}

\newcommand{\diam}{\mathcal{D}}
\newcommand{\wsjc}{W^c_{SJ}}
\newcommand{\wsj}{W_{SJ}}
\newcommand{\wmink}{W^{\mathrm{mink}}_0}
\newcommand{\wminkm}{W^{\mathrm{mink}}_m}
\newcommand{\wrind}{W^{\mathrm{rind}}_0}
\newcommand{\wrindm}{W^{\mathrm{rind}}_m}
\newcommand{\wmirr}{W^{\mathrm{mirror}}_0}
\newcommand{\wmirrm}{W^{\mathrm{mirror}}_m}
\newcommand{\wa}{A_{\mathrm{I}}}
\newcommand{\waa}{A_{\mathrm{II}}}
\newcommand{\waaa}{A_{\mathrm{III}}}
\newcommand{\waaaa}{A_{\mathrm{IV}}}
\newcommand{\ws}{S_{\mathrm{I}}}
\newcommand{\wss}{S_{\mathrm{II}}}
\newcommand{\wsss}{S_{\mathrm{III}}}
\newcommand{\wssss}{S_{\mathrm{IV}}}
\newcommand{\emc}{\e_m^{center}}
\newcommand{\emcr}{\e_m^{corner}}

\begin{document}
\title{Sorkin-Johnston vacuum for a massive scalar field in the 2D causal diamond}
\author{Abhishek Mathur\footnote{\it abhishekmathur@rri.res.in} and Sumati Surya\\{\it {\small{ Raman Research Institute, CV
    Raman Ave, Sadashivanagar, Bangalore, 560080, India}}}}
\date{}
\maketitle
\begin{abstract}
We study the massive scalar field Sorkin-Johnston (SJ) Wightman function $\wsj$  restricted to a flat 2D causal diamond
$\diam$ 
of linear dimension $L$.  Our approach is two-pronged. In the first, we solve the central SJ eigenvalue
problem explicitly in the small mass regime, up to order  $(mL)^4$.  This allows us to  formally construct  $\wsj$ up to
this order. Using a combination of analytical and numerical methods, we obtain
expressions for $\wsj$ both in the center and the corner of $\diam$, to leading order. We find that in the center,
$\wsj$ is more like  the massless Minkowski Wightman function $\wmink$  than the massive one $\wminkm$, while in the corner
it corresponds to  that of the massive mirror $\wmirrm$.  In the second part, in order to explore larger masses, we perform  numerical
simulations using a causal set approximated by a flat 2D causal diamond. We find that in
the
center of the diamond the causal set SJ Wightman function $\wsjc$
resembles $\wmink$ for small masses, as in the continuum,  but beyond a critical value $m_c$ it resembles $\wminkm$, as expected. Our calculations suggest
that unlike  $\wminkm$,  $\wsj$ has a well-defined massless limit, which mimics the behavior of the Pauli Jordan function underlying the SJ construction.  In the corner of the diamond,
moreover, $\wsjc$ agrees with  $\wmirrm$ for all masses, and not, as might be expected, with the Rindler
vacuum.  
\end{abstract}

\section{Introduction}

The standard approach to quantum field theory is inherently observer dependent, as is evident from the Unruh effect for
accelerating observers in  Minkowski spacetime.  In Minkowski spacetime, due to its high degree of
symmetry, there is a preferred  family of inertial observers and hence a unique Poincare invariant vacuum. This
Minkowski vacuum is considered the bedrock of quantum field theory, and  its Poincare invariance can be used to
explain many aspects of the theory.

However,  in a generic curved spacetime no such preferred family of
observers exists which can be used to  single out a preferred vacuum state.  This suggests that the state plays a
subsidiary role in the theory. This is the approach taken in algebraic quantum field theory, where a primary role is
played by the algebra of operators. The choice of state is relegated to a choice of representation of this algebra,
which need not be coordinate invariant.  A 
proposal for a unique vacuum state, the \emph{SJ
  vacuum},   for a  free scalar field theory was developed by Sorkin
and Johnston \cite{sorkin,Johnston:2009fr} for a bounded, globally hyperbolic region $M$ of a  spacetime.
The Pauli-Jordan integral operator, defined as
\begin{equation}
i \hD \circ f(X) \equiv \int_M i\Delta (X,X') f(X') \,  dV_{X'} \label{eq:pjop} 
\end{equation} 
is self adjoint in $M$. Here, $\Delta(X,X')$, is the covariantly defined Pauli-Jordan function (which
is the difference in the retarded and advanced Green  functions) and $dV_X$ is the volume element.  The  associated SJ Wightman function $\wsj$ (or two
point function)  is then simply the positive part of $i \hD$.  $\wsj$ can be shown to be the unique vacuum which
  satisfies the following conditions \cite{sorkin, Afshordi:2012ez}
\bea
W(X,X')-W(X',X) &=& i\Delta(X,X')\quad\mathrm{Commutator \; condition} \nonumber\\ 
W(X,X')-W^*(X',X)&=&0\quad\mathrm{Hermiticity} \nonumber\\
\int_M dV_X\, dV_Y \, f^*(X)W(X,Y)f(Y)&\geq&0 \quad \mathrm{Positive \; semidefinite} \nonumber\\
\int_M dV_{X'} \, W(X,X')W(X'',X')&=&0 \quad \mathrm{orthogonal\;support}. \label{eq:wsjconditions}
\eea
$\wsj$ can be explicitly constructed from the spectral decomposition of $i\hD$, where the spectrum of $i \hD$ is given
by the integral eigenvalue equation
\begin{equation}
i \hD \circ u(X) = \lambda \, u(X). 
  \end{equation} 
  This is what we refer to as the ``central eigenvalue  problem''  in the SJ approach.

However the integral form makes it a challenging task  to find solutions even in simple cases. As a result there
are very few
cases in which $\wsj$ has been obtained explicitly. These include the  massless free scalar SJ vacuum in a 2D flat
causal diamond \cite{Afshordi:2012ez,johnston}, a patch of trousers spacetime \cite{Buck:2016ehk} and the ultrastatic slab spacetime \cite{fewster2012}.
In this work, we study
the SJ vacuum for a massive free scalar field in the 2D flat causal diamond $\diam$ of length $2L$,  both in
the continuum and 
on a  causal set $\cc_\diam$ obtained from sprinkling into $\diam$.

In the continuum  we solve the central SJ eigenvalue problem explicitly  in the small mass approximation  keeping terms only up to $\cO(m^4)$,
with $m^4\ll1$ (in dimensionless units, with $L=1$). The eigenfunctions and eigenvalues so obtained 
reduce  to their massless counterparts when $m=0$ \cite{Afshordi:2012ez}.  This allows us to formally construct
$\wsj$ in $\diam$.

As in \cite{Afshordi:2012ez} we consider two  regimes of interest: one in the center of the diamond, and the other at the
corner.  In a small central region $\diam_l$ of size $l$,   we find analytically that $\wsj$ resembles the {\it massless} Minkowski vacuum $\wmink$ up to a
small mass-dependent constant $\emc$, rather than the massive Minkowski vacuum $\wminkm$.  In the corner,  
$\wsj$ resembles the massive mirror vacuum  $\wmirrm$, with the difference depending on a small  mass-dependent constant $\emcr$,
rather than the expected agreement with the massive Rindler vacuum $\wrindm$.   Both $\emc$ and $\emcr$ are the errors
that arise in the approximation of a quantization condition which is a mass dependent transcendental equation,  and
are therefore non-trivial to calculate analytically.  

In order to find  $\e_m^{center}, \e_m^{corner}$, we evaluate 
$\wsj$ numerically using a convergent  truncation $\wsj^t$ of the mode-sum. The
calculations show that $\emc, \emcr$ contribute negligibly to $\wsj$ both in the center and the corner.   
This confirms that for small mass $\wsj$ corresponds to the massless  
Minkowski vacuum. This behavior is unexpected, and suggests that at least  in this small mass approximation $\wsj$ does not
satisfy the expected massive Poincare invariance of the vacuum but rather the massless Poincare invariance. In the corner, again $\e_m^{corner}$ is found to be small, and confirms that  $\wsj$ resembles  $\wmirrm$ rather than
$\wrindm$.   

We then examine the behavior of this truncated $\wsj^{t}$ in a slightly enlarged  region in the center. We find that it continues to  differ  from $\wminkm$, while agreeing with $\wmink$ at least up to $l\sim 0.1$.   In an
enlarged corner region  $\wsj$ there is a marked deviation from $\wmirrm$, but it still does not resemble the Rindler
vacuum.

In the next part of this work we obtain $\wsjc$ numerically  for a causal set $\cc_\diam$ obtained by sprinkling into $\diam$, for a
range of masses.   We
  find that in the small mass regime $\wsjc$ agrees with our analytic calculation of 
  $\wsj$ in the center of the diamond and therefore resembles $\wmink$. This means that it  {\it differs} from $\wminkm$  in the small mass
regime. However, as the mass is increased, there is a cross-over
point at which the massless and massive Minkowski vacuum coincide. This occurs when the mass $m_c\equiv 2\cof \sim
0.924$, where $\cof \sim 0.462$ is the IR cut-off for the massless vacuum calculated in \cite{Afshordi:2012ez}.  For $m \geq
m_c$, $\wsjc$ then tracks the massive Minkowski vacuum instead of the massless Minkowski vacuum.   
In the corner of the diamond, the causal set $\wsjc$ looks like the mirror vacuum and not the Rindler vacuum
for all masses. 

Our calculations suggest that, as in the case of the  de Sitter SJ vacuum studied in \cite{Surya:2018byh},  the massive $\wsj$ has a well defined
$m\rightarrow 0$ limit, unlike  $\wminkm$.  A possible reason for this is that
the SJ vacuum is built from the Green function which is a continuous function of $m$ even as $m\rightarrow 0$. The behavior of
$\wsj$ for $m>0$  is also curious.   For  $\wmink$, $\cof$   sets a scale and 
dominates in the small $m$ regime, while for large $m$, the opposite is true.  At $m_c$, $\wmink$ and $\wminkm$
coincide at small distance scales, so that  $\wsj$ tracks  $\wmink$ for $m<m_c$ and  $\wminkm$ 
for $m>m_c$ in a continuous fashion.

Whether this  unexpected small mass behavior of $\wsj$ is the result of finiteness of $\diam$ or an intrinsic feature
of the 2D SJ vacuum is unclear at the moment.  Further examination of the massive SJ vacuum in different spacetimes
should shed light on these questions.   The mass dependent behavior in the 2D causal diamond echoes that in 4d de Sitter
spacetime \cite{Surya:2018byh}.  For de Sitter spacetime  it is known that there is no massless de Sitter invariant
vacuum, and that the Mottola-Allen vacua do not have an $m \rightarrow 0$ limit. However, for a causal set
that is approximated by de Sitter spacetime $\wsj^c$ seems to behave very differently, and in  particular,   does have a
well defined $m \rightarrow 0 $ limit.   Understanding how these differences in behavior between the SJ and the standard
vacua manifest themselves in the conditions Eqn (\ref{eq:wsjconditions}) should shed some light. However this is beyond
the scope of the present work.  

We begin in Sec.~\ref{sec:sj} with a short introduction to the SJ approach to quantum field theory for free scalar
field in a bounded globally hyperbolic spacetime.  In Sec.~\ref{sec:soln} we set up the SJ eigenvalue problem for the massive
scalar field in $\diam$ and find the SJ spectrum in the small mass limit to $\cO(m^4)$.  Sec.~\ref{sec:wightmann}
contains the analytic and numerical calculations of $\wsj$ in different regions of $\diam$. In Sec.~\ref{sec:causet} we show the results of simulations of the causal set SJ vacuum $\wsjc$ for a range of
masses.  We then compare  $\wsjc$ with the analytical calculation $\wsj$ in the small mass regime, as well as  with the standard
vacua  in the large mass regime, both in the center and the corner of the
diamond for small and large values of $m$. We end with a brief discussion of our results in Section
\ref{sec:conclusions}.  Appendixes \ref{expressions}, \ref{app:speigel} and \ref{sec:wight-app} contain the details of many of the calculations. In Appendix \ref{sec:rindler} we present a trick to get the 2D Rindler vacuum from the SJ prescription.

\section{The SJ prescription}
\label{sec:sj} 

For a free scalar field $\hp$, with
Gaussian vacuum state $\vac$, the two point function 
\be
W(X,X')\equiv\left<0\left|\hp(X)\hp(X')\right|0\right> \label{vev}
\ee
contains all the information about the theory. In the standard route to quantization $\vac$ is itself defined using an 
observer dependent mode
decomposition of $\hp(x)$. The absence of a preferred class
of observers for a  general curved spacetime $(M,g)$ means that this mode decomposition does not lead to a preferred choice
of $\vac$ and thence $W(X,X')$. 

The SJ prescription provides an observer independent mode decomposition  $\hp$ defined in a compact globally hyperbolic
 spacetime region \cite{sorkin,Johnston:2009fr,Afshordi:2012ez,Buck:2016ehk,fewster2012,aas,Brum:2013bia,Avilan:2014hra}.  Instead of an equal time 
commutation relation, it uses the covariant 
Peierls bracket  
\be
\left[\hp(X),\hp(X')\right]=i\Delta(X,X') \label{peierls}, 
\ee
where the Pauli Jordan function is given by 
\be
i\Delta(X,X')=i \left( G_R(X,X')-G_A(X,X') \right)
\ee
and $G_R(X,X'),G_A(X,X')$ are the retarded and advanced Green functions respectively. $i\Delta(X,X')$ is therefore
imaginary and antisymmetric.  

The Pauli-Jordan operator is an  integral operator,  Eqn (\ref{eq:pjop})  on the space $\cF(M,g)$ of bounded functions
in $(M,g)$ (see \cite{wald}), whose $\mL^2$ inner product  is 
\be
\left(f,g\right)\equiv\int_M dV_X f^*(X)g(X). 
\ee
 $i\hD$ is therefore self adjoint on $\cF(M,g)$. 
The eigenvalues of $i\hD$ are therefore real and come in positive and negative pairs 
\bea
i\hD \circ u_k=\lambda_ku_k \nonumber \\
i\hD \circ u_k^*=-\lambda_ku_k^*,
\label{eq:sjmodes} 
\eea
where $u_k \in \im(i\hD)$. The normalized modes $u_k^{SJ}=\sqrt{\lambda_k} u_k$ are referred to as the \emph{SJ modes}. Since the $\{u_k\}$ are a complete orthonormal basis in $\im(i\hD)$, they give the following spectral decomposition
\be
i\Delta(X,X')=\sum_k\lambda_k\left(u_k(X)u_k^*(X')-u_k^*(X)u_k(X')\right).
\label{eq:spec}
\ee
It can be shown that \cite{fewster2012,wald,sorkin17} 
\be
\im(i\hD)=\ker(\nabla_\mu\nabla^\mu-m^2)\label{ker-im}.
\ee
Thus the SJ modes are also solutions of the KG equation. 

The SJ proposal is to obtain $\wsj$ from $i\Delta$, without reference to preferred observers. Using the properties
of $\wsj$ given in  Eqn.~(\ref{eq:wsjconditions}), it follows that
\be
\wsj=\mathrm{Pos}(i\hD)\Longleftrightarrow\wsj=\frac{1}{2}\left(i\hD+\sqrt{-\hD^2}\right)\Longleftrightarrow\wsj(X,X')=\sum_k\lambda_ku_k(X)u_k^*(X').
\ee
The SJ mode expansion of $\hp(X)$ is then
\begin{equation}
  \hp(X)=\sum_k\sqrt{\lambda_k}\left(\ha_ku_k(X)+\ha_k^\dagger u_k^*(X)\right),
\end{equation} 
with the vacuum $\vac_{SJ}$ defined by $\ha_k\vac_{SJ}=0$.
 
In the discussion above, there is an implicit assumption  that $i\hD$ is self-adjoint. This is guaranteed when $(M,g)$ is bounded, but not so
when this condition is lifted. To rigorously show that $\vac_{SJ}$ reduces to the various known vacua, including the
Minkowski vacuum, it is important to take this into account. In \cite{aas} a mode comparison argument was used to show
that the SJ vacuum in Minkowski spacetime is the Minkowski vacuum.  However, as argued in \cite{Surya:2018byh} a mode comparison
may not indicate the equivalence of vacua.

A more careful approach was adopted  in \cite{Afshordi:2012ez} where the massless SJ vacuum was calculated explicitly in a 2D causal diamond 
$\diam$ of length $2L$.  Evaluating $\wsj$ in the center of the diamond, i.e., with  $|\vec x -\vec x'| <<
L $  and $|\vec x|, |\vec x'| <<L $ it  was shown that $\vac_{SJ} \sim \vac_{mink}$.  Thus, away
from the boundaries, the massless SJ vacuum is indeed the Minkowski vacuum. The goal of this work is to perform a similar
calculation for the massive case in the finite diamond, in which the SJ construction  is well defined. 
  
Important to this calculation is not only the boundedness of $i\hD$ which ensures self-adjointness, but also its
Hilbert-Schmidt property using which the completeness of its  eigenfunctions can be checked.
In higher even dimensions, the massless retarded  Green's function has  $\delta$ functions. While  $i\hD$ 
is self-adjoint for bounded spacetime region, it is not Hilbert Schmidt.

\section{The Spectrum of the Pauli Jordan Function: The small mass limit}
\label{sec:soln}

As we have stated earlier, the SJ modes Eqn.~(\ref{eq:sjmodes}) are also solutions of the KG equation. A natural starting  point for
constructing these modes is therefore to start with a complete set of solutions $\{ s_k\}$ in the space $\cS =
\ker(\Bkg)$ where $\Bkg\equiv \Box-m^2$,  and to find the action of $i \hD$ on this set. In light-cone coordinates the 2D Klein Gordon equation in Minkowski spacetime takes the simple form 
\be
\Bkg(u,v) \phi(u,v) \equiv \left(2\partial_u\partial_v+m^2\right)\phi(u,v)=0.  \label{kgeqn}
\ee
where 
\be
u=\frac{1}{\sqrt{2}}\left(t+x\right), \;\;\;\ \;\;\;v=\frac{1}{\sqrt{2}}\left(t-x\right). 
\ee
Thus, for $m=0$ any differentiable  function $\psi(u)$ or $\xi(v)$ is in $\ker(\Bkg (u,v))$. 

One can generate a larger class of solutions starting  from a given differentiable function $\psi(u)$. The infinite sum 
\be
\phi(u,v)\equiv \sum_{n=0}^\infty\frac{(-1)^nm^{2n}}{2^nn!}v^n\int^n\psi(u),
\ee
with  $\int^n\psi(u)\equiv\int du\int du\dots\int du\psi(u)$, can be seen to belong to $\ker(\Bkg)$. Similarly one can generate solutions starting with a differentiable function $\xi(v)$.  Different choices of $\psi(u),\xi(v)$ gives
different $\phi(u,v)$. 

From the { Weierstrass theorem}, we know that any continuous function $\psi(u)$ in a bounded interval in $u$ can be
written as $\psi(u)=\sum_{n}a_nu^n$ for some $a_n's$. Hence  a natural class of solutions is generated by  $\psi(u)=u^l$,
\be
Z_l(u,v) \equiv
\sum_{n=0}^\infty\frac{(-1)^nm^{2n}l!}{2^nn!(n+l)!}u^{n+l}v^n=\frac{2^{l/2}l!}{m^l}\left(\frac{u}{v}\right)^{l/2}J_l\left(m\sqrt{2uv}\right)
\label{eq:serieszl},
\ee
for $l$ a whole number. Thus the SJ modes,  can in general be written as a sum over $Z_l(u,v)$
and $Z_l(v,u)$ for an appropriate set of $l$ values.  
Since plane waves are an important class of solutions, we note that starting from a function $\psi(u)=e^{au}$ for some
constant $a$ the plane wave solutions 
\be
U_a(u,v) \equiv \sum_{n=0}^\infty\frac{(-1)^nv^nm^{2n}}{2^nn!a^n}e^{au}=e^{au-\frac{m^2}{2a}v} \label{eq:seriesua}
\ee
and similarly, $U_a(v,u)$, can be obtained.

Before we proceed with the construction of the SJ modes, it will be useful to look at its following property.
\vskip 0.1in
\begin{claim}
In $\diam$ the SJ modes can be arranged into a complete set of eigenfunctions, each of which is either symmetric or
antisymmetric under the interchange of $u$ and $v$ coordinates.
\end{claim}
\begin{proof}
Let $u_k$ be an eigenfunction of $i\hD$ with eigenvalue $\lambda_k\neq 0$ i.e.
\be
i\hD\circ u_k=\lambda_ku_k \label{eqn:proof1}.
\ee
Define an operator $\hD'$ with integral kernel $\Delta'(u,v;u',v')=\Delta(v,u;v',u')$ and let $v_k$ such that $v_k(u,v)=u_k(v,u)$. Interchanging $u$ and $v$ since $u,v\in[-L,L]$, Eqn.~(\ref{eqn:proof1}) can be rewritten as
\be
i\hD'\circ v_k=\lambda_kv_k \label{eqn:proof2}.
\ee
Since $\Delta(u,v;u',v')$ is symmetric under $\{u,u'\}\leftrightarrow \{v,v'\}$, this implies that
\be
i\hD\circ v_k=i\hD'\circ v_k=\lambda_kv_k \label{eqn:proof3}.
\ee
Therefore $v_k$ is also an eigenfunction of $i\hD$ with same eigenvalue $\lambda_k$. This means that, the symmetric combination $u^S_k(u,v)=u_k(u,v)+u_k(v,u)$ and the antisymmetric combination $u^A_k(u,v)=u_k(u,v)-u_k(v,u)$ are also eigenfunctions of $i\hD$ with eigenvalue $\lambda_k$.
\end{proof}

In $\mink^2$ for $m=0$ the natural choice of solutions is the set of plane wave modes $\{e^{iku}, e^{ikv}\}$. However, in the
finite causal diamond, the constant function is also a solution. The explicit form of the corresponding SJ modes are
given  in Johnston's thesis \cite{johnston}. There are two
sets of eigenfunctions. The first set found by Johnston are the $f_k=e^{iku}-e^{ikv}$ modes with $k =n \pi/L$ and are
antisymmetric with respect to $u\leftrightarrow v$. The
second set  $g_k=e^{iku}+e^{ikv}-2\cos(kL)$, were found by Sorkin and satisfy the more complicated quantization condition $\tan(kL)=2kL$. These are symmetric with respect to $u\leftrightarrow v$.  The eigenvalues for each set are $\pm L/k$.

We now proceed to set up the calculation for the central SJ eigenvalue problem.  We will find it useful to work with the
dimensionless quantities.
\be
  m L \rightarrow m,\;kL\rightarrow k,\;\frac{u}{L}\rightarrow u,\;\frac{v}{L}\rightarrow v,\;\frac{u'}{L}\rightarrow u',\;\frac{v'}{L}\rightarrow v'.
\ee
The  massive Pauli Jordan function in $\mink^2$ is  
\be
i\Delta(u,v;u',v')=-\frac{i}{2}J_0\left(m\sqrt{2\Du\Dv}\right)\left(\theta (\Du)+\theta(\Dv)-1\right) \label{idelta}
\ee
where $\Du=u-u', \Dv=v-v'$ and $\theta(x)$ is the Heaviside function.  
The SJ modes are thus given by (Eqn.~\ref{eq:sjmodes}) 
\be
-\frac{iL^2}{2}\int_{-1}^1du'dv'J_0\left(m\sqrt{2\Du\Dv}\right)\biggl(\theta(\Du)+\theta(\Dv)-1\biggr) u_k(u',v')=\lambda_ku_k(u,v).
\ee
We will find it useful to make the change of variables $\Du=p, \Dv=q$ so that the above expression becomes  
\be
\frac{iL^2}{2}\left(\intm dp dq -\intp dp dq \right)  J_0\left(m\sqrt{2pq}\right)u_k(u-p,v-q)=\lambda_ku_k(u,v),  
\ee
where we have  used the short-hand $\intm dp dq \equiv \int_0^{u-1}dp\int_0^{v-1}dq$ and $\intp dp dq \equiv
\int_0^{u+1}dp\int_0^{v+1}dq$. 
Our strategy is to begin with the action of $i\hD$ on the   symmetric and antisymmetric combinations of the
$Z_l(u,v)$ and $U_a(u,v)$ solutions defined above,
\begin{eqnarray} 
  U^A_{a}(u,v)\equiv U_{a}(u,v)-U_{a}(v,u), &\;\;\;&  U^S_{a}(u,v)\equiv U_{a}(u,v)+U_{a}(v,u),\nonumber \\ 
        Z^A_{l}(u,v)\equiv Z_{l}(u,v)-Z_{l}(v,u), &\;\;\;&  Z^S_{l}(u,v)\equiv Z_{l}(u,v)+Z_{l}(v,u).     \label{eq:symasym}                                            
\end{eqnarray} 
so that the  general form for the two sets $u^{A/S}$ of SJ modes is given by 
\begin{equation} 
 u^{A/S}_{\vec{a},\vec{l}}(u,v) \equiv  \sum_{a \in \vec{a}} \alpha^{A/S}_aU^{A/S}_a(u,v) + \sum_{l \in \vec{l}}
 \beta^{A/S}_l Z^{A/S}_l(u,v).
\end{equation}
Here $\vec{a}, \vec{l}$ denote  set of values for  $a$ and $l$ which satisfy  quantization conditions. Of course each $U_a(u,v)$ is itself an infinite sum over  $Z_l(u,v)$, but
we nevertheless consider it separately, taking our cue from the massless calculation. 

The expressions 
\begin{eqnarray} 
i\hD\circ U_a(u,v) &=&\frac{iL^2}{2}\left( \intm dp dq-\intp dp dq\right) J_0\left(m\sqrt{2pq}\right)
                          U_a^*(p,q)U_a(u,v), \nonumber \\
i\hD\circ Z_l(u,v) &=&  \frac{iL^2}{2}\left( \intm dp dq-\intp dp dq\right) J_0\left(m\sqrt{2pq}\right)
                         Z_l(u-p,v-q) \label{integration2} 
\end{eqnarray} 
are in general not easy to evaluate and  subsequently manipulate in order to obtain the SJ modes.  We instead begin by
looking for solutions order by order in $m^{2}$ assuming that for some $n$, $m^{2n}<<1$.\footnote{The series expansion
  of $U^{A/S}_{ik}$ in the SJ modes for small $m$ can be truncated to a finite order of $m^2$ if and only if $k$ is of
  the order of unity or higher. However, this is the case for small $m$, since small $k$ corresponds to wavelengths much
  larger than the size of the diamond.}
We use the series form of $Z_l(u,v)$ in Eqn.~(\ref{eq:serieszl}) and $U_a(u,v)$ in Eqn.~(\ref{eq:seriesua})
  as well as 
\be
J_0\left(m\sqrt{2pq}\right) =\sum_{n=0}^\infty\frac{(-1)^nm^{2n}}{2^n(n!)^2}p^nq^n. 
\label{eq:seriesbess}
\ee
As we will show, for $n=4$, we find that, to $\cO(m^4)$ the two families of eigenfunctions, antisymmetric and symmetric are 

{\bf Antisymmetric:}
\be
u^A_k(u,v)=\left[U^A_{ik}(u,v)-\cos(k)\left(\left(\frac{im^2}{2k}-\frac{im^4(6+k^2)}{24k^3}\right)Z^A_1(u,v)-\frac{m^4}{4k^2}Z^A_2(u,v)\right)\right]+\cO(m^6), \label{eq:asym}
\ee
with eigenvalue $-\frac{L^2}{k}$ with $k \in \cK_A$ satisfying the quantization condition 
\be
\sin(k)=\left(\frac{m^2}{k}+\frac{m^4}{12k}\left(1-\frac{3}{k^2}\right)\right)\cos(k)+\cO(m^6).\label{eq:qcondA}
\ee
Solving for $k$, order by order in $m^2$ up to $\cO(m^4)$, as shown in Sec. \ref{sec:compl},  gives $k=\ka(n)$, where
\be
\ka(n)\equiv n\pi+\frac{m^2}{n\pi}+m^4\left(\frac{1}{12n\pi}-\frac{5}{4n^3\pi^3}\right)+\cO(m^6), \label{eq:qcondka}
\ee
where $n\in{\mathbb Z}$ and $n \neq 0$.

{\bf Symmetric:}
\bea
u^S_k(u,v)&=&\left[U^S_{ik}(u,v)-\cos(k)\left(\left(1+\frac{m^2}{2}-\frac{m^4}{8k^2}(2-9k^2)\right)Z^S_0(u,v)\right.\right.\nonumber\\
&&\left.\left.+\left(\frac{3im^2}{2k}-\frac{im^4}{24k^3}(6-31k^2)\right)Z^S_1(u,v)-\frac{m^4}{8k^2}(4-k^2)Z^S_2(u,v)\right)\right]+\cO(m^6), \label{eq:sym}
\eea
with eigenvalue $-\frac{L^2}{k}$, where $k \in \cK_S$ satisfies
\be
\sin(k)=\left(2k-\frac{m^2}{k}(1-2k^2)+\frac{m^4}{12k^3}(3-29k^2+28k^4)\right)\cos(k)+\cO(m^6). \label{eq:qcondS}
\ee
Solving for $k$, order by order in $m^2$ up to $\cO(m^4)$, as shown in Sec. \ref{sec:compl}, gives $k=\ks(k_0)$, where
\be
\ks(k_0)\equiv k_0+m^2\frac{1-2{k_0}^2}{k_0(1-4{k_0}^2)}+m^4\frac{(3-4{k_0}^2)(-5+35{k_0}^2-40{k_0}^4+16{k_0}^6)}{12{k_0}^3(1-4{k_0}^2)^3}+\cO(m^6) \label{eq:qcondks},
\ee
where $k_0$ are the solutions of $\sin(k)=2k\cos(k)$.

We plot these eigenvalues in Fig.~\ref{fig:eigenvalue} for m=0, 0.2 and 0.4. In the expressions for the eigenfunctions,  Eqns (\ref{eq:asym}) and (\ref{eq:sym}), it is to be noted that we have kept 
$U^{A/S}_{ik}$ and $Z^{A/S}_l$ as they are, rather than use their expansion to $\cO(m^4)$.   The reason for this is to
remind ourselves that they are solutions of the Klein Gordon equation. Note that in  Eqn.~(\ref{eq:asym}) and
Eqn.~(\ref{eq:sym}), we keep terms only up to $\cO(m^4)$ within the square bracket. In Sec. \ref{sec:compl} we show that
these form a complete set of orthonormal modes.

Here we have moved away from the $f_k$ and $g_k$ notation of \cite{Afshordi:2012ez,johnston} to $u^A_k$ and $u^S_k$ for
the antisymmetric and symmetric SJ modes respectively.

\begin{figure}[h]
\centerline{\begin{tabular}{ccc}
\includegraphics[height=5cm]{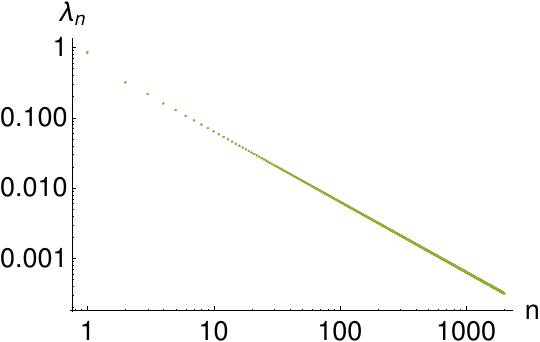}&
\includegraphics[height=5cm]{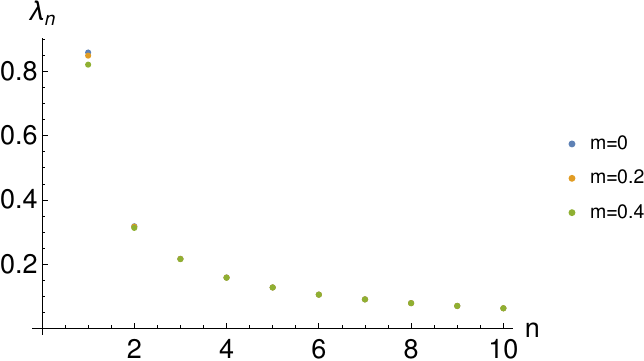}\\
(a)&(b)
\end{tabular}}
\caption{(a):A log-log plot of the SJ eigenvalues $\lambda_n$ vs $n$ for $m=0, 0.2$ and $0.4$, (b): a plot of $\lambda_n$ vs $n$ for small $n$. As one can see, the eigenvalues for $m=0.2$ and $0.4$ are barely distinguishable from $m=0$, except for the very smallest $n$ values.}
\label{fig:eigenvalue}
\end{figure}

\subsection{Details of the calculations of SJ modes}
We now show the calculation in broad strokes below, leaving some of the details to the Appendix \ref{expressions}. We begin by  reviewing the massless case. Here  $Z_l(u,v)$ reduces to $u^l$ and $U_a(u,v)$ to $e^{au}$.

Operating $i \hD$ on $u^l$ or $v^l$  we find that
\begin{eqnarray} 
  i\hD_{m=0}\circ u^l&=&\frac{iL^2}{2(l+1)}\left(\left(1+(-1)^{l+1}\right)-v\left(1-(-1)^{l+1}\right)-2u^{l+1}\right), \nonumber \\
  i\hD_{m=0}\circ v^l&=&\frac{iL^2}{2(l+1)}\left(\left(1+(-1)^{l+1}\right)-u\left(1-(-1)^{l+1}\right)-2v^{l+1}\right), 
\end{eqnarray}
while on the plane wave modes  
\bea
i\hD_{m=0}\circ e^{iku}&=&
-\frac{L^2}{k}\left(e^{iku}-\cos(k)+iv\sin(k)\right), \nonumber \\
i\hD_{m=0}\circ e^{ikv}&=&
-\frac{L^2}{k}\left(e^{ikv}-\cos(k)+iu\sin(k)\right).
\eea 

Here, $k$ takes on all values including $k=0$, which is the constant solution. From the antisymmetric combination
\be
i\hD_{m=0}\circ\left(e^{iku}-e^{ikv}\right)=-\frac{L^2}{k}\left(e^{iku}-e^{ikv}-i\sin(k)(u-v)\right), 
\ee
we find the first set of massless eigenfunctions
\be
u_k^{A(0)}(u,v)\equiv e^{iku}-e^{ikv}
\ee
with  $k \in \cK_f$ satisfying the quantization condition 
\be
\sin(k)=0\;\text{or}\;k=n\pi.  \label{f condition m0}
\ee
with eigenvalues $-\frac{L^2}{k}$. The symmetric  combination
on the other hand gives 
\begin{equation} 
i\hD_{m=0}\circ\left(e^{iku}+e^{ikv}\right)=
-\frac{L^2}{k}\left(e^{iku}+e^{ikv}-2\cos(k)\right)-\frac{iL^2}{k}\sin(k)(u+v).
\end{equation} 
Since the  symmetric eigenfunction can include a  constant piece and noting that 
\begin{equation}
 \hD_{m=0}\circ c =  -icL^2(u+v), 
  \end{equation} 
we find the second set of eigenfunctions 
\be
u_k^{S(0)}(u,v)\equiv e^{iku}+e^{ikv}-2\cos(k)
\ee
with eigenvalue $-\frac{L^2}{k}$, where $k \in \cK_g$ satisfies
\be
\sin(k)=2k\cos(k). \label{g condition m0}
\ee
$\{u_k^{A(0)}\}$ and $\{u_k^{S(0)}\}$ together form a complete set of eigenfunctions of $i\Delta$ as can be shown by \cite{johnston}.

This sets the stage for the calculation of the massive SJ modes. We begin by again looking the action of $i\hD$ on the
solutions $Z_l(u,v)$ and $U_a(u,v)$,  
\begin{eqnarray} 
i\hD\circ Z_l(u,v)&=&\frac{iL^2}{2}\sum_{j,s=0}^\infty\frac{(-1)^{j+s}m^{2(j+s)}l!}{2^{l+s}(j!)^2s!(s+l)!} \gOjsl ,\label{eq:ideltaz}   \\ 
i\hD\circ U_{a}(u,v) &=&\frac{iL^2}{2} U_a(u,v) \sum_{j,s=0}^\infty\frac{(-1)^j
                            m^{2(j+s)}}{2^{j+s}(j!)^2s!a^s}\djsa(u,v),
\label{eq:genzluaexp}                             
\end{eqnarray}
where
\begin{eqnarray}
\gOjsl(u,v) &\equiv & \biggl(\intm dp\, dq - \intp dp\, dq\biggr)p^jq^j(u-p)^{l+s}(v-q)^s, \nonumber \\  
\djsa(u,v) & \equiv &\biggl(\intm dp\, dq - \intp dp\, dq\biggr)p^jq^{j+s}e^{-a p}. \label{deltajs}
\end{eqnarray}
It is useful to re-express Eqn.~(\ref{eq:genzluaexp})  as 
\begin{equation} 
i\hD\circ U_{a}(u,v) = \frac{iL^2}{2} U_a(u,v) \sum_{n=0}^\infty m^{2n} A_{a,n}(u,v), \label{eq:ideltau} 
\end{equation}
where
\begin{equation}
A_{a,n}(u,v)\equiv \sum_{j=0}^{n}\frac{(-1)^j}{2^{n}(j!)^2(n-j)!a^{(n-j)}}\dja(u,v). 
\end{equation}
This gives
\begin{equation}
  i\hD\circ U_{a}(u,v) = -\frac{iL^2}{a}U_{a}(u,v) - \frac{iL^2}{a} \sum_{n=0}^\infty m^{2n} \cF_{a,n}(u,v),
  \label{eq:genexp} 
  \end{equation} 
  where
  \begin{equation}
\cF_{a,n}(u,v) \equiv F_{a,n}(u,v)\sinh(a)+G_{a,n}(u,v)\cosh(a) ,
    \end{equation} 
with
\bea
F_{a,n}(u,v) &\equiv & \sum_{s=0}^n\sum_{j=0}^s\sum_{l=0}^j \frac{(-1)^{n-s+j}v^{n-s}\left((u+1)^{j-l}(v+1)^{s+1}+(u-1)^{j-l}(v-1)^{s+1}\right)}{2^{n+1}a^{n-j+l}(n-s)!j!(s-j)!(j-l)!(s+1)}, \nonumber\\
G_{a,n}(u,v) &\equiv & \sum_{s=0}^n\sum_{j=0}^s\sum_{l=0}^j \frac{(-1)^{n-s+j}v^{n-s}\left((u-1)^{j-l}(v-1)^{s+1}-(u+1)^{j-l}(v+1)^{s+1}\right)}{2^{n+1}a^{n-j+l}(n-s)!j!(s-j)!(j-l)!(s+1)}. \label{f and g}
\eea

Our first guess, inspired by the massless calculation, is that in order to  find the SJ modes,  we will need the antisymmetrized and symmetrized versions
of Eqns (\ref{eq:ideltaz}) and (\ref{eq:ideltau}), which we denote
by $A/S$.  As noted above, and is evident from Eqn.~(\ref{eq:genexp}), in order to obtain the SJ modes,
$\UAS_a(u,v)$  must be supplemented by a function $\HAS_a(u,v)$ made from the $Z_l(u,v)$.

Taking our cue from the massless case,  let us assume that such a function
exists, i.e., 
\be
i\hD\circ\left(U^{A/S}_{a}(u,v)+\HAS_a(u,v)\right)=-\frac{i L^2}{a}\left(U^{A/S}_{a}(u,v)+\HAS_a(u,v)\right) \label{eq:u,h},
\ee
where $k$ satisfies an appropriate quantization condition $\cK^{A/S}$.  Then, from Eqn.~(\ref{eq:genexp}) $\HAS_a(u,v)$
must satisfy 
\begin{equation}
i\hD\circ \HAS_a(u,v)+\frac{i L^2}{a}\HAS_a(u,v)-\frac{iL^2}{a} \sum_{n=0}^\infty m^{2n} \cF^{A/S}_{a,n}(u,v)=0 .
  \end{equation} 

Up to now the discussion has been general. If the expressions above could be calculated in closed form,  then one would be able to
solve the SJ mode problem for any mass $m$.  It is unclear how to proceed to do this, except order by order in $m^2$.

We now demonstrate this explicitly up to $\cO(m^4)$. We begin by taking $a=ik$ and writing Eqn.~(\ref{eq:genexp}) as
\begin{equation} 
i\hD\circ \UAS_{ik}(u,v) \approx -\frac{L^2}{k}\UAS_{ik}(u,v)-\frac{L^2}{k}\left( i\sin(k) \sum_{n=0}^\infty m^{2n} \FAS_{ik,n}(u,v)
     + \cos(k) \sum_{n=0}^\infty m^{2n} \GAS_{ik,n}(u,v)\right),\label{ideltaua}
\end{equation}
where the expressions for $F_{ik,n}(u,v)$ and $G_{ik,n}(u,v)$ for different $n$ have been calculated in  Appendix \ref{expressions}. 
The function $H^{A/S}_k(u,v)$ must therefore satisfy
\begin{equation} 
i\hD\circ H^{A/S}_{ik}(u,v)+\frac{L^2}{k}\left(H^{A/S}_{ik}(u,v)-i\sin(k) \sum_{n=0}^\infty m^{2n} \FAS_{ik,n}(u,v)-\cos(k) \sum_{n=0}^\infty m^{2n} \GAS_{ik,n}(u,v)\right)=0. \label{eq:has}
\end{equation} 

From the result for the massless case, we expect the quantization condition for $k$ to be of the general form
\be
\sin(k)=\cos(k) \sum_{n=0}^\infty m^{2n} \QAS_n(k), \label{f condition}
\ee
with $Q^A_0(k)=0$ and $Q^S_0(k)=2k$. Inserting this into Eqn.~(\ref{eq:has}) gives
\begin{equation} 
  i\hD\circ \HAS_{ik}(u,v)+\frac{L^2}{k}\HAS_{ik}(u,v)-\frac{L^2}{k}\cos(k) \left( \sum_{n=0}^\infty m^{2n} \PAS_n(u,v) \right)
  =0,
  \label{eq:hask}
  \end{equation} 
  where
  \begin{equation}
    \PAS_n(u,v) \equiv \GAS_n(u,v) + i\sum_{j=0}^n \QAS_j(k) \FAS_{n-j}(u,v).  \label{eq:pas}
    \end{equation} 
The challenge is therefore to obtain the explicit form for these expressions. Finding a general expression in this
 manner is very
 challenging, but we will now show that it can be found to $\cO(m^4)$. 

Since the $\HAS_a(u,v)$ must be constructed from the $Z_l(u,v)$, we are interested in the action of $i\hD$ on $Z_l(u,v)$
up to $\cO(m^4)$ i.e.,
\begin{equation}
i\hD\circ Z_l(u,v)=  \frac{iL^2}{2}\sum_{j,s, j+s \leq 2 } \frac{(-1)^{j+s}m^{2(j+s)}l!}{2^{l+s}(j!)^2s!(s+l)!} \gOjsl+\cO(m^6).
  \end{equation}   
We calculate this expression for  $l=0,1,2$, up to $\cO(m^4)$  in the Appendix \ref{expressions}. Using the expression of
$P_n^A(u,v)$ given in Appendix \ref{expressions}, we find that up to  $\cO(m^4)$ the antisymmetric version of  Eqn.~(\ref{eq:hask}) reduces to  
\be
\left(i\hD+\frac{L^2}{k}\right)\circ\left(H^A_{ik}(u,v)+\cos(k)\left(\left(\frac{im^2}{2k}-\frac{im^4(6+k^2)}{24k^3}\right)Z^A_1(u,v)-\frac{m^4}{4k^2}Z^A_2(u,v)\right)\right)\approx 0. 
\ee
Therefore  
\be
u^A_k(u,v)=U^A_{ik}(u,v)-\cos(k)\left(\left(\frac{im^2}{2k}-\frac{im^4(6+k^2)}{24k^3}\right)Z^A_1(u,v)-\frac{m^4}{4k^2}Z^A_2(u,v)\right)+\cO(m^6),
\ee
with eigenvalue $-\frac{L^2}{k}$ with $k \in \cK_A$ satisfying the quantization condition 
\be
\sin(k)=\left(\frac{m^2}{k}+\frac{m^4}{12k}\left(1-\frac{3}{k^2}\right)\right)\cos(k)+\cO(m^6).
\ee
Similarly using the expression of
$P_n^S(u,v)$ given in Appendix \ref{expressions} and after more painstaking algebra, 
we find that Eqn.~(\ref{eq:hask}) can  be written as

\bea
\left(i\hD+\frac{L^2}{k}\right)\circ\left(H^S_{ik}(u,v)+\cos(k)\left(\left(1+\frac{m^2}{2}-\frac{m^4}{8k^2}(2-9k^2)\right)Z^S_0(u,v)\right.\right.\nonumber\\
\left.\left.+\left(\frac{3im^2}{2k}-\frac{im^4}{24k^3}(6-31k^2)\right)Z^S_1(u,v)-\frac{m^4}{8k^2}(4-k^2)Z^S_2(u,v)\right)\right)
\approx 0 .
\eea
Therefore the symmetric eigenfunction is 
\bea
u^S_k(u,v)&=&U^S_{ik}(u,v)-\cos(k)\left(\left(1+\frac{m^2}{2}-\frac{m^4}{8k^2}(2-9k^2)\right)Z^S_0(u,v)\right.\nonumber\\
&&\left.+\left(\frac{3im^2}{2k}-\frac{im^4}{24k^3}(6-31k^2)\right)Z^S_1(u,v)-\frac{m^4}{8k^2}(4-k^2)Z^S_2(u,v)\right)+\cO(m^6),
\eea
with eigenvalue $-\frac{L^2}{k}$, where $k \in \cK_S$ satisfies
\be
\sin(k)=\left(2k-\frac{m^2}{k}(1-2k^2)+\frac{m^4}{12k^3}(3-29k^2+28k^4)\right)\cos(k)+\cO(m^6).
\ee

Unfortunately, the structure of neither the coefficients in $\uAS_k$  nor the quantization condition are enough to
suggest a generalization to all orders. One could of course proceed to the next order $\cO(m^6)$ but the calculation  gets
prohibitively more complex.  

\subsection{Completeness of the eigenfunctions}\label{sec:compl}

We now show that the eigenfunctions $\{ u_k^A| k \in \cK_A\}$ and  $\{ u_k^S| k \in \cK_S\}$  form a complete set of
eigenfunctions of $i\Delta$. If this is the case, then we can decompose $i\Delta$ as  
\be
i\Delta(u,v;u',v')=\sum_{k\in \cK_A}  -\frac{L^2}{k}  u^A_k(u,v){u^A_k}^*(u',v') + \sum_{k\in \cK_S}  -\frac{L^2}{k}
u^S_k(u,v){u^S_k}^*(u',v')+\cO(m^6), 
\ee
which implies that 
\be
\int_{S} du\, dv\,du'\, dv' |\Delta(u,v;u',v')|^2=\sum_{k\in \cK_A} \biggl(\frac{L^2}{k}\biggr) ^2 + \sum_{k\in \cK_S} \biggl(\frac{L^2
}{k}\biggr) ^2+\cO(m^6).\label{eq:trace}
\ee 

To  $\cO(m^4)$  the LHS  of Eqn.~(\ref{eq:trace}) reduces to 
\bea
&& \frac{L^4}{4}\int_{-1}^1dudv\left(\intm dp \, dq+\intp dp\, dq\right) J_0^2\left(m\sqrt{2pq}\right) \nonumber \\ 
&=&
\frac{L^4}{4}\int_{-1}^1dudv\left( \intm dp \, dq+\intp dp\, dq\right)\left(1-m^2pq+\frac{3}{8}m^4p^2q^2\right)+\cO(m^6)\nonumber\\
&=&2L^4\left(1-\frac{4}{9}m^2+\frac{1}{6}m^4\right)+\cO(m^6). 
\eea
For the RHS  $k\in \cK_{A/S}$, we make use of the expansion $\kAS \approx \kAS_0+m^2\kAS_1+m^4\kAS_2$.
For the antisymmetric quantization condition Eqn.~(\ref{eq:qcondA}) since $k_0^A=n \pi$ this gives, up to $\cO(m^4)$  
\be
m^2k_1^A+m^4k_2^A=\frac{m^2}{k_0^A}\left(1-m^2\frac{k_1^A}{k_0^A}\right)-\frac{m^4}{4{k_0^A}^3}+\frac{m^4}{12k_0^A}+\cO(m^6).
\ee
Solving the above equation for different orders of $m^2$, we get
\bea
k_1^A&=&\frac{1}{n\pi}, \\
k_2^A&=&\frac{1}{12n\pi}-\frac{5}{4n^3\pi^3},
\eea
so that 
\bea
\sum_{k \in \cK_A} L^4 \frac{1}{k^2}
&=&2L^4\sum_{n=1}^\infty\frac{1}{n^2\pi^2}\left(1-2m^2\frac{1}{n^2\pi^2}-m^4\left(\frac{1}{6n^2\pi^2}-\frac{11}{2n^4\pi^4}\right)\right)+\cO(m^6)\nonumber\\
&=&2L^4\left(\frac{1}{6}-\frac{m^2}{45}+\frac{m^4}{252}\right)+\cO(m^6).
\eea
For the symmetric contribution  Eqn.~(\ref{eq:qcondS})   up to $\cO(m^4)$ we have 
\be
\sum_{n=0}^2 m^{2n}K_n(k^S_0,k^S_1,k^S_2)+\cO(m^6)=0,
\ee
where
\bea
K_1(k^S_0,k^S_1,k^S_2)&=&\sin(k^S_0)-2k^S_0\cos(k^S_0), \nonumber\\
K_2(k^S_0,k^S_1,k^S_2)&=&\left(\frac{2{k^S_0}^2-1+k^S_1k^S_0}{k^S_0}\right)\cos(k^S_0)-2k^S_1k^S_0\sin(k^S_0),\nonumber\\
K_3(k^S_0,k^S_1,k^S_2)&=&\left(\frac{3-29{k^S_0}^2+28{k^S_0}^4+12k^S_1k^S_0}{12{k^S_0}^3}+2k^S_1+k^S_2-{k^S_1}^2k^S_2\right)\cos(k^S_0)\nonumber\\
&&+\left(\frac{k^S_1-2k^S_1k^S_0-2{k^S_0}^3}{k^S_0}-\frac{3}{2}{k^S_1}^2\right)\sin(k^S_0).
\eea
Equating the above order by order in $m^2$, we get
\bea
\sin(k^S_0)&=&2k^S_0\cos(k^S_0),\\
k^S_1&=&\frac{1-2{k^S_0}^2}{k^S_0(1-4{k^S_0}^2)},\\
k^S_2&=&\frac{(3-4{k^S_0}^2)(-5+35{k^S_0}^2-40{k^S_0}^4+16{k^S_0}^6)}{12{k^S_0}^3(1-4{k^S_0}^2)^3}.
\eea
\bea
\sum_{k \in \cK_S} L^4 \frac{1}{k^2}
&=&2L^4\sum_{k^S_0\in \cK_g}\left(\frac{1}{{k^S_0}^2}-2m^2\left(\frac{1}{{k^S_0}^4}+\frac{2}{{k^S_0}^2}-\frac{8}{4{k^S_0}^2-1}\right)\right.\nonumber\\
&&\left.+m^4\left(\frac{11}{2{k^S_0}^6}+\frac{127}{6{k^S_0}^4}+\frac{280}{3{k^S_0}^2}+\frac{32}{(4{k^S_0}^2-1)^3}+\frac{32}{(4{k^S_0}^2-1)^2}-\frac{1120}{3(4{k^S_0}^2-1)}\right)\right)+\cO(m^6).\nonumber\\ \label{eqn:gcomp}
\eea
We evaluate the above series by using the method developed in \cite{speigel} and used in \cite{Afshordi:2012ez,johnston}, details of which can be found in Appendix \ref{app:speigel}. This leads to
\be
\sum_{k^S_0\in\cK_g}\frac{1}{{k^S_0}^2}=\frac{5}{6}\;\;,\;\;\;\sum_{k^S_0\in\cK_g}\frac{1}{{k^S_0}^4}=\frac{49}{90}\;\;\text{and}\;\;\;\sum_{k^S_0\in\cK_g}\frac{1}{{k^S_0}^6}=\frac{377}{945} \label{eq:ser1}
\ee
and
\bea
\sum_{k^S_0\in\cK_g}\frac{1}{4{k^S_0}^2-1}&=&\frac{1}{4},\nonumber\\
\sum_{k^S_0\in\cK_g}\frac{1}{(4{k^S_0}^2-1)^2}&=&-\frac{1}{4}\left(\frac{\cos(1/2)-2\sin(1/2)}{\cos(1/2)-\sin(1/2)}\right),\nonumber\\
\sum_{k^S_0\in\cK_g}\frac{1}{(4{k^S_0}^2-1)^3}&=&\frac{1}{64}\left(1+\frac{19\cos(1/2)-35\sin(1/2)}{\cos(1/2)-\sin(1/2)}\right).\label{eq:ser2}
\eea
This simplifies Eqn.~(\ref{eqn:gcomp}) to 
\be
\sum_{k\in \cK_S} 2L^4\frac{1}{k^2}=2L^4\left(\frac{5}{6}-\frac{19}{45}m^2+\frac{41}{252}m^4\right)+\cO(m^6).
\ee
Adding the contributions from the antisymmetric and symmetric eigenfunctions the RHS of Eqn.~(\ref{eq:trace}) reduces to 
\be
\sum\lambda_k^2=2L^4\left(1-\frac{4}{9}m^2+\frac{1}{6}m^4\right)+\cO(m^6),
\ee
which is same as its LHS. Thus, to $\cO(m^4)$  the $\uAS_k$ are a complete set of eigenfunctions of $i\hD$.

\section{The Wightman function: the small mass limit}
\label{sec:wightmann}
We can now write down the formal expression for the SJ  Wightman function to $\cO(m^4)$  using the SJ modes
obtained above, as 
\begin{equation}
\wsj(u,v,u',v')=\sum_{k\in\cK_{A},\, k<0}-\frac{L^2}{k}\frac{u^{A}_k(u,v){u^{A}_k}^*(u',v')}{||u_k^{A}||^2}
+\sum_{k\in\cK_{S}, \, k<0}-\frac{L^2}{k}\frac{u^{S}_k(u,v){u^{S}_k}^*(u',v')}{||u_k^{S}||^2}+\cO(m^6),  \label{eq:fullw}
\end{equation}
where $\cK_{A/S}$ denote  the positive SJ eigenvalues.  In particular $k=-\ka(n)$ with $n\in{\mathbb Z}^+$ (Eqn.~(\ref{eq:qcondka})) and  $k=-\ks(k_0)$ with
$k_0$ satisfying  $\tan(k_0)=2k_0$ (Eqn.~(\ref{eq:qcondks})). 
Here $||u^{A/S}_k||$ denotes  the $\mL^2$ norm of the modes $u^{A/S}_k$ 
\be
||u^{A/S}_k||^2=L^2\int_{-1}^1 du \int_{-1}^1 dv u^{A/S}_k(u,v){u^{A/S}_k}^*(u,v). 
\ee
For $k=-\ka(n)$ 
\be
||u^A_k||^2 =
8L^2\left(1+\frac{m^2}{n^2\pi^2}+\frac{m^4}{n^2\pi^2}\left(\frac{1}{12}-\frac{11}{4n^2\pi^2}\right)\right)+\cO(m^6). \label{eq:norma} 
\ee
In the symmetric case, $k =-\ks(k_0) $ the quantization condition is complicated. Following \cite{Afshordi:2012ez}, we
make the approximation
\be
\ks(n) \approx \left(n-\frac{1}{2}\right)\pi, \, \, n\in {\mathbb Z}^+.   \label{eq:appquant}
\ee
As shown in  Fig.~\ref{fig:quant}, we see that except for the first few modes this is a good approximation, and in fact
improves with increasing mass\footnote{Of course, at the same time, our approximation of the SJ modes becomes worse with
  increasing mass.}.  This approximation in the quantization condition makes $\cos(\ks)=0$, thus simplifying
$u^S_k(u,v)$ to 
\be
u^S_{-\ks}(u,v)=U^S_{-i\ks}(u,v) \Rightarrow ||u^S_{\ks}||=8L^2. \label{eq:norms} 
\ee
\begin{figure}[h]
\centerline{\includegraphics[height=6cm]{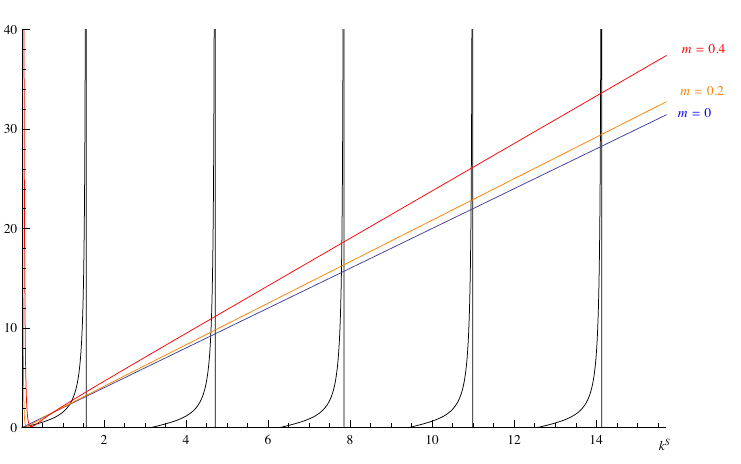}} 
\caption{Plot of the quantization condition, Eqn.~(\ref{eq:qcondS}) for the symmetric SJ modes for m=0,0.2 and 0.4,
  where $k_S>0$. }
\label{fig:quant}
\end{figure}

We examine the antisymmetric and symmetric contributions to $\wsj$ separately 
\be
\wsj=\wsj^A+\wsj^S. \label{eq:wsjsplit} 
\ee

For the antisymmetric contribution, using the quantization condition $k=-k_A(n)$ and the simplification Eqn.~(\ref{eq:norma}) for
the norm
\begin{equation}
  \wsj^A(u,v,u',v') = \sum_{n=1}^\infty \frac{1}{8n\pi}\left(1-\frac{2m^2}{n^2\pi^2}+\frac{m^4}{n^2\pi^2}\left(\frac{7}{n^2\pi^2}-\frac{1}{6}\right)\right) u^A_{k}(u,v) u_{k}^{A*}(u',v') + \cO(m^6).
\end{equation}
To leading order $u^A_{k}$ can be re-expressed as 
\begin{eqnarray} 
  u^A_{k}(u,v) &=& e^{-in\pi u}-e^{-in\pi v}+ \Psi_A(n,u,v) + \cO(m^6),\nonumber  \\ 
\Psi_A(n,u,v) &=&  \sum_{j=1}^3\left(\frac{(-1)^nf_j(m;u,v)}{n^j}+\frac{g_j(m;u,v)e^{-in\pi
      u}}{n^j}-\frac{g_j(m;v,u)e^{-in\pi v}}{n^j}\right), 
\end{eqnarray}
where 
\bea
f_1(m;u,v)\equiv \frac{im^2}{2\pi}(u-v)-\frac{im^4}{24\pi}(u-v)(1+3uv), &\quad& g_1(m;u,v)\equiv -\frac{im^2(2u+v)}{2\pi}-\frac{im^4u}{12\pi},\nonumber\\
f_2(m;u,v)\equiv -\frac{m^4}{4\pi^2}(u^2-v^2), &\quad& g_2(m;u,v) \equiv -\frac{m^4(2u+v)^2}{8\pi^2},\nonumber\\
f_3(m;u,v)\equiv-\frac{3im^4}{4\pi^3}(u-v), &\quad& g_3(m;u,v)\equiv \frac{im^4(15u+6v)}{12\pi^3}. 
\label{eq:gjfj}
\eea
We further split
\be
\wsj^A=\wa+\waa+\waaa+\waaaa+\cO(m^6),   \label{eq:splitwsja}
\ee
where 
\bea
\wa&\equiv& \sum_{n=1}^\infty\frac{1}{8n\pi}\left(1-\frac{2m^2}{n^2\pi^2}+\frac{m^4}{n^2\pi^2}\left(\frac{7}{n^2\pi^2}-\frac{1}{6}\right)\right)\left(e^{-in\pi u}-e^{-in\pi v}\right)\left(e^{in\pi u'}-e^{in\pi v'}\right),\nonumber\\
\waa&\equiv& \sum_{n=1}^\infty\frac{1}{8n\pi}\left(1-\frac{2m^2}{n^2\pi^2}\right)\left(e^{-in\pi u}-e^{-in\pi v}\right)\Psi_A^*(n,u',v'),\nonumber\\
\waaa&\equiv& \sum_{n=1}^\infty\frac{1}{8n\pi}\left(1-\frac{2m^2}{n^2\pi^2}\right)\Psi_A(n,u,v)\left(e^{in\pi u'}-e^{in\pi v'}\right),\nonumber\\
\waaaa&\equiv& \sum_{n=1}^\infty\frac{1}{8n\pi}\Psi_A(n,u,v)\Psi_A^*(n,u',v').\nonumber\\
\label{eq:aone}
\eea
These terms can be further simplified to $\cO(m^4)$ as we have shown in  Appendix.~\ref{sec:wight-app}. 

For the symmetric contribution $\wsj^S$ we use the simplification Eqns (\ref{eq:appquant}) and (\ref{eq:norms}) to express   
\be
\wsj^S = \sum_{n=1}^\infty\frac{1}{4\pi(2n-1)}U^S_{-i\ks}(u,v){U^S}_{-i\ks}^*(u',v')+\e_m(u,v,u',v') + \cO(m^6). 
\ee
Here $\e_m(u,v;u',v')$ is the correction term coming from the approximation of the quantization condition Eqn.~(\ref{eq:appquant}). This is
analytically difficult to obtain and in Sec.~\ref{sec:numerical}, we will evaluate it numerically for different values of $m$.   

Using the $\cO(m^4)$ expansion of $U_{-ik}$ from Eqn.~(\ref{eq:seriesua}), we write $U_{-ik_S}^S$ as
\bea
U_{-ik_S(n)}^S(u,v)&=&\left(e^{-i\left(n-\frac{1}{2}\right)\pi u}+e^{-i\left(n-\frac{1}{2}\right)\pi v}\right)+\Psi_S(n,u,v)+\cO(m^6), \nonumber\\
\Psi_S(n,u,v)&=&-\frac{im^2}{(2n-1)\pi}\left(v e^{-i\left(n-\frac{1}{2}\right)\pi u}+ue^{-i\left(n-\frac{1}{2}\right)\pi v}\right)\nonumber\\
&&\quad\quad\quad-\frac{m^4}{4(2n-1)^2\pi^2}\left(v^2e^{-i\left(n-\frac{1}{2}\right)\pi u}+u^2e^{-i\left(n-\frac{1}{2}\right)\pi v}\right).
\eea
Again for the symmetric part, we can write
\be
\wsj^S=\ws+\wss+\wsss+\wssss+\e_m(u,v,u',v')+\cO(m^6),
\label{eq:splitwsjs} 
\ee
where
\bea
\ws &\equiv& \frac{1}{4\pi}\sum_{n=1}^\infty\frac{1}{2n-1}\left(e^{-i\left(n-\frac{1}{2}\right)\pi u}+e^{-i\left(n-\frac{1}{2}\right)\pi v}\right)\left(e^{i\left(n-\frac{1}{2}\right)\pi u'}+e^{i\left(n-\frac{1}{2}\right)\pi v'}\right),\nonumber\\
\wss &\equiv& \frac{1}{4\pi}\sum_{n=1}^\infty\frac{1}{2n-1}\left(e^{-i\left(n-\frac{1}{2}\right)\pi u}+e^{-i\left(n-\frac{1}{2}\right)\pi v}\right)\Psi_S^*(n,u',v'),\nonumber\\
\wsss &\equiv& \frac{1}{4\pi}\sum_{n=1}^\infty\frac{1}{2n-1}\Psi_S(n,u,v)\left(e^{i\left(n-\frac{1}{2}\right)\pi u'}+e^{i\left(n-\frac{1}{2}\right)\pi v'}\right),\nonumber\\
\wssss &\equiv& \frac{1}{4\pi}\sum_{n=1}^\infty\frac{1}{2n-1}\Psi_S(n,u,v)\Psi_S^*(n,u',v').
\label{eq:sone} 
\eea
Using the following result
\be
\sum_{n=1}^\infty \frac{e^{i\left(n-\frac{1}{2}\right)\pi x}}{(2n-1)^j} =
\li_j\left(e^{i\pi\frac{x}{2}}\right)-\frac{1}{2^j}\li_j\left(e^{i\pi x}\right),
\label{eq:soneexp}
\ee
$\ws,\wss,\wsss$ and $\wssss$ can further be simplified up to $\cO(m^4)$ as we have shown in Appendix \ref{sec:wight-app}.
In particular, $\ws$ can be written as
\be
\ws=\frac{1}{4\pi}\left(\tanh^{-1}\left(e^{-\frac{i\pi(u-u')}{2}}\right)+\tanh^{-1}\left(e^{-\frac{i\pi(v-v')}{2}}\right)+\tanh^{-1}\left(e^{-\frac{i\pi(u-v')}{2}}\right)+\tanh^{-1}\left(e^{-\frac{i\pi(v-u')}{2}}\right)\right).\label{eq:sonetanh}
\ee

Despite these simplifications in $\wsj$, it  is difficult to find a general closed form expression for $\wsj$. Instead, as  was done in \cite{Afshordi:2012ez}, we
focus on two subregions of $\diam$, as shown in Fig.~\ref{fig:cd}. In the center, far away from the boundary,  one expects to obtain the Minkowski
vacuum, while in the corner, one expects the Rindler vacuum. In  the massless case studied by \cite{Afshordi:2012ez} the
former expectation was shown to be the case. However, in the corner, instead of the Rindler vacuum, they found that 
that $\wsj$ looks like the massless mirror vacuum. One of the main motivations to look at the massive case, is to compare with
these results. 

\begin{figure}[h]
\centering{\includegraphics[height=4.5cm]{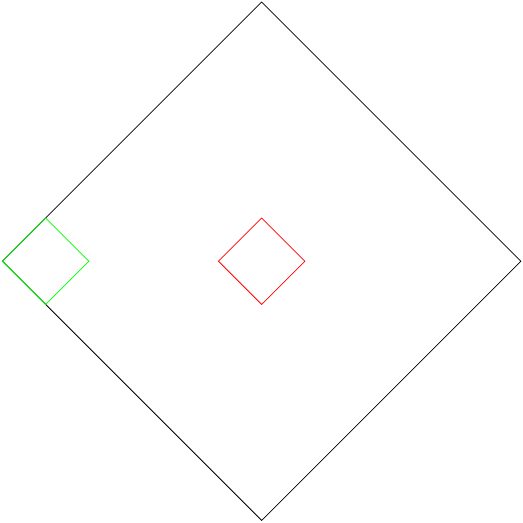}
\caption{The  center and corner regions in the causal diamond $\diam$.}
\label{fig:cd}}
\end{figure}

We now write down the expressions for the various vacua that we wish to compare with: 
\begin{eqnarray}
  \wmink(u,v;u',v')&=&-\frac{1}{4\pi}\ln\left(\cof^2e^{2\gamma}|2\Du\Dv|\right)-\frac{i}{4}\sgn(\Du+\Dv)\theta(\Du\Dv),\label{eq:minkm0}\\
  \wminkm(u,v;u',v')&=&\frac{1}{2\pi}K_0\left(m\sqrt{-2\Du\Dv+i(\Du+\Dv)\e}\right), \label{eq:massmink} \\ 
 \wrind(\eta,\xi,
  \eta',\xi')&=&-\frac{1}{4\pi}\ln\left(\cof^2e^{2\gamma}|\Delta\eta^2-\Delta\xi^2|\right)-\frac{i}{4}\sgn(\Delta\eta)\theta(\Delta\eta^2-\Delta\xi^2), \label{eq:rindm0}
  \\
  \wrindm(\eta,\xi, \eta',\xi') &=& \wminkm(u,v,u',v')-\frac{1}{2\pi}\int_{-\infty}^\infty
                                    \frac{dy}{\pi^2+y^2}K_0(m\gamma_1), \label{eq:rindm} \\
 \wmirr(u,v,u',v')&=&\wmink(u,v;u',v')-\wmink(u,v;v',u'), \label{eq:mirror} \\
\wmirrm(u,v,u',v')&=&\wminkm(u,v;u',v')-\wminkm(u,v;v',u').   \label{eq:mirrorm}                                 
\end{eqnarray}
In the expression Eqn.~(\ref{eq:minkm0})   for the massless Minkowski vacuum,  $\gamma$ is the Euler-Mascheroni constant
and $\cof=0.462$ (obtained in \cite{Afshordi:2012ez} by comparing $\wsj$ with $\wmink$). In the expression  Eqn.~(\ref{eq:massmink})  for the
massive Minkowski vacuum \cite{abdallah}, $K_0$ is the modified Bessel function of the second kind, with $\e$ a constant
such that that $0< \e \ll 1$.
In the expressions Eqn.~(\ref{eq:rindm0}) and Eqn.~(\ref{eq:rindm}) (see \cite{candelas:1976})
 for the Rindler vacua, $\alpha$ is the acceleration
parameter, with
\begin{eqnarray} 
\eta = \frac{1}{\alpha}\tanh^{-1}\left(\frac{u+v}{u-v}\right),&\quad & 
\xi= \frac{1}{2\alpha}\ln\left(-2\alpha^2uv\right), \nonumber \\ 
\Delta\eta=\eta-\eta', \quad \Delta\xi=\xi-\xi',  &\quad& \gamma_1=\sqrt{\xi^2+\xi'^2+2\xi\xi'\cosh(y-\eta+\eta')}.\label{eq.g1}
\end{eqnarray} 

\subsection{The center}

We now consider a small diamond ${\diam}_l$ at the center of  $\diam$ with $l \ll 1$  where one expects  $\wsj$ to resemble  
$\wminkm$.  For small $\Delta u,
\Delta v$, $\wminkm$ can be written as
\be
\wminkm(u,v;u',v')\approx
-\frac{1}{4\pi}\ln\left(\frac{m^2e^{2\gamma}}{2}\left|\Du\Dv\right|\right)-\frac{i}{4}\sgn(\Du+\Dv)\theta(\Du\Dv)J_0\left(m\sqrt{2\Du\Dv}\right).\label{eq:centerminkm}
\ee
To leading logarithmic order this is similar in form to $\wmink$ (Eqn.~(\ref{eq:minkm0})), with $m$ replaced by $2 \cof$.
We plot these functions in Fig.~\ref{fig:wmink}. For $m \ll \cof$ the real part of $\wminkm$ is larger than
$\wmink$ and for $m \gg \cof$ it is smaller. When $m_c=2\cof$, the two coincide in this approximation. 

\begin{figure}[h]
\centerline{\includegraphics{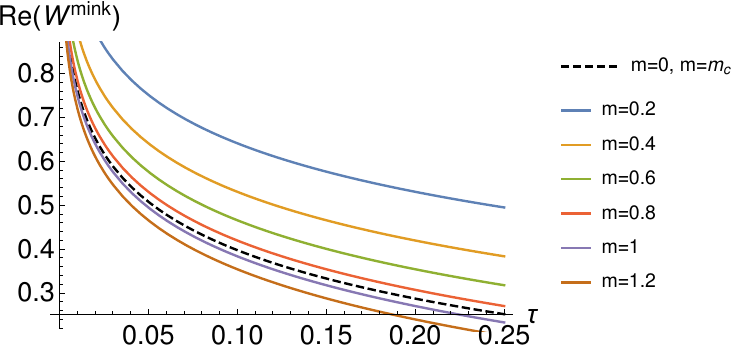}}
\caption{Plot of $\re(\wmink)$ and $\re(\wminkm)$ vs the proper time ($\tau$)}
\label{fig:wmink}
\end{figure}

Let us begin with $\wsj^A$, Eqns  (\ref{eq:splitwsja})  and (\ref{eq:aone}). As shown in  Appendix
\ref{sec:wight-app},  the expressions for $\wa,\waa,\waaa$ and $\waaaa$ can be written in terms of Polylogarithms 
$\li_s(x)$.  For small $x$, i.e., near the center of $\diam$ they simplify for the  $s=1,3$ and $5$ to 
\bea
\li_1\left(e^{i\pi x}\right)&=& -\ln(-i\pi x)-\frac{i\pi x}{2}+\frac{\pi^2x^2}{24}+\cO(x^3), \label{eq:li1}\\
\li_3\left(e^{i\pi x}\right)&=& \zeta(3)+\frac{i\pi^3x}{6}+\left(-\frac{3\pi^2}{4}+\frac{\pi^2}{2}\ln(-i\pi
  x)\right)x^2+\cO(x^3), \\
\li_5\left(e^{i\pi x}\right)&=& \zeta(5)+\frac{i\pi^5x}{90}-\frac{\pi^2\zeta(3) x^2}{2}+\cO(x^3),
\label{eq:polylog} 
\eea
where $\zeta$ are the Riemann Zeta function and $x$ denotes $u$ or $v$.  In the expression for $\wa$, the
constant and linear terms in $x$ cancel out, so that
\bea
\wa &=& -\frac{1}{8\pi}\left(\ln(|u-u'||v-v'|)-\ln(|u-v'||v-u'|)-C_1\frac{i\pi}{2}\right)\nonumber\\
&&-\left(\frac{\pi}{96}+\frac{3m^2}{8\pi}+\frac{m^4}{8\pi}\left(\frac{1}{4}-\frac{7\zeta(3)}{\pi^2}\right)\right)(u-v)(u'-v')\nonumber\\
&&-\frac{m^2}{8\pi}\left(1+\frac{m^2}{12}\right)\left[(u-u')^2\ln\left(-i\pi(u-u')\right)+(v-v')^2\ln\left(-i\pi(v-v')\right)\right.\nonumber\\
&&\left.-(u-v')^2\ln\left(-i\pi(u-v')\right)-(v-u')^2\ln\left(-i\pi(v-u')\right) \right] +\cO(\Delta^3), 
\eea
where $C_1=\sgn(u-u')+\sgn(v-v')-\sgn(u-v')-\sgn(v-u')$ and $\Delta$ collectively denotes either $u-u',v-v',v'-u$ or $ v-u'$.   
For sufficiently small $x$, the logarithmic term dominates significantly over other terms,  and hence in $\diam_l$ 
\be
\wa=-\frac{1}{8\pi}\left(\ln(|u-u'||v-v'|)-\ln(|u-v'||v-u'|)-C_1\frac{i\pi}{2}\right) + \cO(m^2,\Delta^2), \label{eq:wa-cent}
\ee
where we have hidden all the mass dependence in the correction.

Next, $\waa,\waaa$ and $\waaaa$ also involve another set of Polylogarithms of the type $\li_s(-e^{i\pi x})$ for $ s \geq 2$
as well as $\li_s(e^{ i\pi x})$ for $s=2,3,4$, which are multiplied to the functions $g_j(m;u,v)$ and $f_j(m;u,v)$ given in Eqn.~(\ref{eq:gjfj}). The $g_j(m;u,v)$ and $f_j(m;u,v)$ themselves go to zero  either linearly or quadratically with $u,v$. This second set of  Polylogarithms, unlike the first in Eqn.~(\ref{eq:polylog}), are strictly convergent as $x \rightarrow 0$. Hence the $\waa,\waaa $ and $\waaaa$ are  
strongly sub-dominant with respect to $\wa$ so that 
\be
\wsj^A(u,v,u',v') = -\frac{1}{8\pi}\left(\ln(|u-u'||v-v'|)-\ln(|u-v'||v-u'|)-C_1\frac{i\pi}{2}\right)+\cO(m^2,\Delta^2).
\ee
Here we note that while the mass correction is significant in the antisymmetric SJ modes, it becomes  insignificant in $\wsj^A$ in the center of the diamond, compared to the dominating logarithmic term.  Thus we see that in the center of $\diam$,  $\wsj^A$ is identical to the massless case found in \cite{Afshordi:2012ez}. 

We now turn to the symmetric part $\wsj^S$, Eqns  (\ref{eq:splitwsjs})  and (\ref{eq:sone}). The expressions for
$\ws,\wss,\wsss$  and $\wssss$ can again be written in terms of Polylogarithms $\li_s(x)$  as shown in Appendix
\ref{sec:wight-app}. For $\ws$ however, the form given in Eqn.~(\ref{eq:sonetanh}) is easier to analyze.  
Noting that for small $x$ 
\be
\tanh^{-1}\left(e^{i\pi x/2}\right)=-\frac{1}{2}\ln\left(\frac{-i\pi x}{4}\right)-\frac{\pi^2x^2}{96}+\cO(x^3),
\ee
near the center of $\diam$ we see that
\bea
\ws&=&-\frac{1}{8\pi}\left[\ln(|u-u'||v-v'|)+\ln(|u-v'||v-u'|)+4\ln\left(\frac{\pi}{4}\right)-C_2\frac{i\pi}{2}\right]\nonumber\\
&&-\frac{\pi}{384}\left((u-u')^2+(u-v')^2+(v-u')^2+(v-v')^2\right)+\cO(\Delta^3),
\eea
where $C_2=\sgn(u-u')+\sgn(v-v')+\sgn(u-v')+\sgn(v-u')$. Since the logarithmic term dominates,
\be
\ws=-\frac{1}{8\pi}\left[\ln(|u-u'||v-v'|)+\ln(|u-v'||v-u'|)+4\ln\left(\frac{\pi}{4}\right)-C_2\frac{i\pi}{2}\right]+\cO(\Delta^2).  
\ee
Next, we see that $\wss,\wsss$ and $\wssss$ involve a set of Polylogarithms of the type $\li_s(e^{i\pi x}),$ for $
s=2,3$, multiplied by linear and quadratic functions of $u,v,u'$ and $v'$. This set of  Polylogarithms are in fact
strictly convergent as $x \rightarrow 0$. Hence the $\wss,\wsss $ and $\wssss$ are strongly sub-dominant, with respect
to $\ws$, so that
\bea
\wsj^S(u,v,u',v')&=&-\frac{1}{8\pi}\left[\ln(|u-u'||v-v'|)+\ln(|u-v'||v-u'|)+4\ln\left(\frac{\pi}{4}\right)-C_2\frac{i\pi}{2}\right]\nonumber\\
&&\quad+\emc+\cO(m^2,\Delta^2), \label{eq:emc} 
\eea
where $\emc$ is the correction in the center coming from the
approximation to the quantization condition Eqn.~(\ref{eq:appquant}). We will determine this 
numerically in Section \ref{sec:numerical}.
Up to this mass correction $\wsj^S$ resembles the massless case found in \cite{Afshordi:2012ez}.

Putting these pieces together we find that
\be
\wsj^{center}(u,v,u',v') \approx 
-\frac{1}{4\pi}\ln|\Du\Dv|-\frac{i}{4}\sgn(\Du+\Dv)\theta(\Du\Dv)-\frac{1}{2\pi}\ln\left(\frac{\pi}{4}\right)+\e_m^{center}. \label{eq:wcent}
\ee
A direct comparison with $\wmink$ gives 
\begin{equation}
\wsj^{center}(u,v,u',v') - \wmink(u,v,u',v') \approx -\frac{1}{2\pi}\ln\left(\frac{\pi}{4}\right)+\e_m^{center} +
\frac{1}{4\pi}\ln\left(2 \cof^2e^{2\gamma}\right),  \label{eq:compcenter} 
  \end{equation} 
where $\cof \approx 0.462$ is fixed by comparing the massless $\wsj$ with $\wmink$ as in \cite{Afshordi:2012ez}.

\subsection {The corner}

We now consider either of the two spatial corners of the diamond, $\diam_c \subset \diam$ as shown in Fig.~\ref{fig:cd}. We use the small $\Delta u, \Delta v$ form of $\wminkm$ to
express
\be
\wmirrm \approx -\frac{1}{4\pi}\ln\left|\frac{\Du\Dv}{(u-v')(v-u')}\right|-\frac{i}{4}\sgn(\Du+\Dv)\left(\theta(\Du\Dv)-\theta((u-v')(v-u'))\right) \label{eq:approxmirrm} .
 \ee
 As in \cite{Afshordi:2012ez} we make the coordinate transformation
\be
\{u,u',v,v'\}\rightarrow \{u_c,u'_c,v_c,v'_c\} \equiv \{u-1,u'-1,v+1,v'+1\}, \label{eq:corner}
\ee
which brings the origin $(0,0)$ to the left corner of the diamond. 

For $\wsj^A$ (Eqn.~(\ref{eq:splitwsja}) and Eqn.~(\ref{eq:aone})),  we note that $\wa$ is invariant under this coordinate
transformation and hence given by Eqn.~(\ref{eq:wa-cent}) near the origin of $\diam_c$. In $\waa,\waaa$ and $\waaaa$ the
constant terms cancel out and, similar to the center calculation, they goes to zero linearly with $u,v$ and hence are
strongly sub-dominant with respect to $\wa$. 
Therefore,  in the corner, $\wsj^A$ simplifies to
\be
\wsj^A(u,v,u',v') = -\frac{1}{8\pi}\left(\ln(|u-u'||v-v'|)-\ln(|u-v'||v-u'|)-C_1\frac{i\pi}{2}\right)+\cO(m^2,\Delta),
\ee
and the sub-dominant part is now linear in $\Delta$.

For $\wsj^S$ (Eqn.~(\ref{eq:splitwsjs}) and Eqn.~(\ref{eq:sone})), under the coordinate transformation
\be
\ws=\frac{1}{4\pi}\left(\tanh^{-1}\left(e^{-\frac{i\pi(u-u')}{2}}\right)+\tanh^{-1}\left(e^{-\frac{i\pi(v-v')}{2}}\right)-\tanh^{-1}\left(e^{-\frac{i\pi(u-v')}{2}}\right)-\tanh^{-1}\left(e^{-\frac{i\pi(v-u')}{2}}\right)\right).
\ee
In the corner $\diam_c\subset\diam$ this simplifies to
\bea
\ws&=&\frac{1}{8\pi}\left[-\ln(|u-u'||v-v'|)+\ln(|u-v'||v-u'|)+C_1\frac{i\pi}{2}\right]\nonumber\\
&&-\frac{\pi}{384}\left((u-u')^2+(v-v')^2-(u-v')^2-(v-u')^2\right)+\cO(\Delta^3).
\eea
For sufficiently small $\Delta$, the logarithmic term dominates the other terms so that
\be
\ws = \frac{1}{8\pi}\left[-\ln(|u-u'||v-v'|)+\ln(|u-v'||v-u'|)+C_1\frac{i\pi}{2}\right]+\cO(\Delta^2).
\ee
As in the center, $\wss$ and $\wsss$ go  to zero while 
\be
\wssss=\frac{7\zeta(3)m^4}{8\pi^3}+\cO(\Delta)\approx 0.034m^4.
\ee
Therefore in the corner  we see that 
\be
\wsj^S \approx  \frac{1}{8\pi}\left[-\ln(|u-u'||v-v'|)+\ln(|u-v'||v-u'|)+C_1\frac{i\pi}{2}\right]+0.034m^4+\emcr  \label{eq:emcr}
\ee
i.e., there is a mass correction to the massless $\wsj^S$. $\emcr$ is, as in the center calculation, a small but finite term coming from the
approximation to the quantization condition Eqn.~(\ref{eq:appquant}), which we will  evaluate numerically in
Sec.~\ref{sec:numerical}. 

Putting these pieces together we find that in the corner $\wsj$ takes the form 
\bea
\wsj^{corner}(u,v,u',v') &\approx& -\frac{1}{4\pi}\ln\left|\frac{\Du\Dv}{(u-v')(v-u')}\right|-\frac{i}{4}\sgn(\Du+\Dv)\left(\theta(\Du\Dv)-\theta((u-v')(v-u'))\right)\nonumber\\
&&+0.034m^4+\emcr.
\eea
A direct comparison with $\wmirrm$ Eqn (\ref{eq:approxmirrm})  gives 
\be
\wsj^{corner}(u,v,u',v') - \wmirrm(u,v,u',v') \approx 0.034m^4+\emcr  \label{eq:compcorner} .
\ee 

\subsection {Numerical simulations for determining $\e_m$} \label{sec:numerical}
The formal expansion of $\wsj $ in terms of the SJ modes Eqn.~(\ref{eq:fullw}) can be truncated and evaluated numerically
in $\diam$. Here we do not need to use the approximation of the quantization condition Eqn (\ref{eq:appquant}). This
allows us to evaluate the ensuing corrections $\e_m^{center},\e_m^{corner}$ numerically,  and thus quantify the
comparisons of $\wsj$ obtained in the center and  corner of $\diam$ with the standard vacua. 

We begin with the $N^{\mathrm{th}}$ truncation $\wsj^t$ of the series form of $\wsj$ Eqn(\ref{eq:fullw}) in the full
diamond $\diam$ for $N=100, 200, \ldots 1000$.  Fig \ref{fig:truncations}  shows an explicit convergence of $\wsj^t$ for
these values of $N$. 
For the plot we considered  the pairs  $(u,v)=(x,x)$ and $(u',v')=(-x,-x)$ for timelike separated points, and $(u,v)=(x,-x)$ and $(u',v')=(-x,x)$
for spacelike separated points.  From this point onwards, we will consider $\wsj^t$ for $N=1000$.
\begin{figure}[h]
  \centering{\begin{tabular}{cc}
\includegraphics[height=4.2cm]{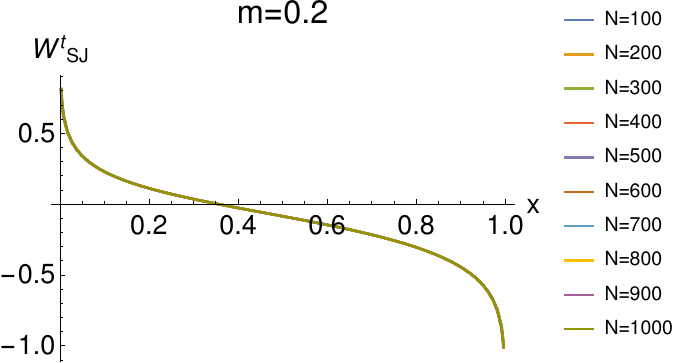} &\hskip 2cm 
                                                            \includegraphics[height=4.2cm]{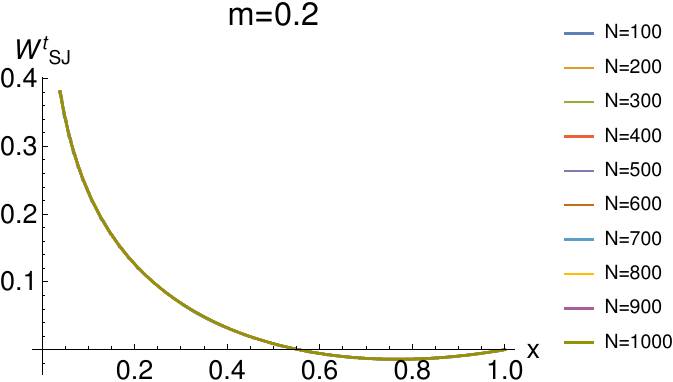}\\
                \includegraphics[height=4.2cm]{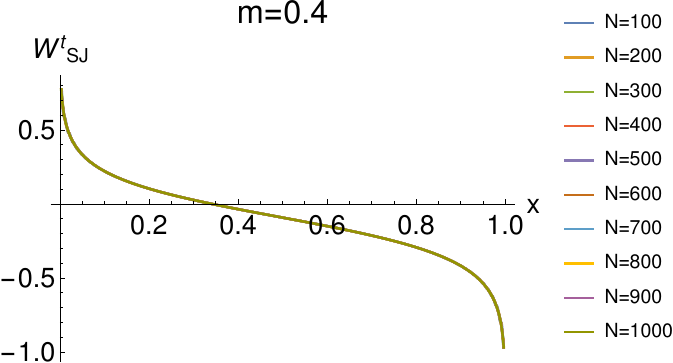} &\hskip 2cm 
\includegraphics[height=4.2cm]{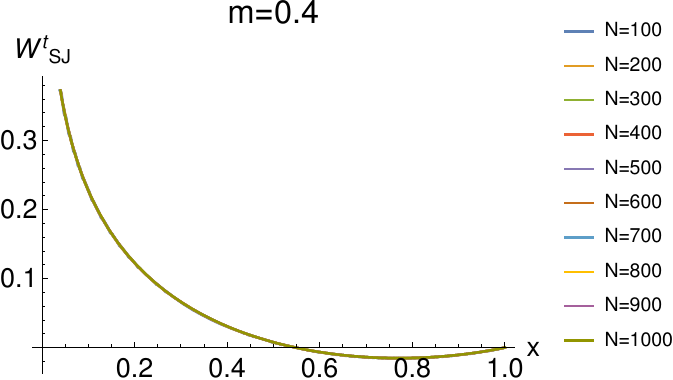} \\
\end{tabular}
\caption{We show the convergence of the truncation of the series $\wsj^t$ with $N$ for $m=0.2,0.4$ for timelike separated points (left) and spacelike separated points (right).}
\label{fig:truncations}}
\end{figure} 
  
Next, we consider the difference $\wsj^t-\wsj^{t,approx}$ where the latter uses the approximation Eqn.~
(\ref{eq:appquant}), both in the center and the corner of $\diam$ in order to obtain $\emc,\emcr$.  It suffices to
look at their  symmetric parts $\wsj^{S,t}$ since only these contribute (see Eqns (\ref{eq:emc}),
(\ref{eq:emcr})). $\emc$ and $\emcr$ are {\it not} strictly constants. However, as we will see, they are approximately
so. As in \cite{Afshordi:2012ez}, they are evaluated by taking a set of randomly selected points in a small diamond in
the center as well as in the corner. Here we take $10$ points and consider all $55$ pairs between them to calculate
$\emc,\emcr$. What we find in Fig.~\ref{fig:em} is that they are very nearly equal and hence we can consider their
average. The explicit averages for these masses are tabulated in Table \ref{tab:em} for future reference.  
\vskip 1cm 
\begin{figure}[h]
\centerline{\begin{tabular}{cc}
\includegraphics[height=4.2cm]{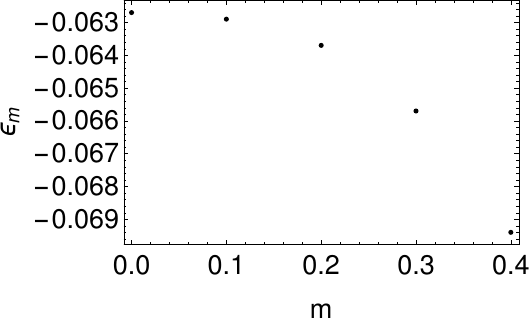}& \hskip 1cm 
\includegraphics[height=4.2cm]{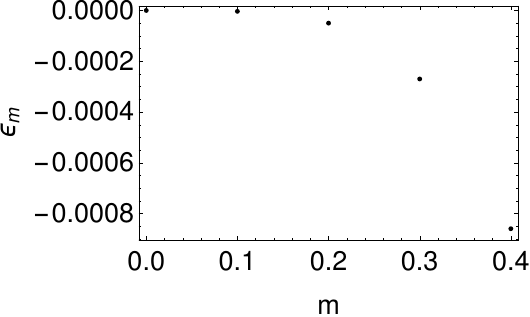}\\
\end{tabular}}
\caption{$\emc$ and $\emcr$  evaluated in a small diamond of $l=10^{-5}$ in the center and the corner of $\diam$, for
  m=0,0.1,0.2,0.3 and 0.4. The standard deviation is very small and hence we can take $\emc$ and $\emcr$ to be
  approximately constant.}
\label{fig:em}
\end{figure} 
\vskip 1cm 
\begin{table}[htb]
\centering{
\begin{tabular}{|c|c|c|}
\hline
mass & $\e_m^{center}$ & $\e_m^{corner}$\\
\hline
0 & -0.0627 & 0\\
0.1 & -0.0629 & $-3.5\times 10^{-6}$\\
0.2 & -0.0637 & -0.00005\\
0.3 & -0.0657 & -0.00027\\
0.4 & -0.0694 & -0.00086\\
\hline
\end{tabular}
\caption{A tabulation of $\emc, \emcr$  for different $m$}
\label{tab:em}}
\end{table}
\vskip 0.5cm

This allows us to now compare $\wsj$ calculated in the center Eqn (\ref{eq:wcent}) with $\wmink, \wminkm$. 
The difference with $\wmink$ given in  Eqn (\ref{eq:compcenter}) is indeed very small. For  $m=0.2$,  for example,  
\begin{equation}
\wmink - \wsj^{center}  \simeq -\frac{1}{4 \pi} \log(2 \times 0.462^2)- 
 \frac {\gamma}{2\pi} - (-\frac{1}{2 \pi} \log(\frac{\pi}{4}) - \emc) \simeq 0.001. 
\end{equation}
Similarly, in the corner, the difference with $\wmirrm$ is again very small. For example for $m=0.2$ it gives
\begin{equation} 
\wmirrm - \wsj^{corner}   \simeq   0.034\times (0.2)^4+\emcr \simeq 4\times 10^{-6} 
 \end{equation} 
Thus, we see that in the small mass limit, $\wsj$ does not differ from the massless Minkowski vacuum in the center region, and
continues to mimic the mirror vacuum in the corner.

Since our analytical calculation is restricted to a very small $\Delta u, \Delta v$, where perhaps the effect of a
small mass is small, we can use the truncation $\wsj^t$ for comparisons with the standard vacuum in larger regions of $\diam$. 
This is shown in the residue plots in Figs.~\ref{fig:res-m0}. In the full diamond, we consider 
 the pairs  $(u,v)=(x,x)$ and $(u',v')=(-x,-x)$ for timelike separated points, and $(u,v)=(x,-x)$ and $(u',v')=(-x,x)$
 for spacelike separated points. We find that for $m=0.2$, $l\sim 0.02$, $\wsj^t$ differs very
 little from the 
 massless Minkowski vacuum, while as the mass increases, so does the discrepancy. 
 \begin{figure}[h]
\centering{\begin{tabular}{cc}
              \includegraphics[height=4cm]{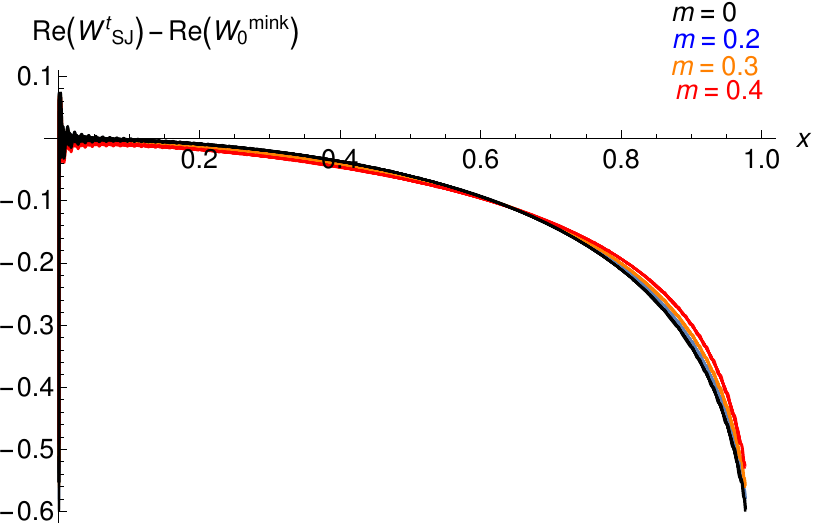}&\hskip 1cm
                                                                     \includegraphics[height=4cm]{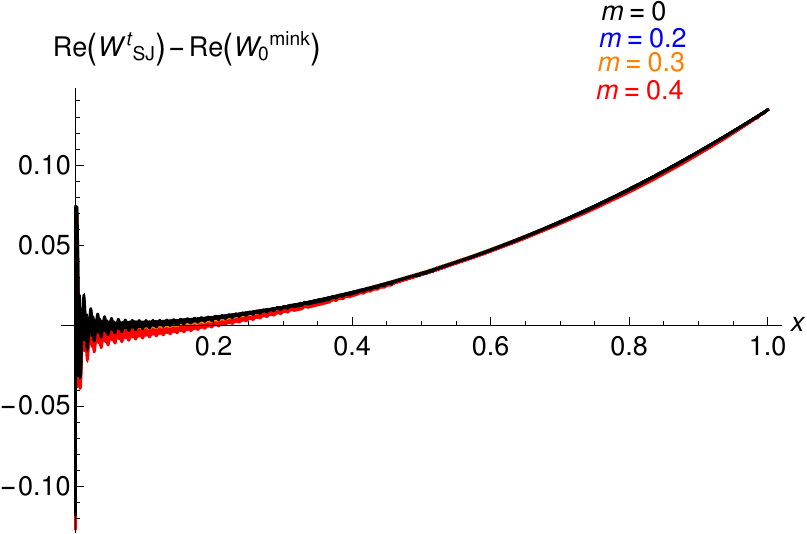}\\
\includegraphics[height=4cm]{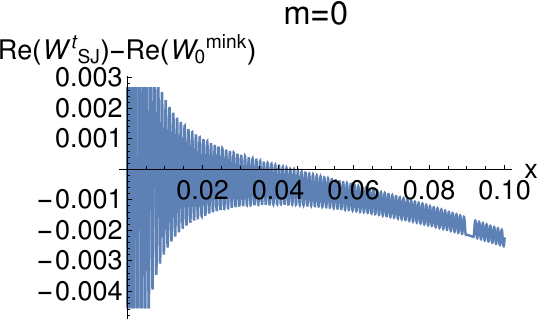}&\hskip 2cm
                                                            \includegraphics[height=4cm]{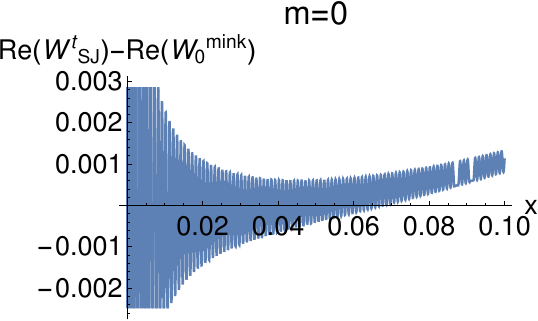}\\

         \includegraphics[height=4cm]{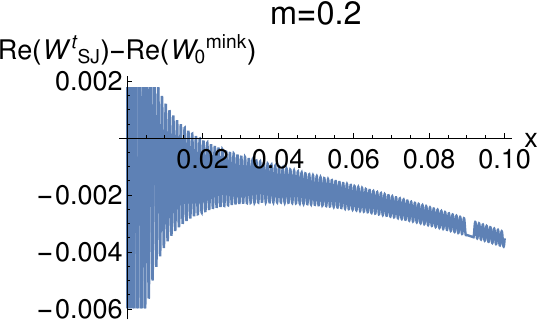}&\hskip 2cm
                                                            \includegraphics[height=4cm]{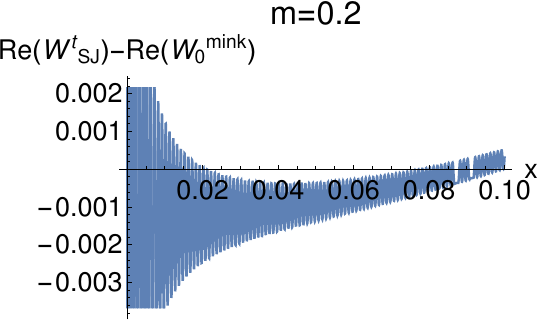}\\

              \includegraphics[height=4cm]{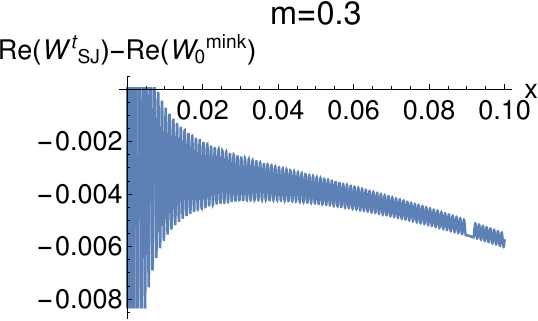}&\hskip 2cm
                                                            \includegraphics[height=4cm]{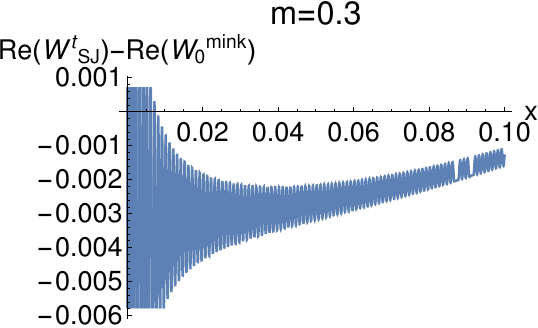}\\

              \includegraphics[height=4cm]{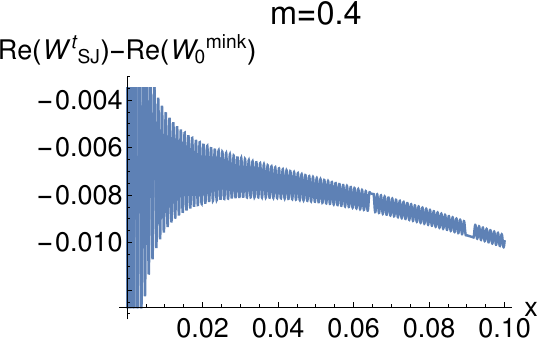}&\hskip 2cm
                                                            \includegraphics[height=4cm]{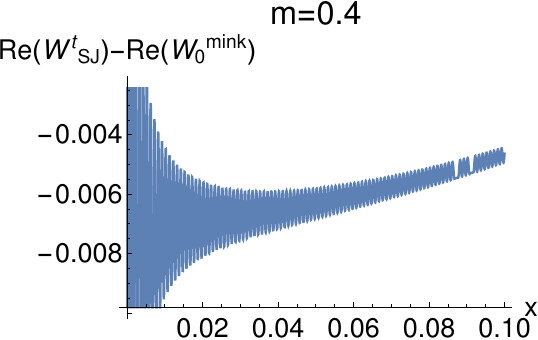}
            \end{tabular}
\caption{Residue plot of $\re(\wsj^t-\wmink)$ for timelike and spacelike separated points respectively, for the full
  diamond, as well as in a center region of size $l\sim 0.1$.}
\label{fig:res-m0}}
\end{figure}
On the other hand, as we see in  Figs.~\ref{fig:res-m}  we find that $\wsj^t$ clearly does  {\it not} agree with the massive Minkowski vacuum, in this small mass
limit.
\begin{figure}[h]
\centerline{\begin{tabular}{cc}
\includegraphics[height=5cm]{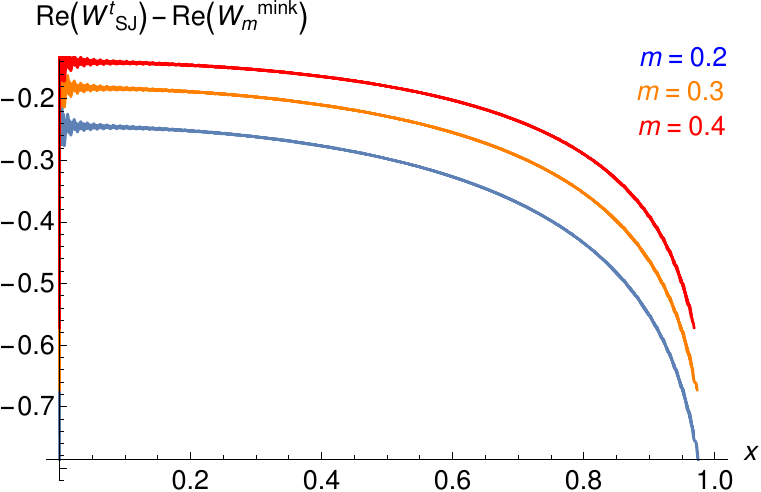}&\hskip 1cm
\includegraphics[height=5cm]{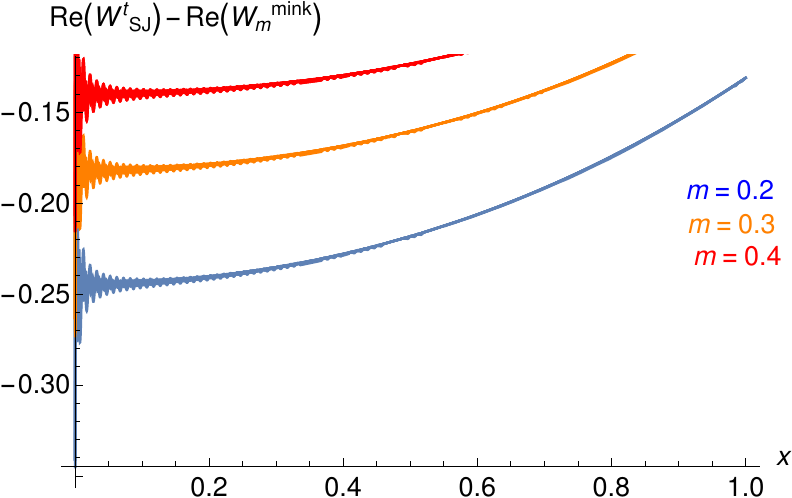}
\end{tabular}}
\caption{Residue plot of $\re(\wsj^t-\wminkm)$ for timelike and spacelike separated points respectively, for the full
  diamond. The discrepancy is obvious. }
\label{fig:res-m}
\end{figure}
 
A similar calculation in the corner shows that $\wsj^t$ looks like the massive mirror vacuum rather than the Rindler vacuum. Here, we consider pairs of points: $(u,v)=(l+x,-l+x)$ and $(u',v')=(l-x,-l-x)$ for timelike separation  and
$(u,v)=(l+x,-l-x)$ and $(u',v')=(l-x,-l+x)$ for spacelike separation, where the origin $(0,0)$ is at the left corner of the diamond $\diam$ and $2l$ is the length of the corner diamond $\diam_c$. This is shown in the residue plots in Figs.~\ref{fig:corres01} and \ref{fig:corres}.

\begin{figure}[h]
\centering{\begin{tabular}{cc}
\includegraphics[height=4cm]{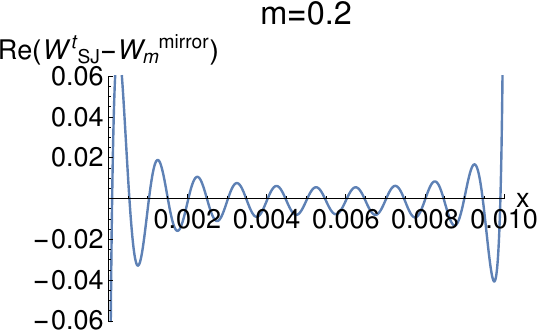}&\hskip 2cm
\includegraphics[height=4cm]{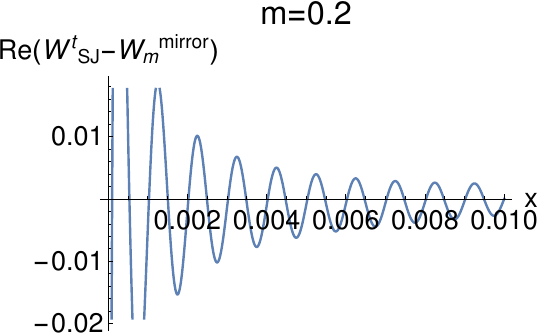}\\
\includegraphics[height=4cm]{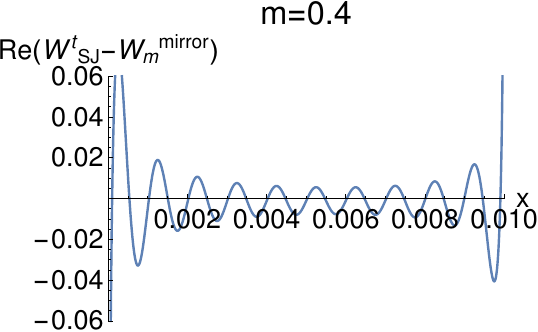}&\hskip 2cm
\includegraphics[height=4cm]{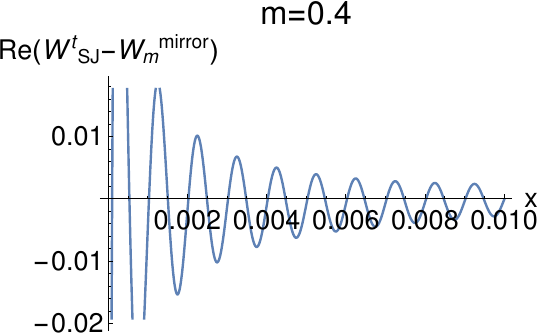}
\end{tabular}
\caption{Residue plot of $\re(\wsj^t-\wminkm)$ for timelike and spacelike separated points respectively in the corner
  region,  $l\sim 0.01$.}
\label{fig:corres01}}
\end{figure}
\begin{figure}[h]
\centering{\begin{tabular}{cc}
\includegraphics[height=4cm]{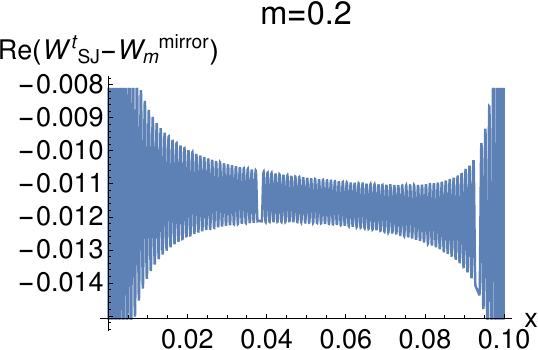}&\hskip 2cm
\includegraphics[height=4cm]{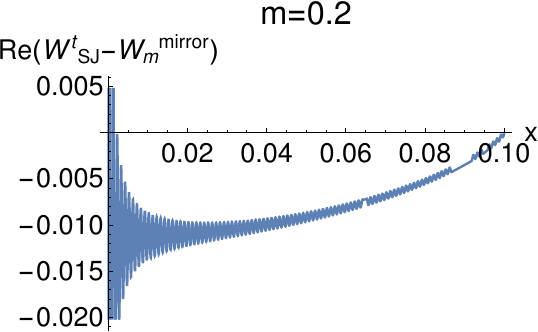}\\
\includegraphics[height=4cm]{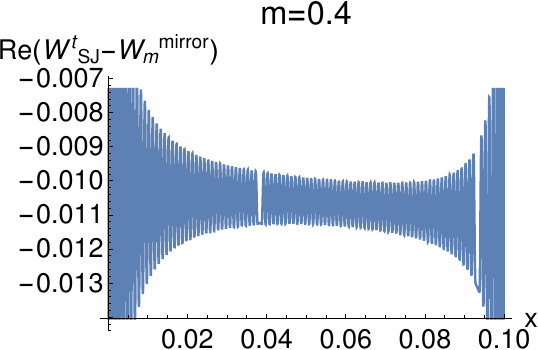}&\hskip 2cm
\includegraphics[height=4cm]{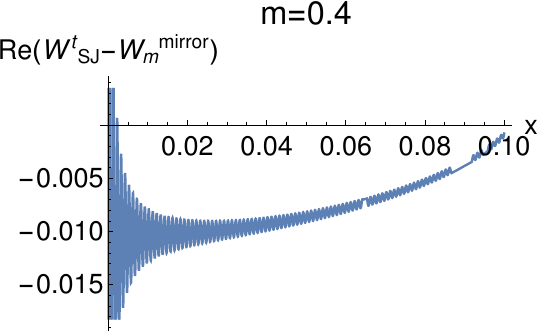}
\end{tabular}
\caption{Residue plot of $\re(\wsj^t-\wminkm)$ for timelike and spacelike separated points respectively in the corner
  region, $l\sim 0.1$.}
\label{fig:corres}}
\end{figure}
\begin{figure}[h]
\centerline{\begin{tabular}{ccc}
\includegraphics[height=3.5cm]{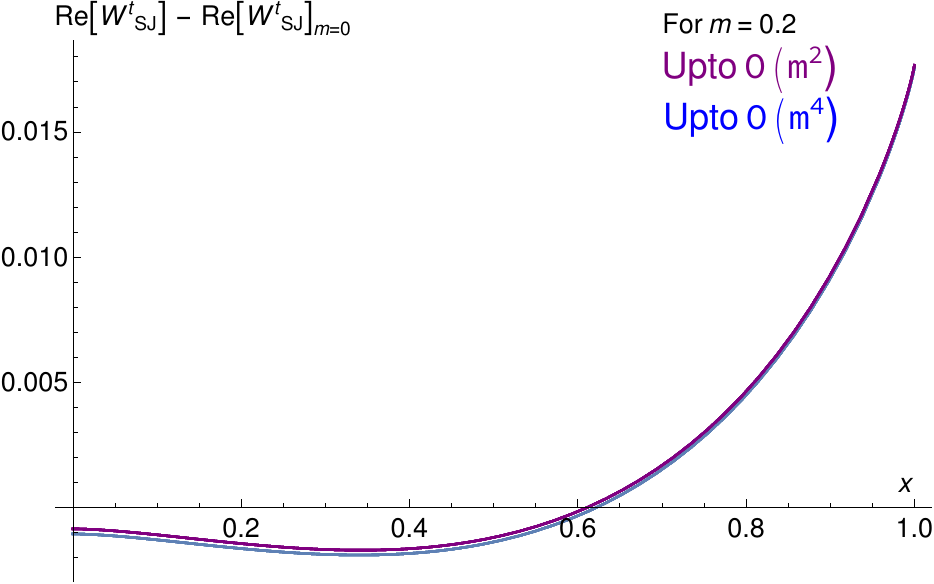}&
\includegraphics[height=3.5cm]{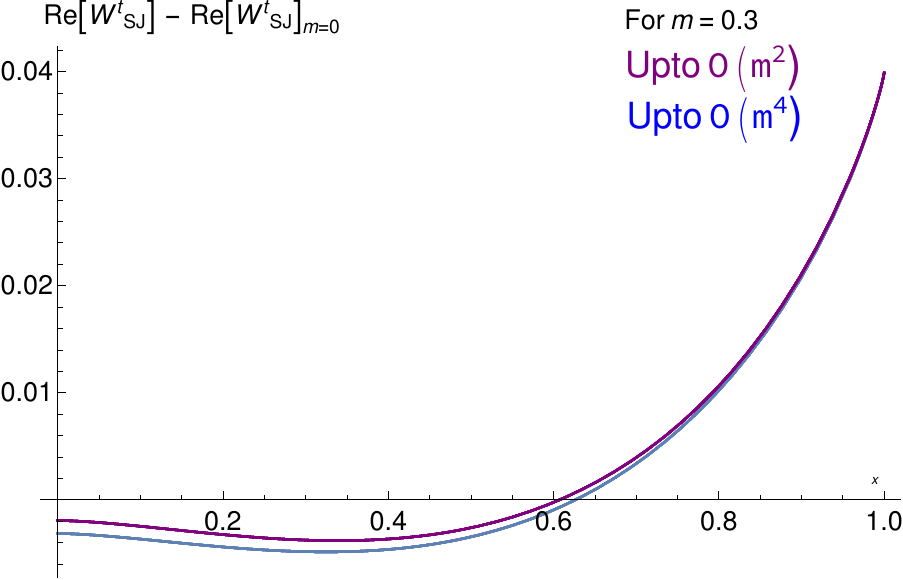}&
\includegraphics[height=3.5cm]{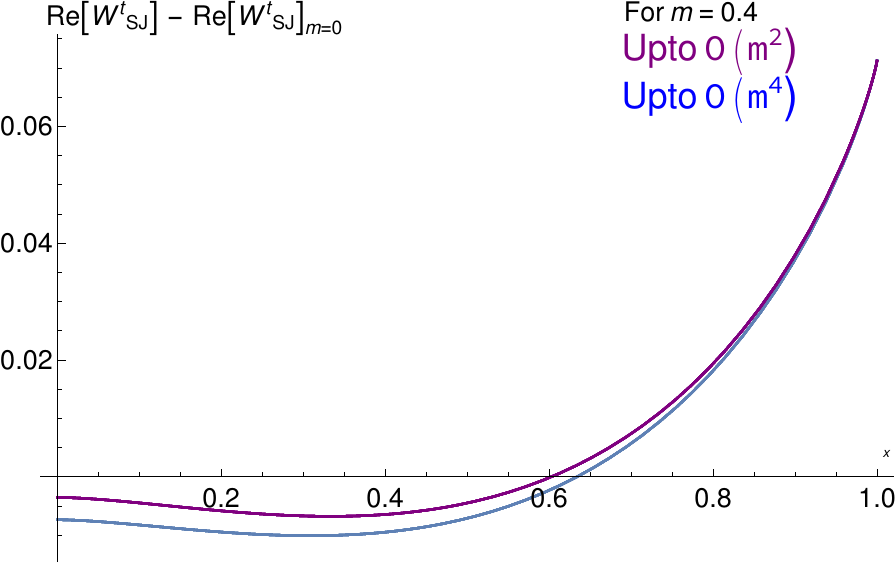}\\
\includegraphics[height=3.5cm]{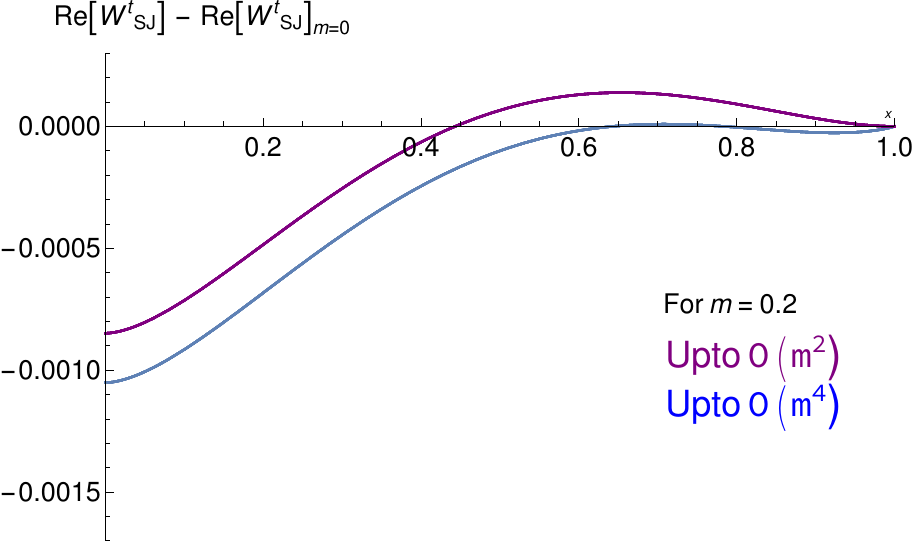}&
\includegraphics[height=3.5cm]{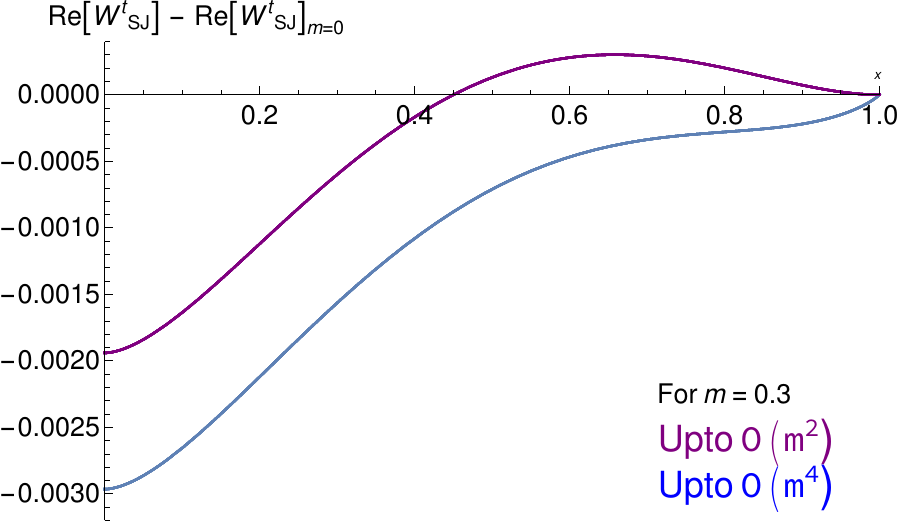}&
\includegraphics[height=3.5cm]{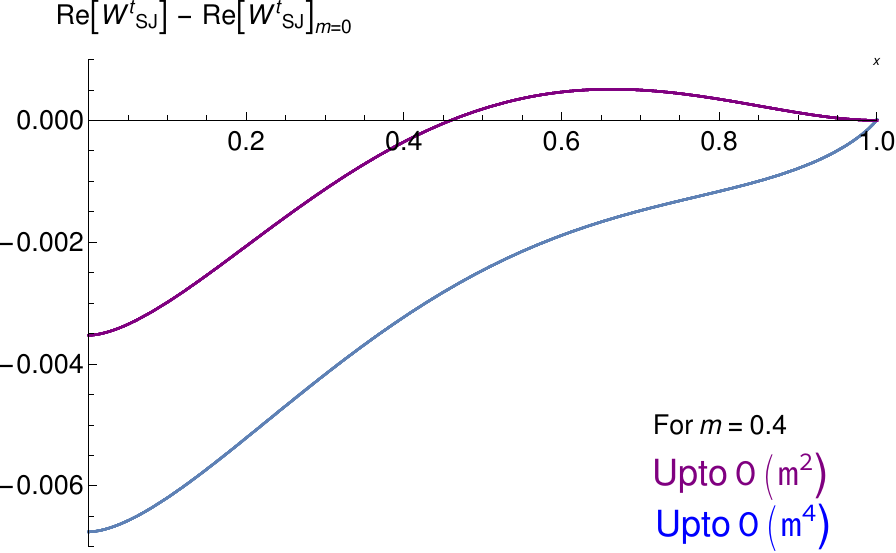}
\end{tabular}}
\caption{Plot of $\re(\wsj^t)-\re(\wsj)_{m=0}$ vs x for $\cO(m^2)$ and $\cO(m^4)$ corrections. The plots in the first line
  are all for timelike separated points while those in the second line are for spacelike separated points.}
\label{fig:diff-re}
\end{figure}
Our calculation suggest that the $\cO(m^4)$ corrections are largely irrelevant to $\wsj$ in the center and the corner
of $\diam$.  A question that occurs is whether increasing the order of the correction makes a significant difference. In Fig.~\ref{fig:diff-re}
we show the sensitivity of  the difference in $\wsj^t$ with $\wmink$,   to $\cO(m^2)$ and $\cO(m^4)$.   As we can see,
the $\cO(m^4)$  corrections while  not negligible, are relatively small for $m\sim 0.2$.

What we have seen from our calculations so far is that in the small mass approximation, $\wsj$  continues to behave in
the center like the massless Minkowski vacuum, and in the corner as the massive Mirror vacuum. This behavior is very
curious since it suggests an unexpected mass dependence in  $\wsj$, not seen in the standard vacuum. In order to explore
this we must examine $\wsj$ for large masses. Because  we are limited in our analytic calculations, we now proceed to
a fully numerical calculation of  $\wsj$ in a causal set for comparison.  

\section{The massive SJ Wightman function in the causal set }\label{sec:causet}

This curious behavior of the SJ vacuum seems to be a result of our small mass approximation. Since we do not know how
to evaluate it analytically for finite mass we look for a numerical evaluation on a causal set $\cc_\cm$ that is approximated
by $\diam$ (see \cite{Bombelli:1987aa,Surya:2019ndm} for an introduction to causal sets). 

$\cc_\cm$ is obtained via a Poisson sprinkling into  $\diam$ at density $\rho$. The expected total number of elements is
then $\langle N \rangle =\rho V_\cm$, where $V_\cm$ is the total volume of the spacetime manifold in which the elements are sprinkled. The partial order is then determined by the causal relation among the elements i.e. $X_i \prec X_j$ iff $X_j$ is in the causal future of $X_i$.

The causal set SJ Wightman function $\wsjc$ is constructed using the same procedure as in the continuum,
namely starting from the  causal set retarded Green function. The massive Green function  in $\diam$ is
\cite{johnston,Johnston:2008za}
\be
G_m=\left(\id+\frac{m^2}{\rho}G_0\right)^{-1}G_0,
\ee
where $\id$ is the $N\times N$ identity matrix and $G_0$ is the massless retarded Green function. Defining the  causal
matrix  $C$ on $\cc_\cm$ as   $C_{ij}=1$ if $X_i \prec X_j$ and $C_{ij}=0$ otherwise, we see that $G_0=C/2$. 

We sprinkle $N=10,000$ elements in $\diam$ of length $2$, i.e., of  density $\rho=2500$ for
$m=0.2,0.4,0.6,0.8,1, 2$ and $10$.   
In Fig.~\ref{fig:dis-ev}  we plot the SJ eigenvalues for these various masses. We find that the eigenvalues for small
masses are very close to the massless eigenvalues, especially for small $n$. As $n$ increases, they become
indistinguishable. In Fig.~\ref{fig:scatterwsjc} we show the scatter plot of $\wsjc$. For the smaller masses, $\wsjc$
tracks the massless case closely, but at larger masses $m\sim 10$ it shows the characteristic behavior expected of the
massive Minkowski  vacuum \cite{Johnston:2009fr}.  
\begin{figure}[h]
\centering{\begin{tabular}{cc}
\includegraphics[height=4cm]{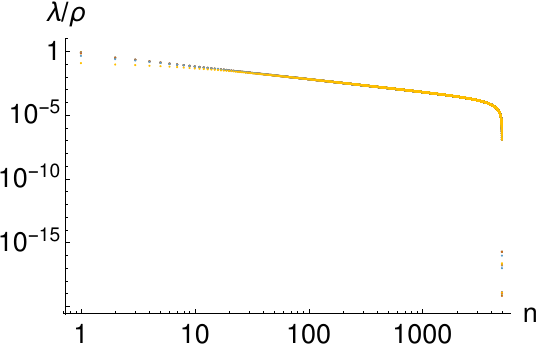}& \hskip 1cm
\includegraphics[height=4cm]{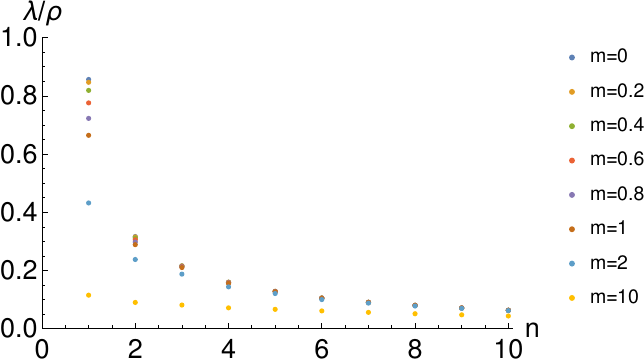}\\
(a)&(b)
\end{tabular}
\caption{(a):A log-log plot of the SJ eigenvalues $\lambda$ divided by density $\rho$ vs $n$ for $m=0,0.2,0.4,0.6,0.8,1,2$ and $10$, (b): a plot of $\lambda/\rho$ vs $n$ for small $n$.}
\label{fig:dis-ev}}
\end{figure}
\vskip 0.5cm
\begin{figure}[h]
\centering{\begin{tabular}{cc}
\includegraphics[height=4cm]{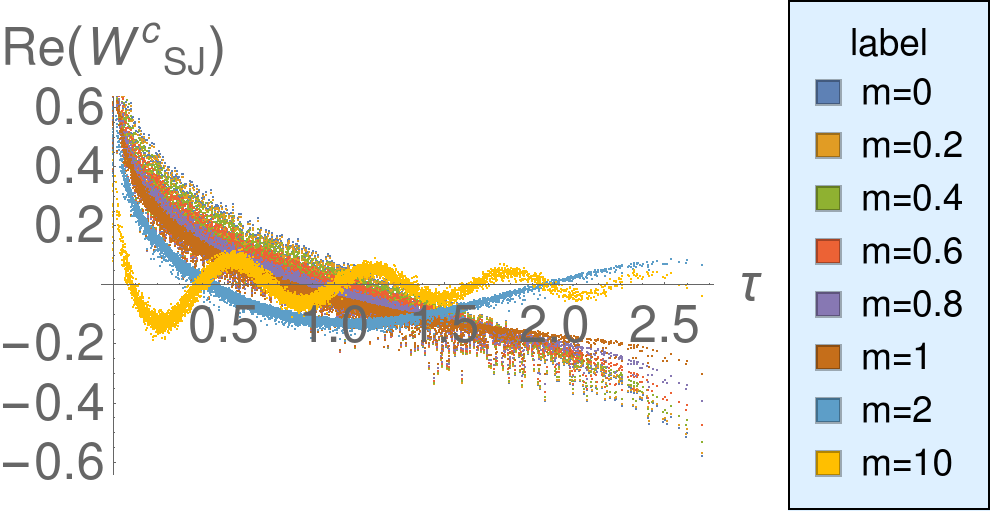}& \hskip 1cm
\includegraphics[height=4cm]{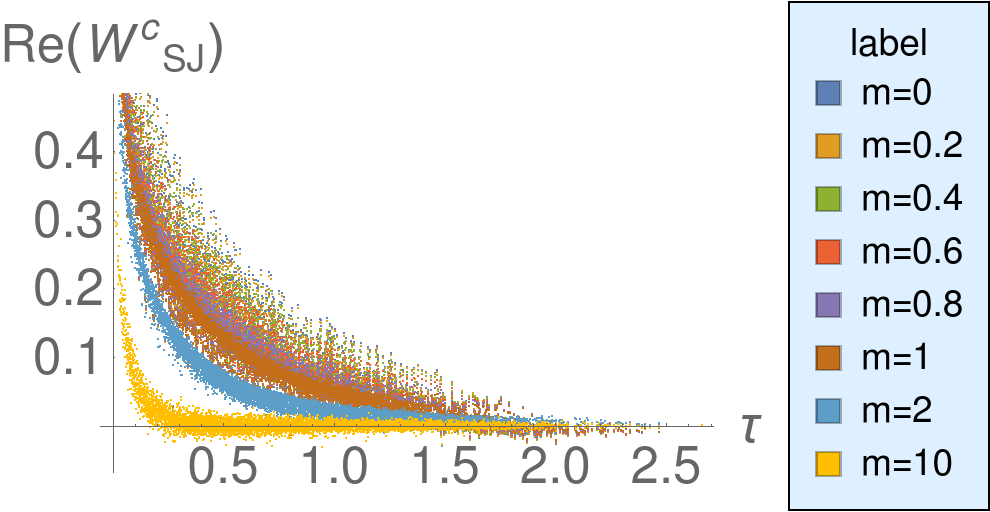}\\
(a)&(b)
\end{tabular}
\caption{$\wsjc$ for $m=0,0.2,0.4,0.6,0.8,1,2$ and $10$ for timelike and spacelike separated points.}
\label{fig:scatterwsjc}}
\end{figure}

Next, we focus our attention to the center of the diamond so that we can compare with our analytic results.  We consider
a central region $\diam_l$ with $l =0.1$. Figs ~\ref{fig:smallmass-sj-tau} and ~\ref{fig:largemass-sj-tau} shows $\wsjc$ vs proper time and proper distance for  timelike
and spacelike separated pairs, respectively for small and large masses.  The comparisons with the massless and massive
Minkowski vacuum show a curious behavior. For the small $m$ values $\wsjc$  agrees perfectly with our analytic results
above, namely that $\wsj$ is more like $\wmink$ than $\wminkm$. However, as $m$ increases,  $\wminkm$ approaches
$\wmink$,  coinciding with it at $m=2\cof$. After this value of $m$,
$\wsjc$ then tracks $\wminkm$ rather than $\wmink$. This transition is continuous, and suggests that the small $m$ behavior
of $\wsjc$ goes continuously over to $\wmink$, unlike $\wminkm$. 

Next we compare $\wsjc $ in the  corner of the diamond with  $\wmirrm$
and $\wrindm$ for all pair of spacetime points in the left corner of the diamond for a range of masses. Instead of
plotting the actual functions, we consider the correlation plot as was done in \cite{Afshordi:2012ez}. 
To generate these plots we considered a small causal diamond in the corner of length $l=0.2$ which contained  118
elements. $\wmirrm$ and $\wrindm$ were calculated for each pair of elements and compared with $\wsjc$ (see Figs.~\ref{fig:sj-mir} and \ref{fig:sj-rind}). In
\cite{Afshordi:2012ez} the IR cut-off $\cof$ was determined from Fig.~\ref{fig:sj-rind} for $m=0$ by setting the
intercept  to zero. We observe that there is much better
correlation between $\wsj$ and $\wmirrm$ as compared to $\wrind$ for all masses which is in
agreement with our analytic calculations.

\section{Discussion}
\label{sec:conclusions}

In this work, we calculated the massive scalar field SJ  modes up to  fourth order of mass. The
procedure we have developed for solving the central  eigenvalue problem can be used in principle  to find the SJ modes for higher
order mass corrections.  
 
Our work shows that $\wsjc$ in the causal set  is compatible with our analytic results in the small
mass regime. The curious behavior of  $\wsjc$ with mass in the center of the diamond suggests a hidden
subtlety in the finite region, ab-initio construction, that has hitherto been missed. In particular, it shows that the massive $\wsj$ in 2D has a well defined massless limit, unlike $\wminkm$. Such a continuous behavior with mass was also seen in the calculation of $\wsjc$ in de Sitter spacetime \cite{Surya:2018byh}. A possible source for this
behavior is that $\wsj$ is built from the advanced/retarded Green functions, which  themselves have a well defined 
 massless limit. It is surprising however  that $\wsj$ for small mass lies in the massless representation of the Poincare
 algebra rather than the expected massive representation. What this means for the uniqueness of the SJ vacuum is unclear
 and we hope to explore this in future work. 
 
In the corner of the diamond, we see that as in the massless case,  $\wsj$ resembles the  massive mirror
vacuum for all masses.  Thus, the expectation  (see \cite{Afshordi:2012ez})  that the massive 
$\wsj$ must be the Rindler vacuum seems to be incorrect.

We end with a broad comment on the SJ formalism. It is possible to construct a $\wsj$ using a different  inner product on
$\cF(M,g)$, instead of the $\mL^2$ inner product adopted in this work.  One way of doing this is to introduce  a non
trivial weight function in the integral. Thus, different choices of inner product give different SJ Wightman functions
even with the same defining conditions (Eqn.~(\ref{eq:wsjconditions})).  As an almost trivial example, in  Appendix
~\ref{sec:rindler}  we show that the choice of inner product can yield the Rindler vacuum in the corner.  In future work
we hope to  explore  this possibility in more detail. 

\begin{figure} 
\centerline{\begin{tabular}{cc}
\includegraphics[height=5cm]{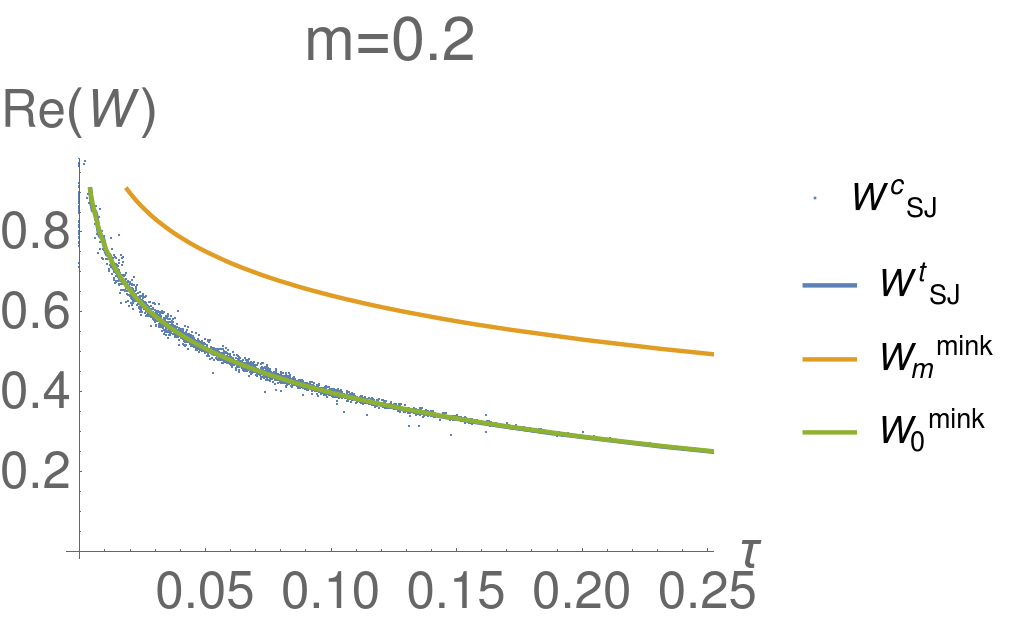}&
\includegraphics[height=5cm]{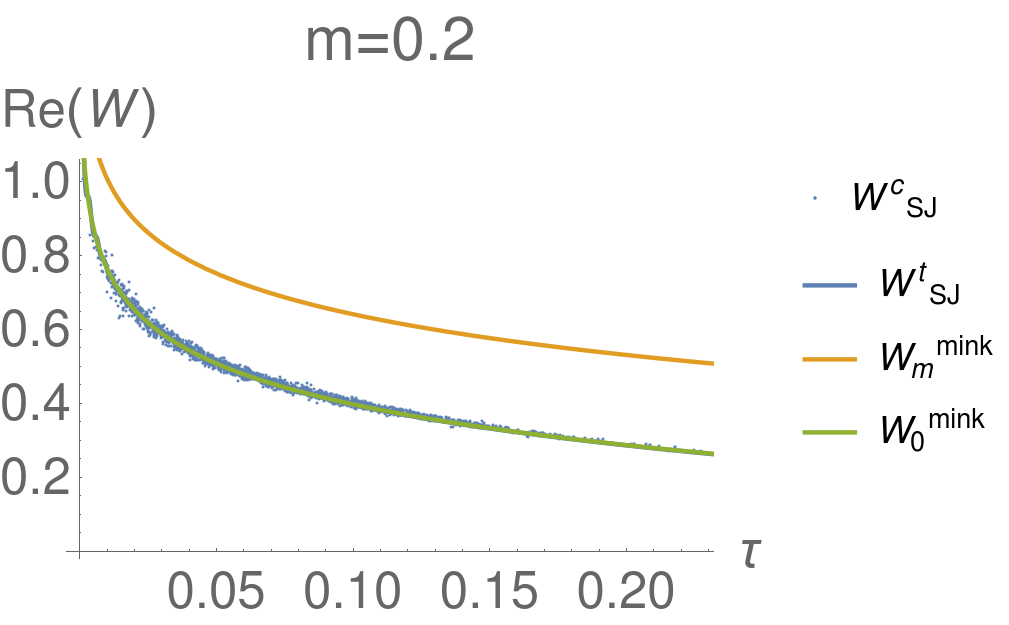}\\
\includegraphics[height=5cm]{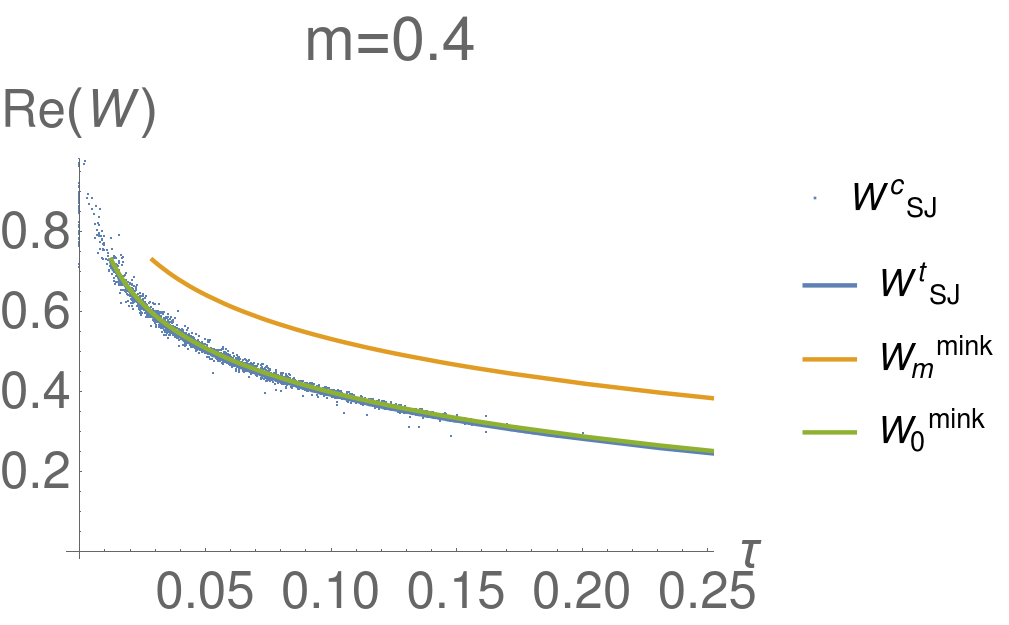}&
 \includegraphics[height=5cm]{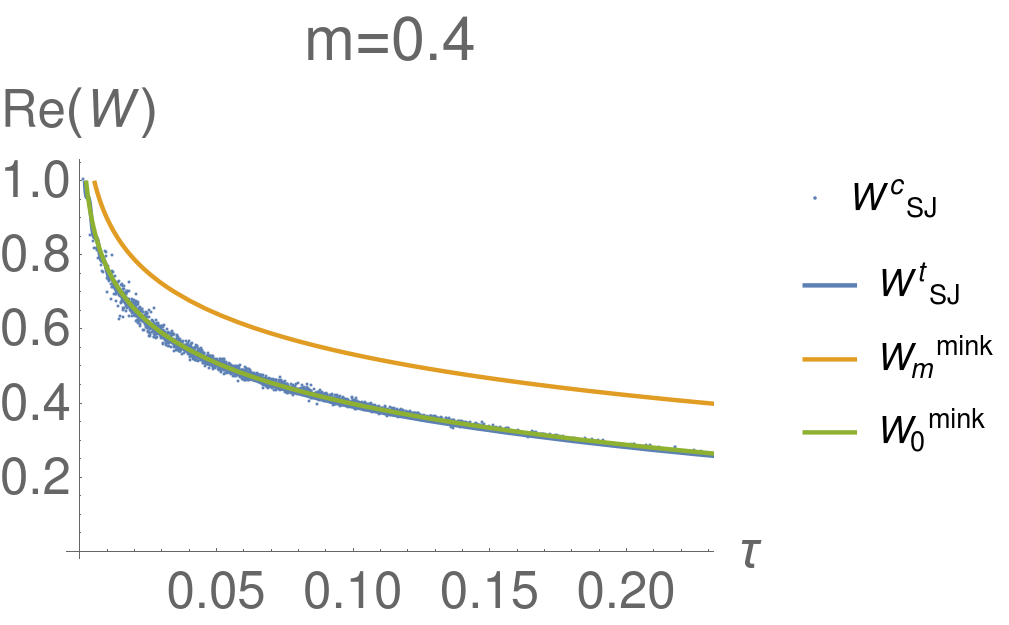}\\
            \end{tabular}}
\caption{$\wsjc$ (blue dots) vs proper time ($\tau$) in the center of the diamond. The plots on the left are for
  timelike separated points and those on the right are for spacelike separated points, for the small mass regime, $m=0.2$ and $ 0.4$. We show
  $\wmink$ (green), $\wminkm$ (orange) and our previous analytic calculation of $\wsj$ (blue
  line). The scatter plot clearly follows the massless green curve for these masses.}
        \label{fig:smallmass-sj-tau}\end{figure}
        \begin{figure} 
\centerline{\begin{tabular}{cc}
\includegraphics[height=5cm]{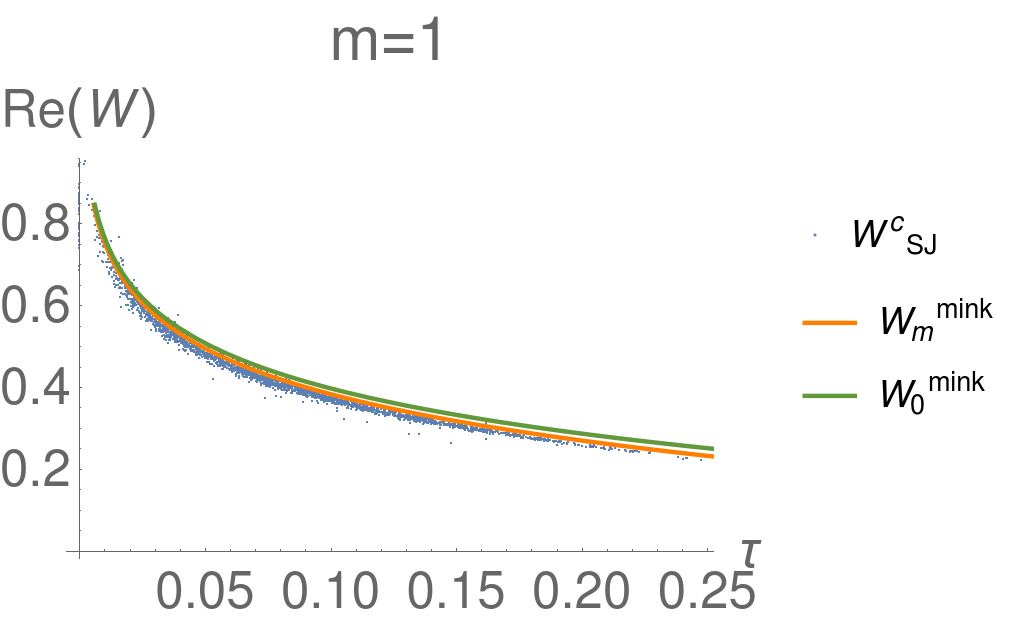}&
\includegraphics[height=5cm]{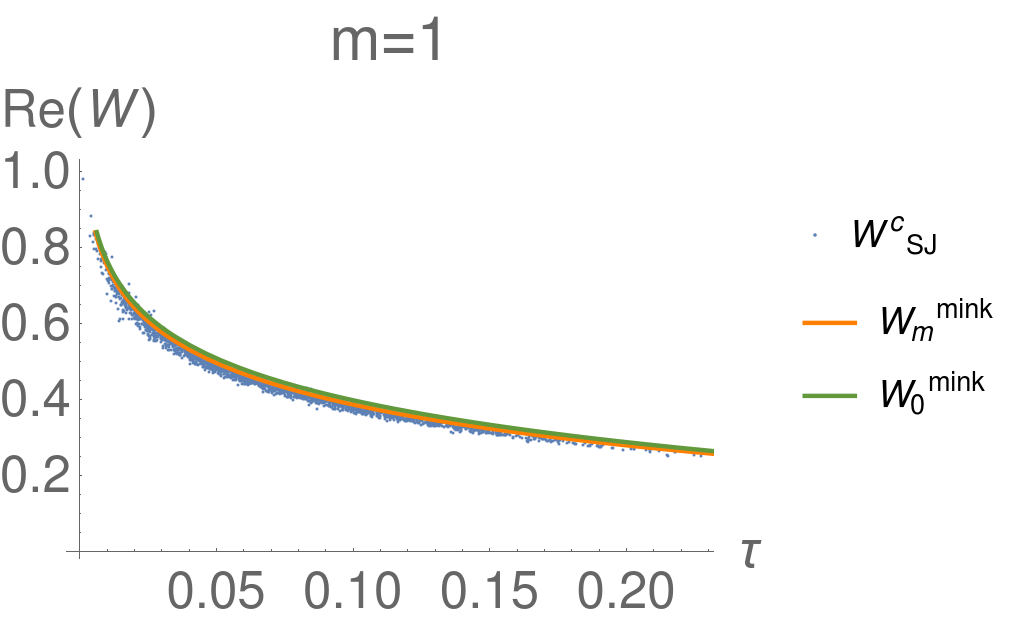}\\
\includegraphics[height=5cm]{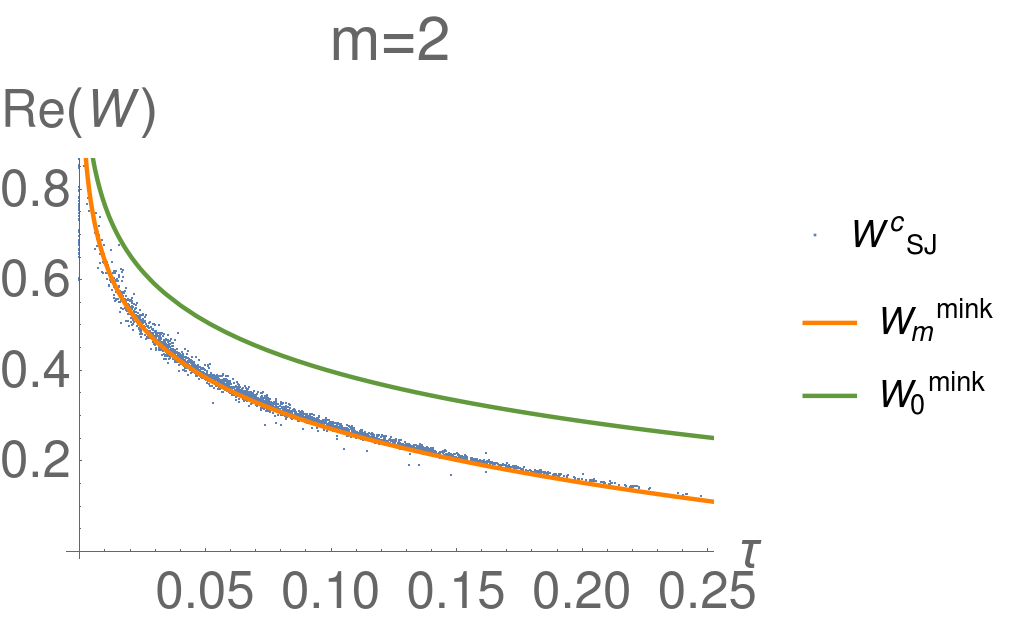}&
\includegraphics[height=5cm]{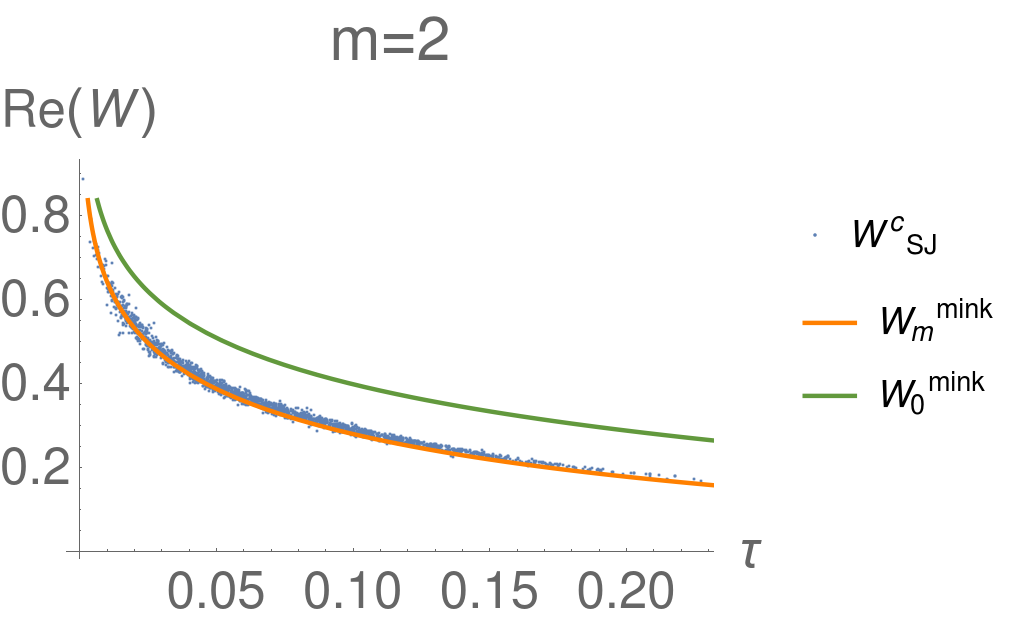}\\
\end{tabular}}
\caption{ The same plots as in Fig \ref{fig:smallmass-sj-tau} but for  $m=1$ and $m=2$. The scatter plot follows 
  the massive orange curve for $m\geq m_c$.}
\label{fig:largemass-sj-tau}
\end{figure}
\vskip 1 cm 
\begin{figure}
\centering{\begin{tabular}{cccc}
\includegraphics[height=4cm]{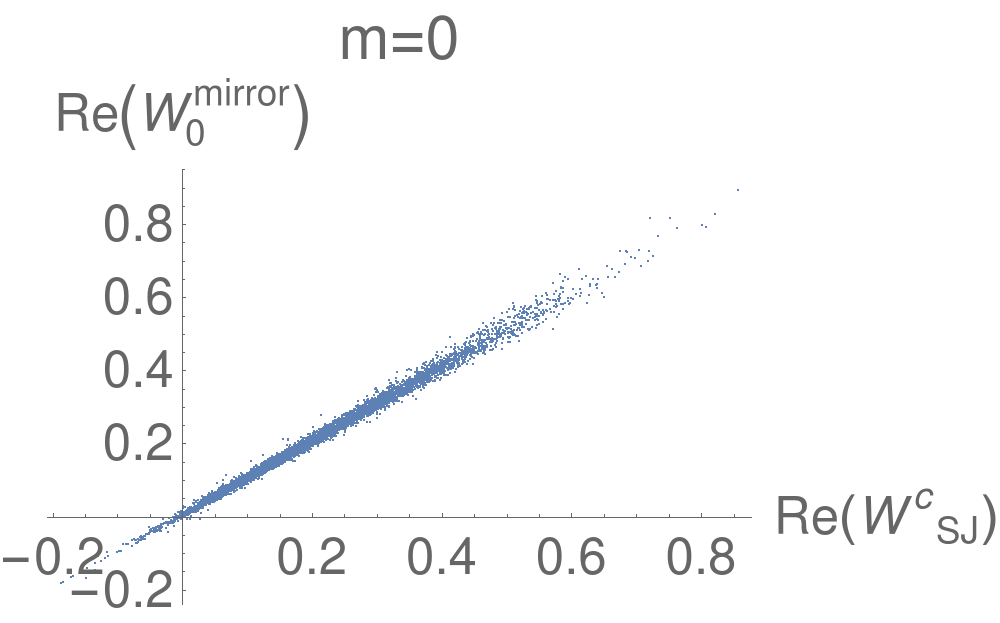}& 
\includegraphics[height=4cm]{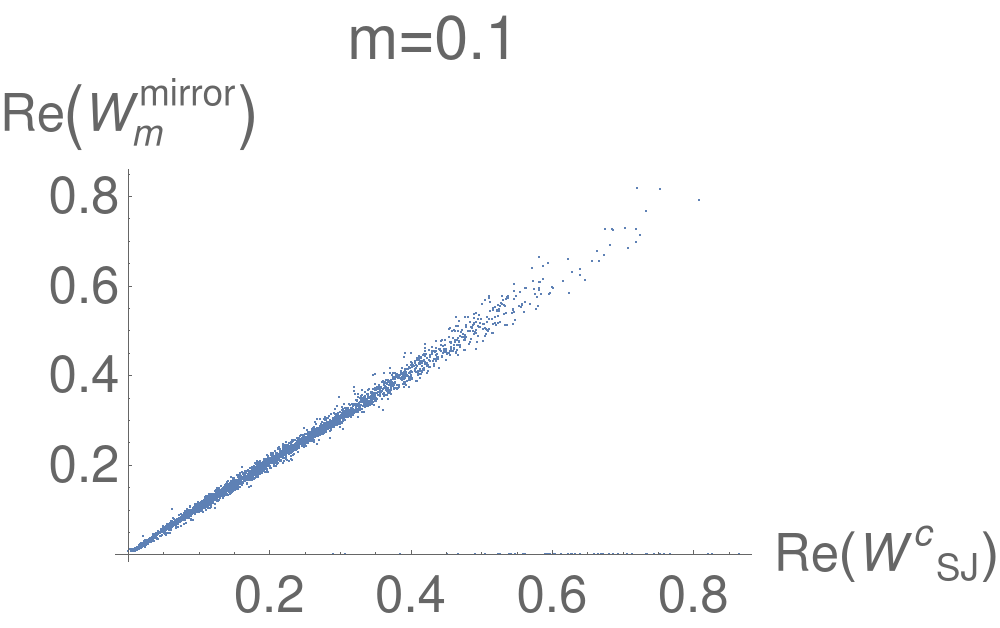} \\ 
\includegraphics[height=4cm]{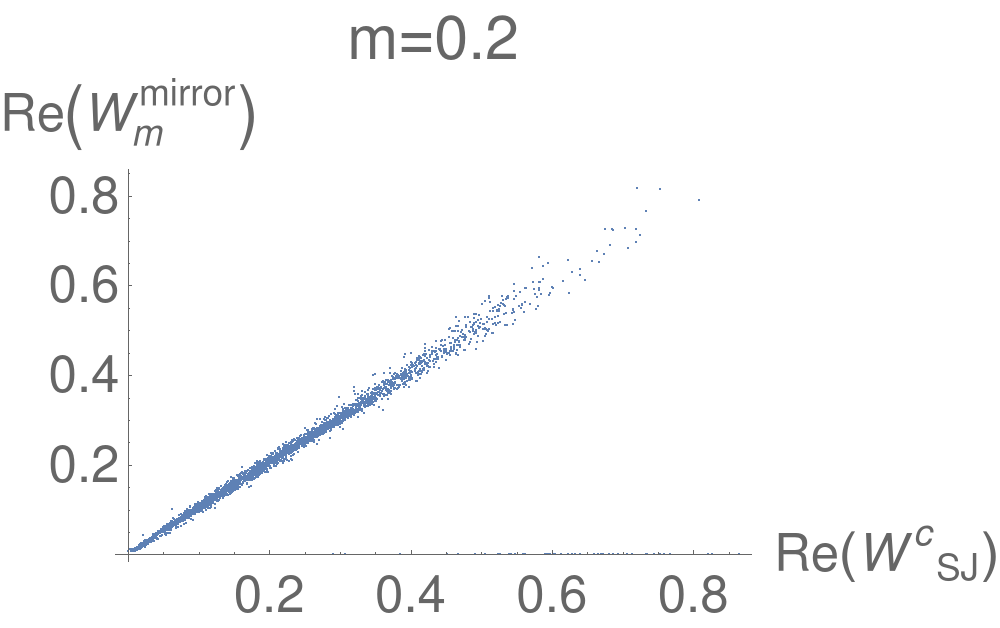}& 
\includegraphics[height=4cm]{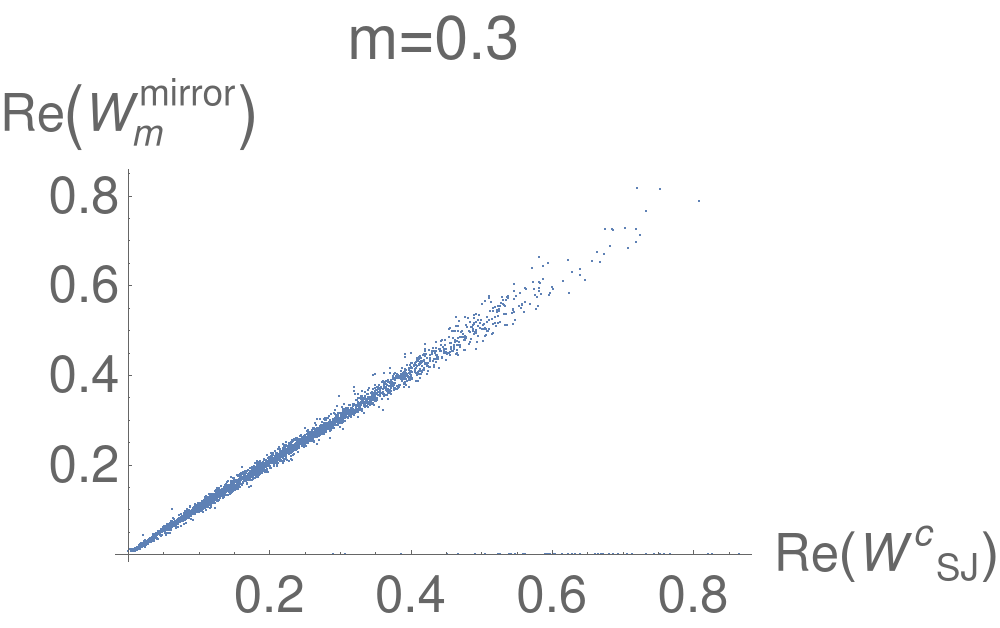}\\
\includegraphics[height=4cm]{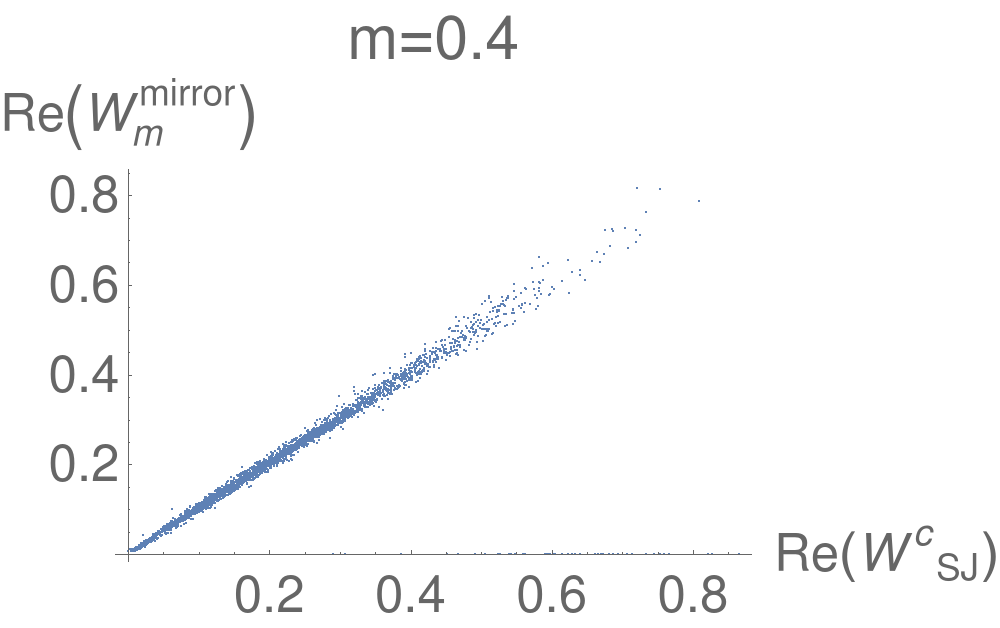}& 
\includegraphics[height=4cm]{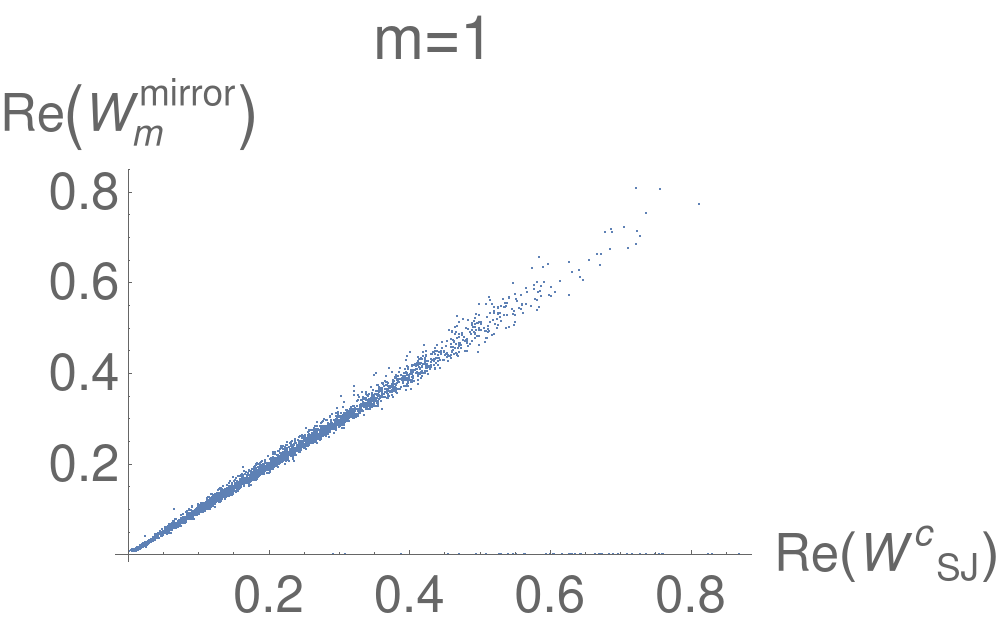}\\ 
\includegraphics[height=4cm]{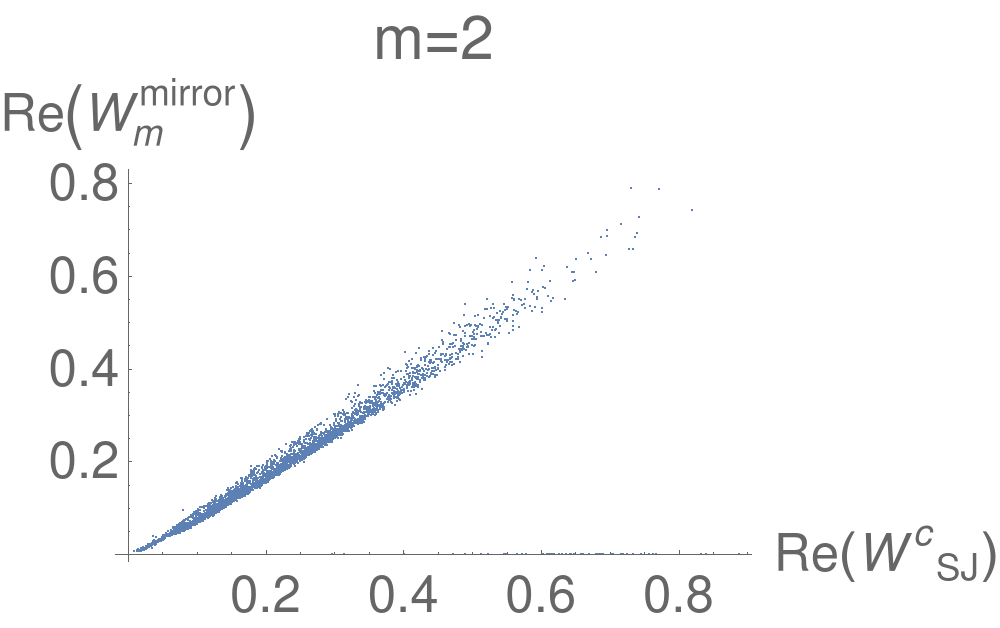}& 
\includegraphics[height=4cm]{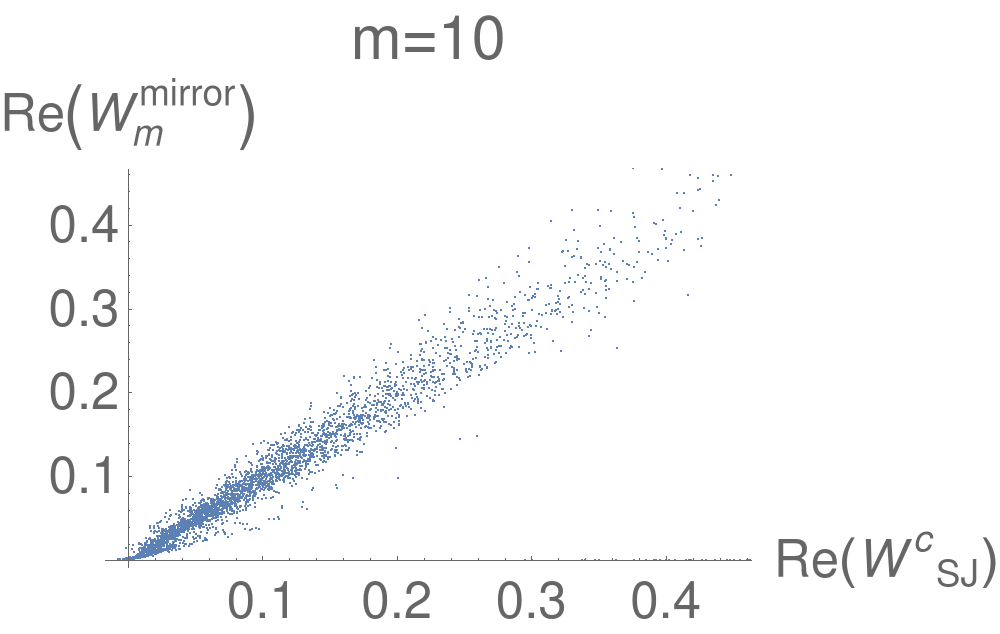}
\end{tabular}
\caption{Correlation plot of  $\wsjc$ vs $\wmirrm$ in the left corner of the diamond for a range of masses.}
\label{fig:sj-mir}}
\end{figure}
\vskip 1 cm 
\begin{figure} 
\centerline{\begin{tabular}{cccc}
\includegraphics[height=4cm]{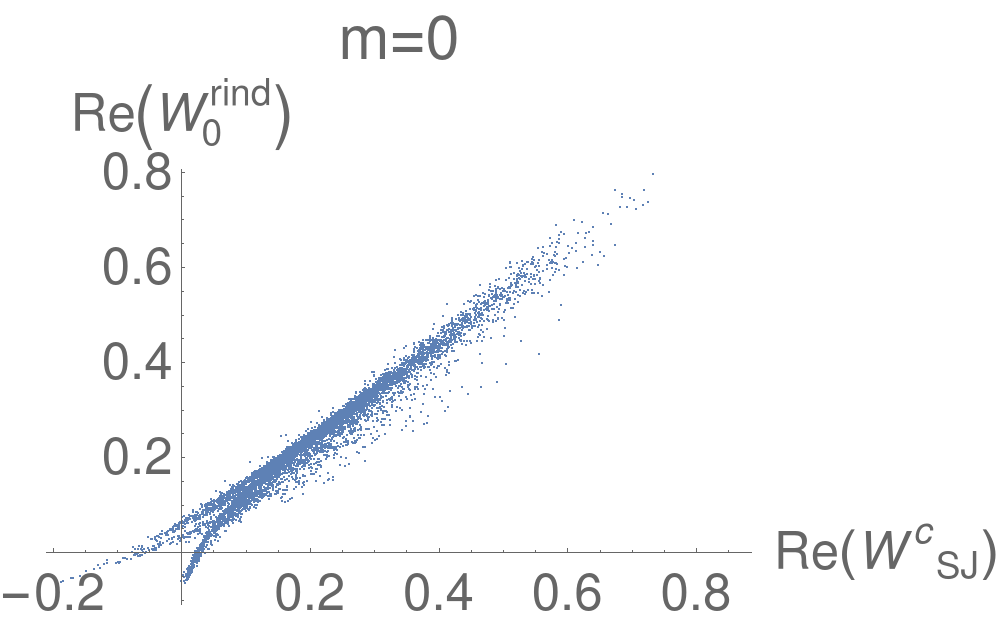}& 
\includegraphics[height=4cm]{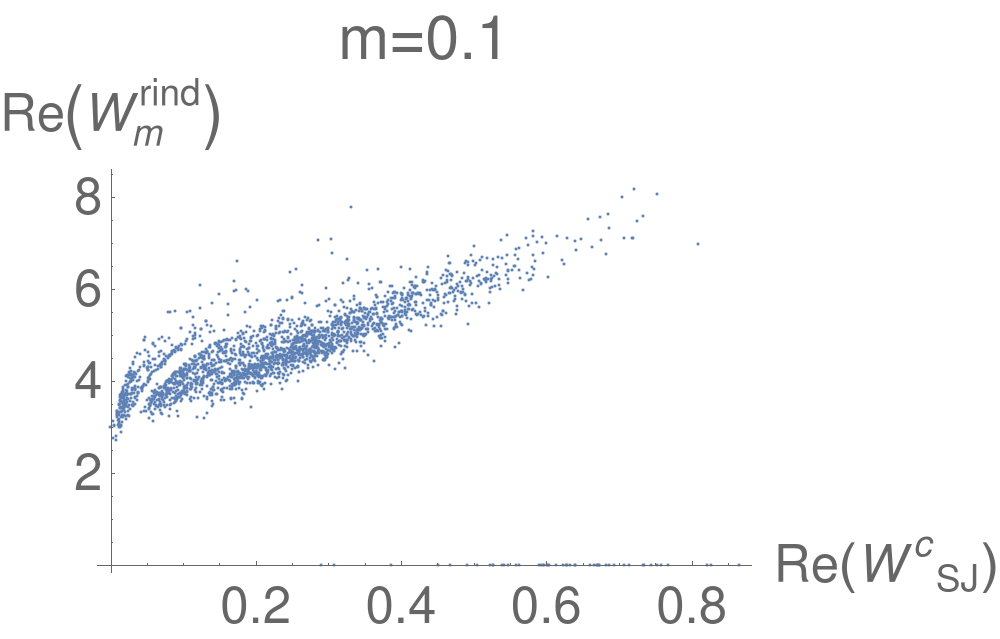}\\
\includegraphics[height=4cm]{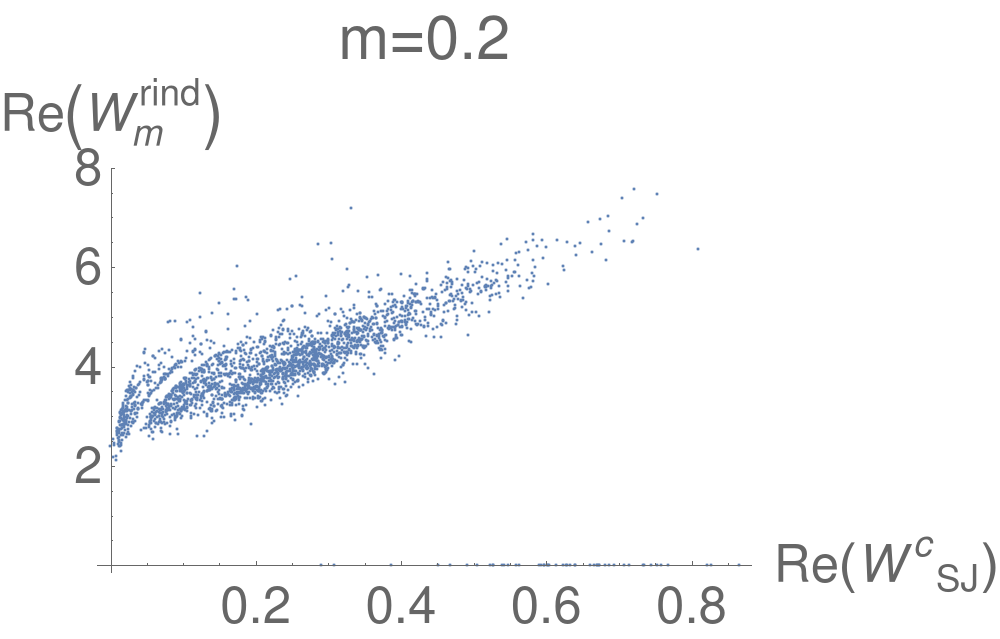}&
\includegraphics[height=4cm]{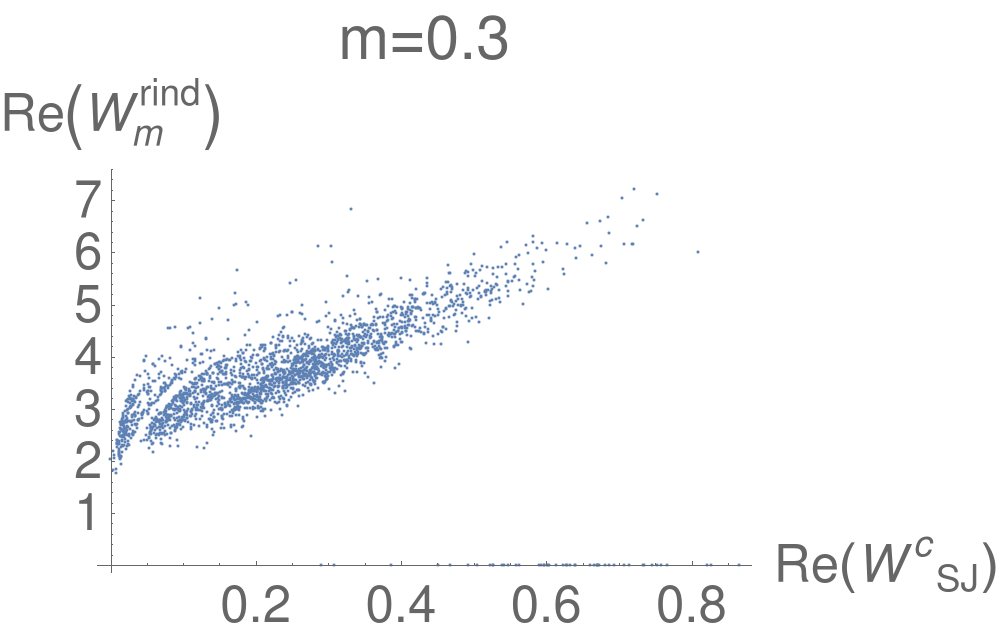}\\
\includegraphics[height=4cm]{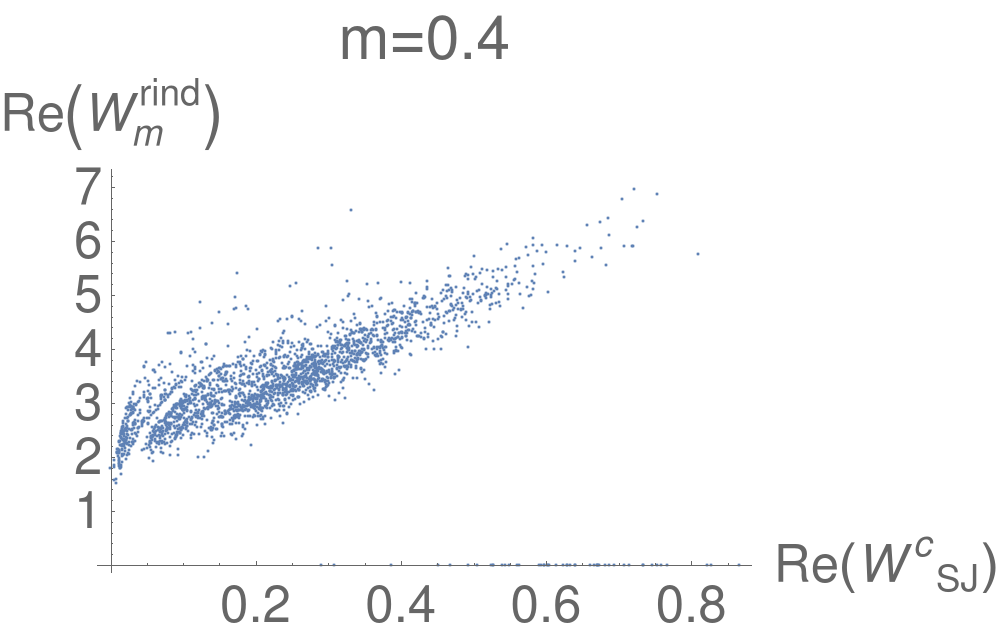}&
\includegraphics[height=4cm]{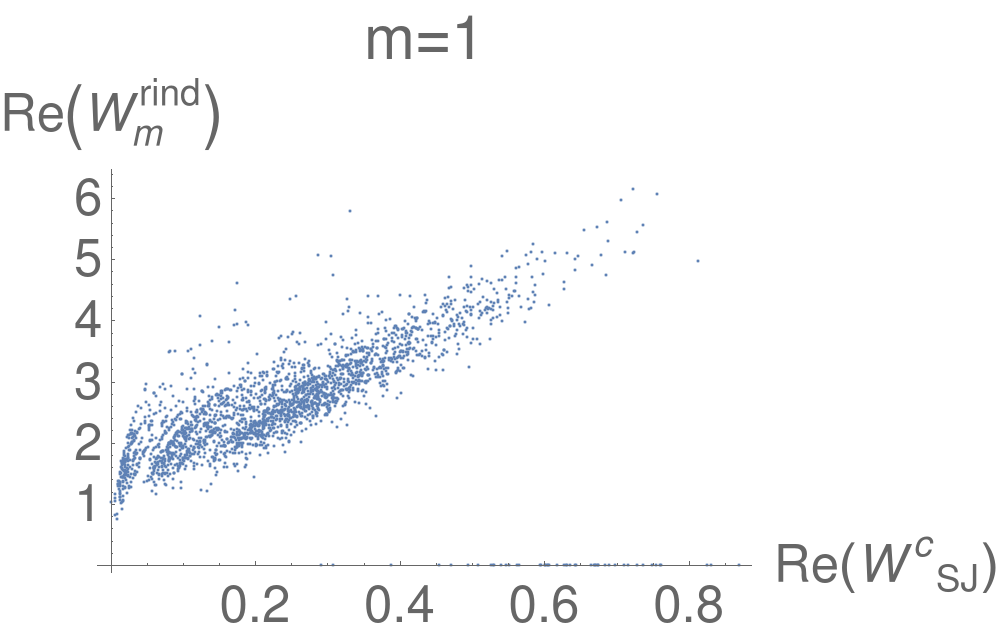}\\
\includegraphics[height=4cm]{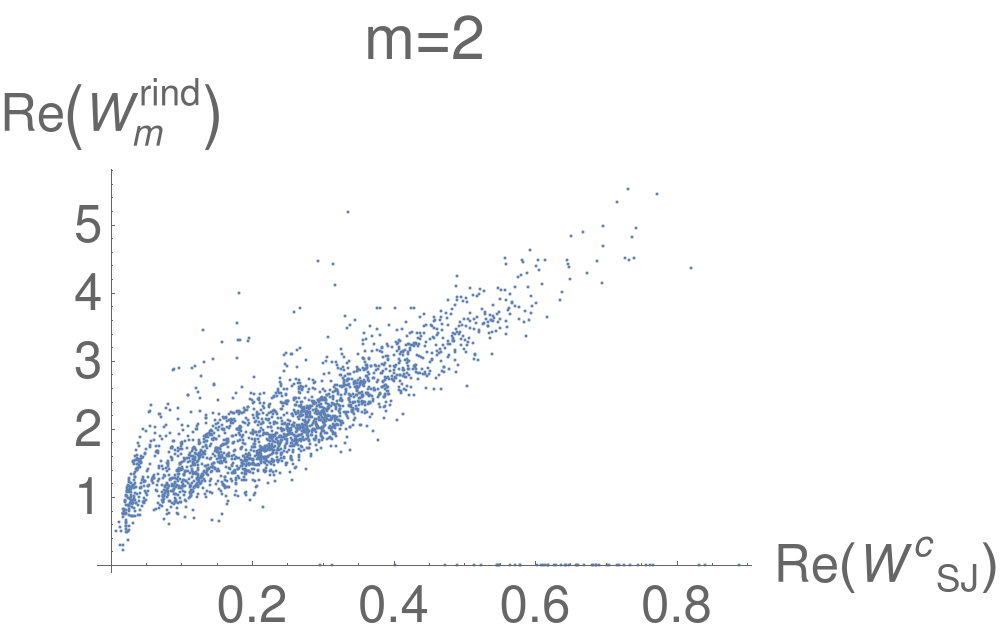}&
\includegraphics[height=4cm]{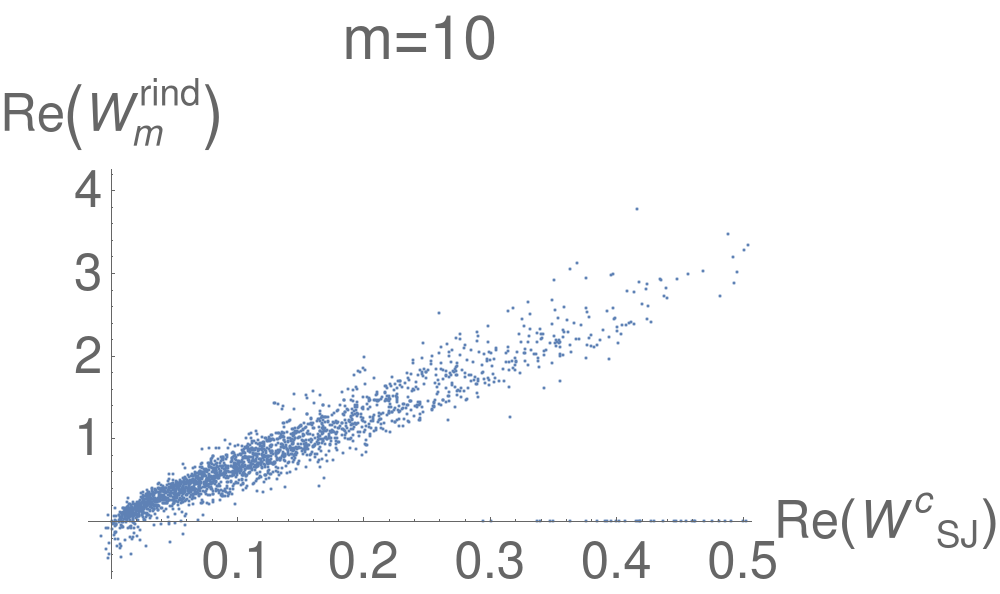}
\end{tabular}}
\caption{Correlation plot of  $\wsjc$ vs $\wrindm$ in the left corner of the diamond for a range of masses.}
\label{fig:sj-rind}
\end{figure}
\section{Acknowledgement}
We would like to thank Nomaan X, Rafael D. Sorkin, Yasaman K. Yazdi, Sujit K. Nath, and Joseph Samuel  for helpful discussions. We would also like to thank S. Vaidya for help with references. During this work  S.S. was supported in part by FQXi-MGA-1510 of the Foundational Questions Institute and an Emmy Noether Fellowship at the Perimeter Institute for Theoretical Physics. S.S. is currently partly supported by a Visiting Fellowship at the Perimeter Institute for Theoretical Physics.

\appendix

\section{Some expressions and derivation of results used in Sec.~\ref{sec:soln}}\label{expressions}

In this appendix we add some of the details of the calculations of Sec.~\ref{sec:soln}. These details include the simplified expression of $F_{ik,n}(u,v)$ and $G_{ik,n}(u,v)$ for $n=0,1,2$ , $Z^{A/S}_l(u,v)$ and $i\hD\circ Z^{A/S}_l(u,v)$, for $l=0,1,2$ and $\PAS_n(u,v)$ for $n=0,1,2$ up to the order in $m^2$, which is required in the calculation of SJ modes up to $\cO(m^4)$. Some details of the calculations of $u^A_k(u,v)$ and $u^S_k(u,v)$ can be found in Appendix \ref{sec:dasym} and \ref{sec:dsym} respectively.

Evaluating $F_{ik,n}(u,v)$ and $G_{ik,n}(u,v)$ defined in Eqn.~(\ref{f and g}) for $n=0,1,2$, we get
\bea
F_{ik,0}(u,v)&=&v, \nonumber\\
F_{ik,1}(u,v)&=&\frac{iv^2}{2k}-\frac{1}{4}(v^2u+2v+u), \nonumber\\
F_{ik,2}(u,v)&=&-\frac{v^3}{8k^2}-\frac{i}{24k}(2v^3u+3v^2-1)+\frac{1}{48}(v^3u^2+v^3+6v^2u+3vu^2+3v+2u). \nonumber
\eea
\bea
G_{ik,0}(u,v)&=&-1, \nonumber\\
G_{ik,1}(u,v)&=&-\frac{iv}{2k}+\frac{1}{4}(v^2+2uv+1), \nonumber\\
G_{ik,2}(u,v)&=&\frac{v^2}{8k^2}+\frac{i}{24k}(2v^3+3uv^2-u)-\frac{1}{48}(2v^3u+3v^2u^2+3v^2+6uv+u^2+1). \nonumber
\eea

Next, we list $Z^A_l(u,v)$ and $Z^S_l(u,v)$ defined in Eqn.~(\ref{eq:serieszl}) and Eqn.~(\ref{eq:symasym}) for $l=0,1,2$ up to the required order of $m^2$.
 \begin{eqnarray} 
Z^A_0(u,v) =0, & \quad & Z^S_0(u,v) \approx 2-m^2uv+\frac{m^4}{8}u^2v^2, \nonumber \\
Z^A_1(u,v) \approx (u-v)-\frac{m^2}{4}uv(u-v), & \quad & Z_1^S(u,v)\approx (u+v)-\frac{m^2}{4}uv(u+v),  \nonumber \\
Z^A_2(u,v)\approx u^2-v^2,& \quad & Z^S_2(u,v) \approx u^2+v^2.
 \end{eqnarray}

Next, we list $i\hD\circ Z^A_l(u,v)$ and $i\hD\circ Z^S_l(u,v)$ for $l=0,1,2$ up to the required order of $m^2$, where $i\hD\circ Z_l(u,v)$ is described in Eqn.~(\ref{eq:ideltaz}) 
\begin{eqnarray} 
  i\hD \circ Z^A_0(u,v)&=&0,   \nonumber \\
  i\hD \circ Z^S_0(u,v) &\approx&  -i\frac{L^2}{24}(u+v) (48-12 m^2(1+uv)+m^4(3+3uv+u^2v^2))),
                                                               \nonumber \\
  i\hD \circ Z^A_1(u,v)&\approx& iL^2\left(-\frac{1}{2}(u^2-v^2)+\frac{m^2}{24}(2uv+1)(u^2-v^2)\right),  \nonumber \\
  i\hD \circ Z^S_1(u,v) &\approx&  iL^2\left(\frac{1}{2}(2-u^2-v^2)-\frac{m^2}{24}\left(6(1+2uv)+(u^2+v^2)(1-2uv)\right)\right),\nonumber \\
  i\hD \circ Z^A_2(u,v) &\approx& \frac{iL^2}{3}\left((u-v)-(u^3-v^3)\right),   \nonumber \\
i\hD \circ Z^S_2(u,v) &\approx& \frac{iL^2}{3}\left((u+v)-(u^3+v^3)\right).
  \end{eqnarray} 

$\PAS_n(u,v)$ defined in Eqn.~(\ref{eq:pas}) for $n=0,1,2$.
\bea
P_0^A(u,v) &=&0, \nonumber\\
P_1^A(u,v) &=& \left(i\left(\frac{1}{2k}-Q^A_1(k)\right)(u-v)-\frac{1}{4}(u^2-v^2)\right),\nonumber\\
P_2^A(u,v) &= & -\frac{u^2-v^2}{8k^2}-\frac{i}{24k}(u-v)(2u^2+2v^2+5uv+1)+\frac{1}{24}(1+uv)(u^2-v^2),
  \nonumber \\ &&+ Q^A_1(k)\left(\frac{u^2-v^2}{2k}+\frac{i}{4}(u-v)(uv+1)\right)-iQ^A_2(k)(u-v) ,
\label{eq:pai}
\eea
\bea
P_0^S(u,v) &=& \frac{}{}-2+2ik(u+v), \nonumber\\
P_1^S(u,v)&=& -\frac{i(u+v)}{2k}+\frac{1}{4}(u^2+v^2+4uv+2)-\left(u^2+v^2+\frac{ik}{2}(uv+3)(u+v)\right)+iQ_1^S(k)(u+v),
\nonumber\\
P_2^S(u,v)&=& \frac{u^2+v^2}{8k^2}+\frac{i}{24k}(u+v)(2u^2+2v^2+uv-1)-\frac{1}{48}((2uv+4)(u^2+v^2)+6v^2u^2+12uv+2)
\nonumber\\ &&
+2k\left(-\frac{i(u^3+v^3)}{8k^2}+\frac{1}{24k}((2uv+3)(u^2+v^2)-2)+\frac{i}{48}(u+v)(u^2v^2+u^2+v^2+8uv+5)\right)\nonumber\\
&&
+Q_{1}^S(k)\left(-\frac{u^2+v^2}{2k}-\frac{i}{4}(uv+3)(u+v)\right)+iQ_{2}^S(k)(u+v), \label{psk}
\eea
where $Q^A_n(k)$ and $Q^S_n(k)$ for $n=1,2$ can be found in Sec.~\ref{sec:dasym} and \ref{sec:dsym} respectively.

\subsection{Details of the calculations for the antisymmetric SJ modes}\label{sec:dasym}
In this section we solve Eqn.~(\ref{eq:hask}) for $H^A_k(u,v)$ by constructing each $m^{2n}P^A_n(u,v)$ out of $Z_l(u,v)$ and $i\Delta\circ Z_l(u,v)$ for different $l$. Let us start with the first non zero $P^A_n(u,v)$. It can be observed that $m^2P^A_1(u,v)$ can be constructed out of $m^2Z^A_1(u,v)$ and $m^2i\Delta\circ Z^A_1(u,v)$ up to $\cO(m^2)$ as
\be
m^2P^A_1(u,v)=\frac{im^2}{2L^2}\left(\frac{L^2}{k}\left(1-2kQ^A_1(k)\right)Z^A_1(u,v)-i\Delta\circ Z^A_1(u,v)\right)
\ee
To make the term in the bracket look like $\left(i\Delta+\frac{L^2}{k}\right)\circ Z^A_1(u,v)$, we fix
\be
Q^A_1(k)=\frac{1}{k}.
\ee
Therefore Eqn.~(\ref{eq:hask}) for $H^A_k(u,v)$ up to $\cO(m^4)$ can be written as
\bea
&&\left(i\Delta+\frac{L^2}{k}\right)\circ\left(H^A_k(u,v)+\frac{im^2\cos(k)}{2k}Z^A_1(u,v)\right)-\frac{m^4L^2\cos(k)}{k}\left(\frac{3(u^2-v^2)}{8k^2}-\frac{i}{12k}(u^3-v^3)\right.\nonumber\\
&&\quad\left.+\frac{5i}{24k}(u-v)+\frac{1}{48}(u^2-v^2)-iQ^A_2(k)(u-v)\right)=0.\nonumber\\ \label{extra term a2}
\eea
In the remaining terms, i.e., the terms which are not yet written as $Z^A_l(u,v)$ or $i\Delta\circ Z^A_l(u,v)$, the highest order of $u$ and $v$ are $u^3$ and $v^3$, which can be identified with $i\Delta\circ Z_2(u,v)$. Therefore we use
\be
-\left(i\Delta+\frac{L^2}{k}\right)\circ\frac{m^4\cos(k)}{4k^2}Z^A_2(u,v)=-\frac{m^4L^2\cos(k)}{k}\left(\frac{i}{12k}(u-v)-\frac{i}{12k}(u^3-v^3)+\frac{1}{4k^2}(u^2-v^2)\right),
\ee
to write Eqn.~(\ref{extra term a2}) as
\bea
\left(i\Delta+\frac{L^2}{k}\right)\circ\left(H^A_k(u,v)+\cos(k)\left(\frac{im^2}{2k}Z^A_1(u,v)-\frac{m^4}{4k^2}Z^A_2(u,v)\right)\right)\nonumber\\
-\frac{m^4L^2\cos(k)}{k}\left(\frac{u^2-v^2}{8k^2}+\frac{i}{8k}(u-v)+\frac{1}{48}(u^2-v^2)-iQ^A_2(k)(u-v)\right)=0. \label{extra term a3}
\eea
The remaining terms in Eqn.~(\ref{extra term a3}) can be written as
\be
\left(i\Delta+\frac{L^2}{k}\right)\circ\left(-\frac{im^4\cos(k)}{24k^3}(6+k^2)Z^A_1(u,v)\right),
\ee
by fixing
\be
Q^A_2(k)=\frac{1}{12k}-\frac{1}{4k^3}.
\ee
Finally Eqn.~(\ref{extra term a3}) can be written as
\be
\left(i\Delta+\frac{L^2}{k}\right)\circ\left(H^A_k(u,v)+\cos(k)\left(\left(\frac{im^2}{2k}-\frac{im^4(6+k^2)}{24k^3}\right)Z^A_1(u,v)-\frac{m^4}{4k^2}Z^A_2(u,v)\right)\right)=0
\ee
which implies that
\be
u^A_k(u,v)=U^A_{ik}(u,v)-\cos(k)\left(\left(\frac{im^2}{2k}-\frac{im^4(6+k^2)}{24k^3}\right)Z^A_1(u,v)-\frac{m^4}{4k^2}Z^A_2(u,v)\right)+\cO(m^6)
\ee
with eigenvalue $-\frac{L^2}{k}$, where $k$ satisfies
\be
\sin(k)=\left(\frac{m^2}{k}+\frac{m^4}{12k}\left(1-\frac{3}{k^2}\right)\right)\cos(k)+\cO(m^6)\label{f cond.}
\ee
\subsection{Details of the calculations for the symmetric SJ modes}
\label{sec:dsym}
In this section we solve Eqn.~(\ref{eq:hask}) for $H^S_k(u,v)$ by constructing each $m^{2n}P^S_n(u,v)$ out of $Z_l(u,v)$ and $i\Delta\circ Z_l(u,v)$ for different $l$. Let us start with the first non zero $P^S_n(u,v)$. It can be observed that $P^S_0(u,v)$ can be constructed out of $Z^S_0(u,v)$ and $i\Delta\circ Z^S_0(u,v)$ up to $\cO(m^0)$ as
\be
P^S_0(u,v)=\left(i\Delta+\frac{L^2}{k}\right)\circ\left(-\frac{k}{L^2}Z^S_0(u,v)\right).
\ee
Therefore Eqn.~(\ref{eq:hask}) for $H^S_k(u,v)$ up to $\cO(m^4)$ can be written as
\bea
&&\left(i\Delta+\frac{L^2}{k}\right)\circ\left(H^S_k(u,v)+Z^S_0(u,v)\cos(k)\right)-\frac{L^2\cos(k)}{k}\left(m^2\left(-\frac{3}{4}(u^2+v^2)\right.\right.\nonumber\\
&&\left.\left.+i\left(Q^S_1(k)-k-\frac{1}{2k}\right)(u+v)+\frac{1}{2}\right)+m^4\left(\frac{u^2+v^2}{8k^2}+\frac{i}{24k}(u+v)(2u^2+2v^2+uv-1)\right.\right.\nonumber\\
&&\left.\left.-\frac{1}{24}((-3uv-4)(u^2+v^2)+6uv+5)+\left(-\frac{i(u^3+v^3)}{4k}+\frac{ik}{24}(u+v)(u^2+v^2+5uv+2)\right)\right.\right.\nonumber\\
&&\left.\left.+Q^S_1(k)\left(-\frac{u^2+v^2}{2k}-\frac{i}{4}(uv+3)(u+v)\right)+iQ^S_2(k)(u+v)\right)\right)=0. \label{extra term s2}
\eea
Since the extra terms in Eqn.~(\ref{extra term s2}) has $m^2$ as a factor, we need to look for $Z^S_l$ and $i\Delta\circ Z^S_l$ only up to $\cO(m^2)$. $\cO(m^2)$ terms in Eqn.~(\ref{extra term s2}) can be written in terms of $\left(i\Delta+\frac{L^2}{k}\right)\circ Z^S_0(u,v)$ and  $\left(i\Delta+\frac{L^2}{k}\right)\circ Z^S_1(u,v)$ for
\be
Q^S_1=2k-\frac{1}{k}
\ee
as
\be
\left(i\Delta+\frac{L^2}{k}\right)\circ m^2\cos(k)\left(\frac{3i}{2k}Z^S_1(u,v)+\frac{1}{2}Z^S_0(u,v)\right).
\ee
Therefore Eqn.~(\ref{extra term s2}) can further be written as
\bea
\left(i\Delta+\frac{L^2}{k}\right)\circ\left(H^S_k(u,v)+\cos(k)\left(\left(1+\frac{m^2}{2}\right)Z^S_0(u,v)+\frac{3im^2}{2k}Z^S_1(u,v)\right)\right)+\frac{im^4L^2\cos(k)}{48k^3}\left(8ik^2\right.\nonumber\\
\left.+k\left(-34-kQ^S_2(k)+56k^2\right)(u+v)+i(30-37k^2)(u^2+v^2)+2k(4-k^2)(u^3+v^3)\right)=0. \label{extra term s3}
\eea
Remaining $\cO(m^4)$ terms in Eqn.~(\ref{extra term s3}) can be written in terms of $\left(i\Delta+\frac{L^2}{k}\right)\circ Z^S_0(u,v)$, $\left(i\Delta+\frac{L^2}{k}\right)\circ Z^S_1(u,v)$, $\left(i\Delta+\frac{L^2}{k}\right)\circ Z^S_2(u,v)$ for
\be
Q^S_2(k)=\frac{3-29k^2+28k^4}{12k^3}
\ee
as
\be
-\frac{m^4\cos(k)}{8k^2}\left((4-k^2)Z^S_2(u,v)+\frac{i(6-31k^2)}{3k}Z^S_1(u,v)+(2-9k^2)Z^S_0(u,v)\right).
\ee
Hence Eqn.~(\ref{extra term s3}) can be written as
\bea
\left(i\Delta+\frac{L^2}{k}\right)\circ\left(H^S_k(u,v)+\cos(k)\left(\left(1+\frac{m^2}{2}-\frac{m^4}{8k^2}(2-9k^2)\right)Z^S_0(u,v)\right.\right.\nonumber\\
\left.\left.+\left(\frac{3im^2}{2k}-\frac{im^4}{24k^3}(6-31k^2)\right)Z^S_1(u,v)-\frac{m^4}{8k^2}(4-k^2)Z^S_2(u,v)\right)\right)=0.
\eea
Therefore the symmetric SJ modes are
\bea
u^S_k(u,v)&=&U^S_{ik}(u,v)-\cos(k)\left(\left(1+\frac{m^2}{2}-\frac{m^4}{8k^2}(2-9k^2)\right)Z^S_0(u,v)\right.\nonumber\\
&&\left.+\left(\frac{3im^2}{2k}-\frac{im^4}{24k^3}(6-31k^2)\right)Z^S_1(u,v)-\frac{m^4}{8k^2}(4-k^2)Z^S_2(u,v)\right)+\cO(m^4),
\eea
with eigenvalue $-\frac{L^2}{k}$, where $k$ satisfies
\be
\sin(k)=\left(2k-\frac{m^2}{k}(1-2k^2)+\frac{m^4}{12k^3}(3-29k^2+28k^4)\right)\cos(k)+\cO(m^4). \label{g cond.}
\ee

\section{Summation of series with inverse powers of roots of a transcendental equation} \label{app:speigel}
In this appendix we make use of the work of \cite{speigel} to evaluate the series (Eqn.~(\ref{eq:ser1}) and Eqn.~(\ref{eq:ser2})), which involves the roots of the transcendental equation (Eqn.~(\ref{g condition m0})). They are used in Sec(\ref{sec:compl}) to determine the completeness of the SJ modes

Let us start with a brief discussion on the work of \cite{speigel}. Consider a transcendental equation of the form
\be
S(x) \equiv 1+\sum_{n=1}^\infty a_n x^n=0 \label{eqn:gentr}
\ee
with $x_1,x_2,x_3\dots$ as its roots, which means the equation can be factorized as
\be
\left(1-\frac{x}{x_1}\right)\left(1-\frac{x}{x_2}\right)\left(1-\frac{x}{x_3}\right)\dots=0 \label{eqn:gentrfac}
\ee
On comparing Eqn.~(\ref{eqn:gentr} and \ref{eqn:gentrfac}), we find that
\be
a_1= \sum_{i=1}^\infty \frac{1}{x_i}, \quad
a_2= \sum_{i<j} \frac{1}{x_ix_j}, \quad
a_3= \sum_{i<j<k} \frac{1}{x_ix_jx_k} 
\ee
and so on. It is straight forward to see that
\be
\sum_{i=1}^\infty\left(\frac{1}{x_i}\right)^2=\left(\sum_{i=1}^\infty \frac{1}{x_i}\right)^2-2\sum_{i<j} \frac{1}{x_ix_j}=a_1^2-2a_2
\ee
and similarly
\be
\sum_{i=1}^\infty\left(\frac{1}{x_i}\right)^3=3a_1a_2-3a_3-a_1^3.
\ee
Similarly we can get the sum of higher inverse powers of the roots.

Now let us come to the equation of our interest i.e. Eqn.~(\ref{g condition m0}), which on series expansion becomes
\be
S(k^2)\equiv 1-\left(1-\frac{1}{3!}\right)k^2+\left(\frac{2}{4!}-\frac{1}{5!}\right)k^4-\left(\frac{2}{6!}-\frac{1}{7!}\right)k^6\dots=0.\label{eq:expand}
\ee
The roots of Eqn.~(\ref{eq:expand}) are $k^S_0\in\cK_g$, and therefore
\bea
\sum_{k^S_0\in\cK_g}\frac{1}{{k^S_0}^2}&= &a_1 =\frac{5}{6},\nonumber\\
\sum_{k^S_0\in\cK_g}\frac{1}{{k^S_0}^4}&= &a_1^2-2a_2 =\frac{49}{90},\nonumber\\
\sum_{k^S_0\in\cK_g}\frac{1}{{k^S_0}^6}&= &3a_1a_2-3a_3-a_1^3 =\frac{377}{945}.
\eea
We are also interested in the series involving the inverse power of $4{k^S_0}^2-1$, where $k^S_0\in\cK_g$. We start with finding an equation whose solutions are given by $4{k^S_0}^2-1$. If ${k^S_0}^2$ are the solutions of $S(k^2)=0$, then $4{k^S_0}^2-1$ are the solutions of $S\left(\frac{k^2+1}{4}\right)=0$.
\be
S\left(\frac{k^2+1}{4}\right) \equiv 1-\frac{1}{4}k^2+\frac{5\cos(1/2)-9\sin(1/2)}{32\left(\cos(1/2)-\sin(1/2)\right)}k^4-\frac{53\cos(1/2)-97\sin(1/2)}{384\left(\cos(1/2)-\sin(1/2)\right)}k^6\dots=0.
\ee
Using the same method as above, we find
\bea
\sum_{k^S_0\in\cK_g}\frac{1}{4{k^S_0}^2-1}&=&\frac{1}{4},\\
\sum_{k^S_0\in\cK_g}\frac{1}{(4{k^S_0}^2-1)^2}&=&-\frac{1}{4}\left(\frac{\cos(1/2)-2\sin(1/2)}{\cos(1/2)-\sin(1/2)}\right),\\
\sum_{k_0\in\cK_g}\frac{1}{(4{k^S_0}^2-1)^3}&=&\frac{1}{64}\left(1+\frac{19\cos(1/2)-35\sin(1/2)}{\cos(1/2)-\sin(1/2)}\right).
\eea

\section{Some expressions used in Sec.~\ref{sec:wightmann}}\label{sec:wight-app}

Here we list the expressions of $\wa,\waa,\waaa$ and $\waaaa$ defined in Eqn.~(\ref{eq:aone}) in terms of Polylogarithms.
\bea
\wa&\equiv& \sum_{n=1}^\infty\frac{1}{8n\pi}\left(1-\frac{2m^2}{n^2\pi^2}+\frac{m^4}{n^2\pi^2}\left(\frac{7}{n^2\pi^2}-\frac{1}{6}\right)\right)\left(e^{-in\pi u}-e^{-in\pi v}\right)\left(e^{in\pi u'}-e^{in\pi v'}\right)\nonumber\\
&=&\frac{1}{8\pi}\left[\li_1\left(e^{-i\pi(u-u')}\right)+\li_1\left(e^{-i\pi(v-v')}\right)-\li_1\left(e^{-i\pi(u-v')}\right)-\li_1\left(e^{-i\pi(v-u')}\right)\right]\nonumber\\
&&-\frac{m^2}{4\pi^3}\left(1+\frac{m^2}{12}\right)\left[\li_3\left(e^{-i\pi(u-u')}\right)+\li_3\left(e^{-i\pi(v-v')}\right)-\li_3\left(e^{-i\pi(u-v')}\right)-\li_3\left(e^{-i\pi(v-u')}\right)\right]\nonumber\\
&&+\frac{7m^4}{8\pi^5}\left[\li_5\left(e^{-i\pi(u-u')}\right)+\li_5\left(e^{-i\pi(v-v')}\right)-\li_5\left(e^{-i\pi(u-v')}\right)-\li_5\left(e^{-i\pi(v-u')}\right)\right],
\eea
\bea
\waa&\equiv& \sum_{n=1}^\infty\frac{1}{8n\pi}\left(1-\frac{2m^2}{n^2\pi^2}\right)\left(e^{-in\pi u}-e^{-in\pi v}\right)\Psi_A^*(n,u',v')\nonumber\\
&=& \frac{1}{8\pi} \sum_{j=1}^3 f_j^*(m;u',v')\left[\li_{j+1}\left(-e^{-i\pi u}\right)-\li_{j+1}\left(-e^{-i\pi v}\right)\right]+\frac{im^4}{8\pi^4}(u'-v')\left[\li_{4}\left(-e^{-i\pi u}\right)-\li_{4}\left(-e^{-i\pi v}\right)\right]\nonumber\\
&&+\frac{1}{8\pi}\sum_{j=1}^3\left( g_j^*(m;u',v')\left[\li_{j+1}\left(e^{-i\pi(u-u')}\right)-\li_{j+1}\left(-e^{-i\pi(v-u')}\right)\right]\right.\nonumber\\
&&\left.\quad\quad\quad-g_j^*(m;v',u')\left[\li_{j+1}\left(-e^{-i\pi(u-v')}\right)-\li_{j+1}\left(-e^{-i\pi(v-v')}\right)\right]\right)\nonumber\\
&&-\frac{im^4}{8\pi^4}\left((2u'+v')\left[\li_{4}\left(-e^{-i\pi(u-u')}\right)-\li_{4}\left(-e^{-i\pi(v-u')}\right)\right]\right.\nonumber\\
&&\left.\quad\quad\quad-(2v'+u')\left[\li_{4}\left(-e^{-i\pi(u-v')}\right)-\li_{4}\left(-e^{-i\pi(v-v')}\right)\right]\right),
\eea
\bea
\waaa&\equiv& \sum_{n=1}^\infty\frac{1}{8n\pi}\left(1-\frac{2m^2}{n^2\pi^2}\right)\Psi_A(n,u,v)\left(e^{in\pi u'}-e^{in\pi v'}\right)\nonumber\\
&=& \frac{1}{8\pi} \sum_{j=1}^3 f_j(m;u,v)\left[\li_{j+1}\left(-e^{i\pi u'}\right)-\li_{j+1}\left(-e^{i\pi v'}\right)\right]-\frac{im^4}{8\pi^4}(u-v)\left[\li_{4}\left(-e^{i\pi u'}\right)-\li_{4}\left(-e^{i\pi v'}\right)\right]\nonumber\\
&&+\frac{1}{8\pi}\sum_{j=1}^3\left( g_j(m;u,v)\left[\li_{j+1}\left(e^{-i\pi(u-u')}\right)-\li_{j+1}\left(-e^{-i\pi(v-u')}\right)\right]\right.\nonumber\\
&&\left.\quad\quad\quad-g_j(m;v,u)\left[\li_{j+1}\left(-e^{-i\pi(u-v')}\right)-\li_{j+1}\left(-e^{-i\pi(v-v')}\right)\right]\right)\nonumber\\
&&+\frac{im^4}{8\pi^4}\left((2u+v)\left[\li_{4}\left(-e^{-i\pi(u-u')}\right)-\li_{4}\left(-e^{-i\pi(v-u')}\right)\right]\right.\nonumber\\
&&\left.\quad\quad\quad-(2v+u)\left[\li_{4}\left(-e^{-i\pi(u-v')}\right)-\li_{4}\left(-e^{-i\pi(v-v')}\right)\right]\right),
\eea
\bea
\waaaa&\equiv& \sum_{n=1}^\infty\frac{1}{8n\pi}\Psi_A(n,u,v)\Psi_A^*(n,u',v')\nonumber\\
&=&\frac{m^4}{32\pi^3}\left[\zeta(3)(u-v)(u'-v')-(u-v)(2u'+v')\li_3\left(-e^{i\pi u'}\right)-(2u+v)(u'-v')\li_3\left(-e^{-i\pi u}\right)\right.\nonumber\\
&&\left.+(u-v)(u'+2v')\li_3\left(-e^{i\pi v'}\right)+(u+2v)(u'-v')\li_3\left(-e^{-i\pi v}\right)\right.\nonumber\\
&&\left.+(2u+v)(2u'+v')\li_3\left(e^{-i\pi(u-u')}\right)+(u+2v)(u'+2v')\li_3\left(e^{-i\pi(v-v')}\right)\right.\nonumber\\
&&\left.-(2u+v)(u'+2v')\li_3\left(e^{-i\pi(u-v')}\right)-(u+2v)(2u'+v')\li_3\left(e^{-i\pi(v-u')}\right)\right],
\eea

Here we list the expressions of $\ws,\wss,\wsss$ and $\wssss$ defined in Eqn.~(\ref{eq:sone}) in terms of Polylogarithms.
\bea
\ws&\equiv& \frac{1}{4\pi}\sum_{n=1}^\infty \frac{1}{(2n-1)}\left(e^{-i\left(n-\frac{1}{2}\right)\pi u}+e^{-i\left(n-\frac{1}{2}\right)\pi v}\right)\left(e^{i\left(n-\frac{1}{2}\right)\pi u'}+e^{i\left(n-\frac{1}{2}\right)\pi v'}\right)\nonumber\\
&=&\frac{1}{4\pi}\left[\li_1\left(e^{-i\pi\frac{(u-u')}{2}}\right)+\li_1\left(e^{-i\pi\frac{(u-v')}{2}}\right)+\li_1\left(e^{-i\pi\frac{(v-u')}{2}}\right)+\li_1\left(e^{-i\pi\frac{(v-v')}{2}}\right)\right]\nonumber\\
&&-\frac{1}{8\pi}\left[\li_1\left(e^{-i\pi(u-u')}\right)+\li_1\left(e^{-i\pi(u-v')}\right)+\li_1\left(e^{-i\pi(v-u')}\right)+\li_1\left(e^{-i\pi(v-v')}\right)\right],
\eea
\bea
\wss &\equiv& \frac{1}{4\pi}\sum_{n=1}^\infty\frac{1}{2n-1}\left(e^{-i\left(n-\frac{1}{2}\right)\pi u}+e^{-i\left(n-\frac{1}{2}\right)\pi v}\right)\Psi_S^*(n,u',v')\nonumber\\
&=& \frac{im^2v'}{4\pi^2}\left[\li_2\left(e^{-i\pi\frac{(u-u')}{2}}\right)+\li_2\left(e^{-i\pi\frac{(v-u')}{2}}\right)-\frac{1}{4}\li_2\left(e^{-i\pi(u-u')}\right)-\frac{1}{4}\li_2\left(e^{-i\pi(v-u')}\right)\right]\nonumber\\
&&+\frac{im^2u'}{4\pi^2}\left[\li_2\left(e^{-i\pi\frac{(u-v')}{2}}\right)+\li_2\left(e^{-i\pi\frac{(v-v')}{2}}\right)-\frac{1}{4}\li_2\left(e^{-i\pi(u-v')}\right)-\frac{1}{4}\li_2\left(e^{-i\pi(v-v')}\right)\right]\nonumber\\
&&-\frac{m^4v'^2}{16\pi^3}\left[\li_2\left(e^{-i\pi\frac{(u-u')}{2}}\right)+\li_2\left(e^{-i\pi\frac{(v-u')}{2}}\right)-\frac{1}{4}\li_2\left(e^{-i\pi(u-u')}\right)-\frac{1}{4}\li_2\left(e^{-i\pi(v-u')}\right)\right]\nonumber\\
&&-\frac{m^4u'^2}{16\pi^3}\left[\li_2\left(e^{-i\pi\frac{(u-v')}{2}}\right)+\li_2\left(e^{-i\pi\frac{(v-v')}{2}}\right)-\frac{1}{4}\li_2\left(e^{-i\pi(u-v')}\right)-\frac{1}{4}\li_2\left(e^{-i\pi(v-v')}\right)\right],\nonumber\\
\eea
\bea
\wsss &\equiv& \frac{1}{4\pi}\sum_{n=1}^\infty\frac{1}{2n-1}\Psi_S(n,u,v)\left(e^{i\left(n-\frac{1}{2}\right)\pi u'}+e^{i\left(n-\frac{1}{2}\right)\pi v'}\right)\nonumber\\
&=& -\frac{im^2v}{4\pi^2}\left[\li_2\left(e^{-i\pi\frac{(u-u')}{2}}\right)+\li_2\left(e^{-i\pi\frac{(v-u')}{2}}\right)-\frac{1}{4}\li_2\left(e^{-i\pi(u-u')}\right)-\frac{1}{4}\li_2\left(e^{-i\pi(v-u')}\right)\right]\nonumber\\
&&-\frac{im^2u}{4\pi^2}\left[\li_2\left(e^{-i\pi\frac{(u-v')}{2}}\right)+\li_2\left(e^{-i\pi\frac{(v-v')}{2}}\right)-\frac{1}{4}\li_2\left(e^{-i\pi(u-v')}\right)-\frac{1}{4}\li_2\left(e^{-i\pi(v-v')}\right)\right]\nonumber\\
&&-\frac{m^4v^2}{16\pi^3}\left[\li_2\left(e^{-i\pi\frac{(u-u')}{2}}\right)+\li_2\left(e^{-i\pi\frac{(v-u')}{2}}\right)-\frac{1}{4}\li_2\left(e^{-i\pi(u-u')}\right)-\frac{1}{4}\li_2\left(e^{-i\pi(v-u')}\right)\right]\nonumber\\
&&-\frac{m^4u^2}{16\pi^3}\left[\li_2\left(e^{-i\pi\frac{(u-v')}{2}}\right)+\li_2\left(e^{-i\pi\frac{(v-v')}{2}}\right)-\frac{1}{4}\li_2\left(e^{-i\pi(u-v')}\right)-\frac{1}{4}\li_2\left(e^{-i\pi(v-v')}\right)\right],\nonumber\\
\eea
\bea
\wssss &\equiv& \frac{1}{4\pi}\sum_{n=1}^\infty\frac{1}{2n-1}\Psi_S(n,u,v)\Psi_S^*(n,u',v')\nonumber\\
&=&\frac{m^4}{4\pi^3}\left[vv'\left(\li_3\left(e^{-i\pi\frac{(u-u')}{2}}\right)-\frac{1}{8}\li_3\left(e^{-i\pi(u-u')}\right)\right)\right.+vu'\left(\li_3\left(e^{-i\pi\frac{(u-v')}{2}}\right)-\frac{1}{8}\li_3\left(e^{-i\pi(u-v')}\right)\right)\nonumber\\
&&\quad+ uv'\left(\li_3\left(e^{-i\pi\frac{(v-u')}{2}}\right)-\frac{1}{8}\li_3\left(e^{-i\pi(v-u')}\right)\right)+\left.uu'\left(\li_3\left(e^{-i\pi\frac{(v-v')}{2}}\right)-\frac{1}{8}\li_3\left(e^{-i\pi(v-v')}\right)\right)\right],\nonumber\\
\eea

\section{Modifying the inner product to get the 2D Rindler Vacuum}
\label{sec:rindler}
In this section we obtain the massless Rindler Wightman function in the right Rindler Wedge as a particular limit of the
massless SJ Wightman function in 2D causal diamond. We achieve this by deviating from the standard $\mL^2$ inner product
on the function space $\cF(M,g)$, by  introducing a suitable non-trivial weight function $w(X)$, 
\be
(f,g)_w=\int_M f^*(X)g(X) \,  w(X) dV_X \label{eq:wip}
\ee
where $dV_X$ is the spacetime volume element. $w(X)$ takes real, positive and finite value for all $X$. The inner product
defined in Eqn.~(\ref{eq:wip}) is well defined in $(M,g)$ and  satisfies the defining properties of an  inner product: 
\begin{itemize}
\item $(f,g)_w$ is linear in $g$.
\item $(f,g)_w$ is anti-linear in $f$.
\item $(f,f)_w\geq 0$. Equality holds iff $f=0$.
\end{itemize}
Similarly,  we redefine the integral operator  $i\hD$  to make it consistent with this inner product 
\be
\left(i\hD\circ_w f\right)(X)=\int_M i\Delta(X,X')f(X') \,  w(X') dV_{X'}. \label{eq:ihdw}
\ee
It is straightforward to check that even with this modification, $i\hD$ is hermitian: 
\be
\left(f,i\hD\circ_w g\right)_w = \left(i\hD\circ_w f,g\right)_w. 
\ee
Next, we see that: 
\begin{claim}
$\kker(\Bkg)=\im_w(i\hD)$ for $w(X)$ real, positive and finite valued  in  $X$.
\end{claim}
\begin{proof}
For any  $\chi\in\im_w(i\hD)$,  there exists a $\psi\in\cF(M,g)$ such that $\chi=i\hD\circ_w\psi$. Since 
\be
i\hD\circ_w(\psi) = i\hD\circ (w\psi)
\ee
this implies that $\chi=i\hD\circ (w\psi) \in\im (i\hD)$, since  $w\psi\in\cF(M,g)$. Thus $\im_w (i\hD)\subseteq  \im (i\hD)$. Conversely, for
any $\chi'\in\im (i\hD)$, there exists a $\psi'\in\cF(M,g)$ such that $\chi'=i\hD\circ\psi'$. Since  $w$ is  real,
positive and finite valued  in  $X$,  $\psi/w \in \cF(M,g)$ and hence  $\chi'=i\hD\circ_w (\psi/w) \in \im_w
(i\hD)$. Hence $\im_w (i\hD) = \im (i\hD) = \kker(\Bkg)$.
\end{proof}

The 2D Minkowski metric in Rindler coordinates is 
\be
ds^2=e^{2a\xi}\left(-d\eta^2+d\xi^2\right)
\ee
where  
\be
t=a^{-1}e^{a\xi}\sinh(a\eta)\;,\quad x=a^{-1}e^{a\xi}\cosh(a\eta)
\ee
and $a>0$ is the acceleration parameter. Consider a causal diamond of length $2l$ centered at $(0,0)$ in 
$(\eta,\xi)$ coordinates. The center of the diamond  $(u,v)=(0,0)$ in the $u-v$
plane is at $(t,x)=(0,a^{-1})$,  and thus to the corner of the diamond in the $t-x$ plane as shown in Fig.~\ref{fig:rind}  
\begin{figure}[h]
\centering{\begin{tabular}{cc}
\includegraphics[height=4.5cm]{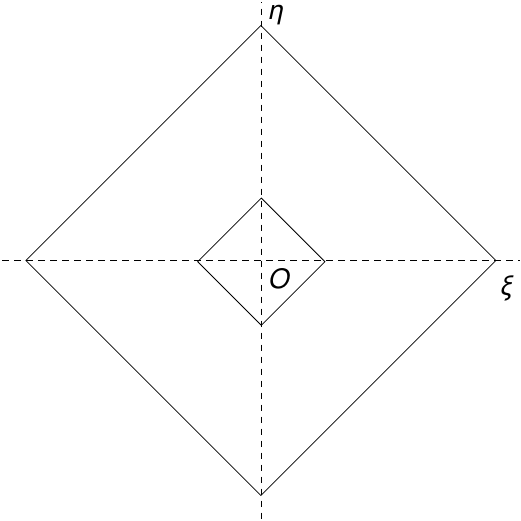}& \hskip 1.5cm
\includegraphics[height=4.5cm]{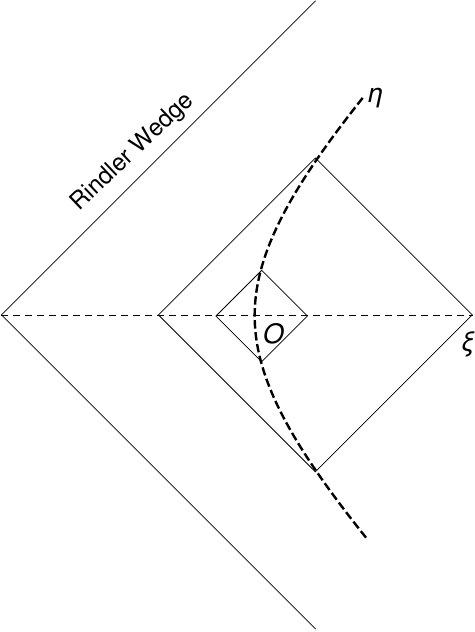}\\
(a)&\hskip 1.5cm (b)
\end{tabular}
\caption{A small causal diamond centered in a causal diamond $\diam$ in the $\eta-\xi$ plane is shifted to the corner of
  $\diam$ in the $t-x$ plane.}
\label{fig:rind}}
\end{figure}
The Pauli Jordan function is then similar to that in Minkowski coordinates
\be
i\Delta(u,v;u',v')=-\frac{i}{2}\left(\theta(u-u')+\theta(v-v')-1\right), 
\ee
where we have used the new light cone coordinates  $u=\frac{1}{\sqrt{2}}(\eta+\xi)$ and $ v=\frac{1}{\sqrt{2}}(\eta-\xi)$. 
The ``w-SJ''  modes  $u_k^w$ are then given by 
\bea
\int_{-L}^Li\Delta(u,v;u',v')u^w_k(u',v') w(u',v') e^{2a\xi'}du'dv'=\lambda_ku^w_k(u,v) \label{eq:rindev}
\eea
If  we now choose $w(u,v)=e^{-2a\xi}$, Eqn.~(\ref{eq:rindev}) is exactly the same as the eigenfunction equation for the
massless SJ modes in $\diam$ and hence  $\wsj$  is the same as the massless SJ function of \cite{Afshordi:2012ez}. Thus,
at the center of this diamond $\wsj$ takes the same form as Eqn.~(\ref{eq:minkm0}). The critical difference is that in
this case the $u$ and $v$ are lightcone coordinates for a  Rindler observer instead of an  inertial observer. Thus, in
$(t,x)$ coordinates, $\wsj$ {\it is} the Rindler vacuum (see Eqn (\ref{eq:rindm0})). The small diamond at the center
of $\diam$  the $\eta-\xi$ plane is a small diamond near (but not at)  the corner of  $\diam$ in the $t-x$
plane. Here, $\wsj$ then resembles the Rindler vacuum.   

Of course, the question is whether $\wsj$ will also look like $\wmink$ near the center of the diamond in the $t-x$ plane,
i.e. at $(t,x)=(0,a^{-1}\cosh(\sqrt{2}L a))$, which is $(0,a^{-1}\ln(\cosh(\sqrt{2}La)))$ in the $\eta-\xi$ plane. This is the
mirror vacuum, $\wmirr$ which rather than corresponding to  $\wmink$ is a  ``Rindler-mirror'' vacuum. This is clearly
not desirable. 

What we have presented here is a ``trick'' for achieving a desired form for the vacuum in the corner. However, this
messes up the expected form at the center. The question is
whether a smooth modification of $w$ from $1$ in the center of the $t-x$ plane diamond to $\exp(-a \xi)$ at the corners
could lead to  the desired form. However, modifications of the inner product mean that the SJ vacuum is no longer
unique.
\vskip 1in
\leftline{\bf\Large Erratum}

We correct a simulation error in our paper which led to the  incorrect conclusion in Sec.~\ref{sec:causet} that the causal set
Sorkin-Johnston Wightman function $\wsjc$ is incompatible with the Rindler Wightman function $\wrindm$ in the corner
of a 2d Minkowski diamond for a scalar field with large mass (with respect to the size of the diamond). Instead we find
that it is {\it as} compatible as the mirror  Wightman function $\wmirrm$, which we had shown is compatible with $\wsjc$
for all masses. As we discuss now, this seeming compatibility with $\wrindm$ can be traced to the fact that for our
simulations, $\wrindm \sim \wmirrm$ for large mass.  Note that this does not affect the analytic results of our paper for small mass, nor its
broader conclusions which remain unchanged.

The error in our paper was due to using  incompatible coordinates to simulate  $\wrindm$ which led to the erroneous Fig.~\ref{fig:sj-rind}. This was used to suggest 
that {\it only} $\wmirrm$ is compatible with  $\wsjc$.  We find instead, that there {\it is}  a correlation between  $\wsjc$ and $\wrindm$, but further analysis shows that $\wrindm$ and $\wmirrm$ themselves become indistinguishable for larger $m$.
To flesh this out we have explored a larger range of masses than discussed in our paper. Figs. \ref{fig:mirrsj} and
\ref{fig:rindsj} 
show that while the correlation of $\wsjc$ and $\wmirrm$ remains largely unchanged with mass, that with  $\wrindm$ 
 increases with mass.  This can be traced to the  increased correlation between $\wrindm$ and
$\wmirrm$ with mass as shown in Fig~\ref{fig:rindmirr}.  This in turn is related to the
dominance of $\wminkm$  in the expressions for $\wrindm$ and $\wmirrm$ (Eqns.~\eqref{eq:rindm} and \eqref{eq:mirrorm} in our paper)  for large
mass,   as shown in  Fig~\ref{fig:mrm}. 
The difference  is not captured by our current causal set simulations for which $N\sim 10,000$ elements are sprinkled  into the larger diamond (of height $2 \sqrt{2}$) to give
 $\sim 118$ elements in the corner diamond (of height $\sim 0.3$). Whether
 this ``degeneracy'' in choice of vacuum is broken with significantly larger simulations is a question we leave to
 future investigations.

 It was recently brought to our notice that the $m=10$ case was also studied in \cite{yasaman}.
\vskip 0.1in
{\bf{Acknowledgement:}}
We are grateful to Hans Muneesamy for pointing out the simulation error in our paper. His master's thesis \cite{hans} also discusses the $m=1$ and $m=10$ cases.

\begin{figure}[H]
 \centering
  \subfloat[$m=0.1$]{\includegraphics[height=1.9cm]{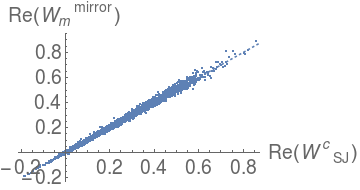} }\, 
  \subfloat[$m=0.2$]{\includegraphics[height=1.9cm]{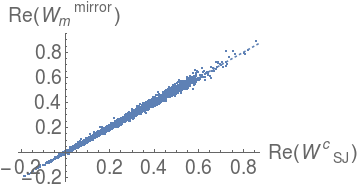}} \, 
  \subfloat[$m=0.3$]{\includegraphics[height=1.9cm]{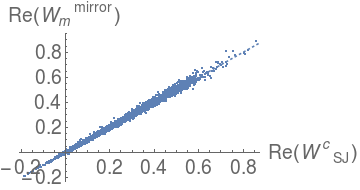}} \,
  \subfloat[$m=0.4$]{\includegraphics[height=1.9cm]{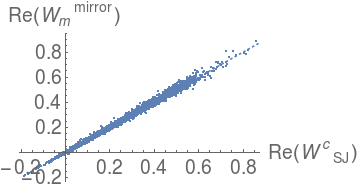} }\, 
  \subfloat[$m=1$]{\includegraphics[height=1.9cm]{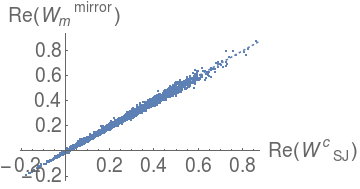}} \, 
  \subfloat[$m=2$]{\includegraphics[height=1.9cm]{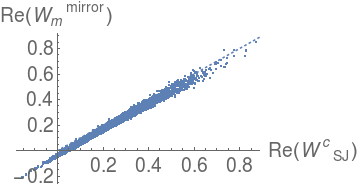}} \,
 \subfloat[$m=5$]{\includegraphics[height=1.9cm]{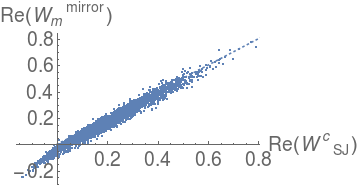}} \,
 \subfloat[$m=8$]{\includegraphics[height=1.9cm]{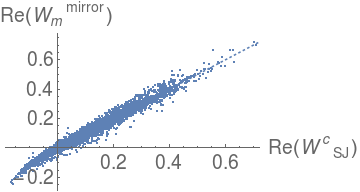}} \,
  \subfloat[$m=10$]{\includegraphics[height=1.9cm]{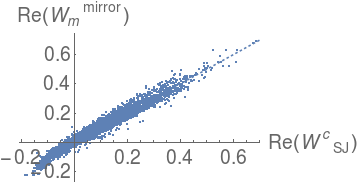}}\,
 \subfloat[$m=12$]{\includegraphics[height=1.9cm]{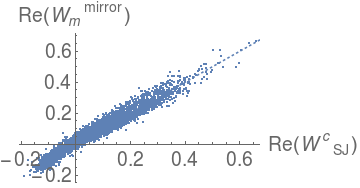}} \,
 \subfloat[$m=15$]{\includegraphics[height=1.9cm]{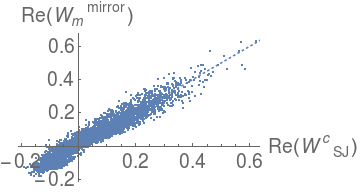}} 
  \caption{A correlation plot of the real parts of $\wsjc$ vs $\wmirrm$ in the left-hand corner of the 2d causal diamond for a range of masses. The
    diagonal is denoted by a dotted line. As is evident, the correlation remains largely unchanged with mass. The increase in scatter with mass is related to the fact that the density of sprinkling is left
 unchanged. }
   \label{fig:mirrsj} 
 \end{figure}
\begin{figure}[H]
 \centering
  \subfloat[$m=0.1$]{\includegraphics[height=1.9cm]{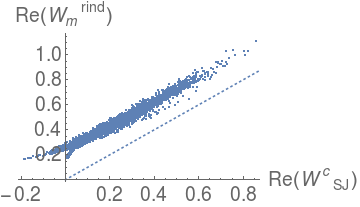} }\, 
  \subfloat[$m=0.2$]{\includegraphics[height=1.9cm]{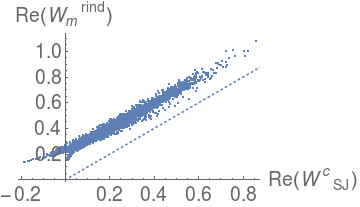}} \, 
  \subfloat[$m=0.3$]{\includegraphics[height=1.9cm]{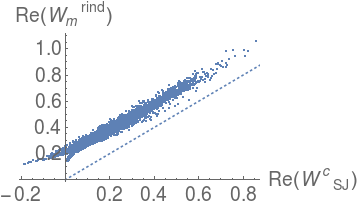}} \,
  \subfloat[$m=0.4$]{\includegraphics[height=1.9cm]{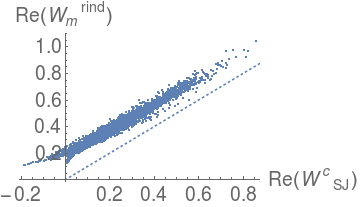} }\, 
  \subfloat[$m=1$]{\includegraphics[height=1.9cm]{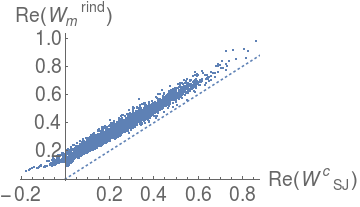}} \, 
  \subfloat[$m=2$]{\includegraphics[height=1.9cm]{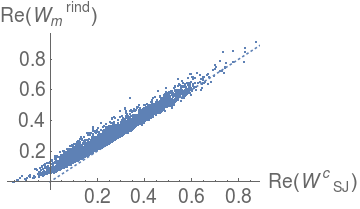}} \,
 \subfloat[$m=5$]{\includegraphics[height=1.9cm]{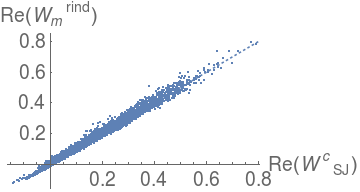}} \,
 \subfloat[$m=8$]{\includegraphics[height=1.9cm]{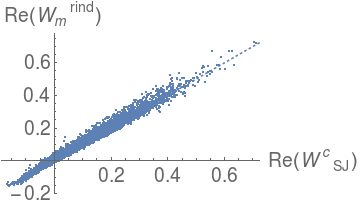}} \,
  \subfloat[$m=10$]{\includegraphics[height=1.9cm]{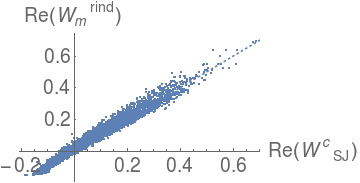}}\,
 \subfloat[$m=12$]{\includegraphics[height=1.9cm]{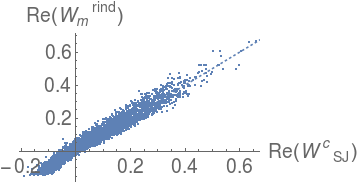}} \,
 \subfloat[$m=15$]{\includegraphics[height=1.9cm]{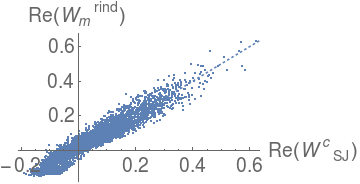}} 
  \caption{A correlation plot of the real parts of $\wsjc$ vs $\wrindm$ for the same
    range of masses. For small masses, the correlation is poor but improves
    with mass.}
   \label{fig:rindsj} 
 \end{figure}    
\begin{figure}[h]
    \centering
    \subfloat[$m=0.1$]{\includegraphics[height=1.9cm]{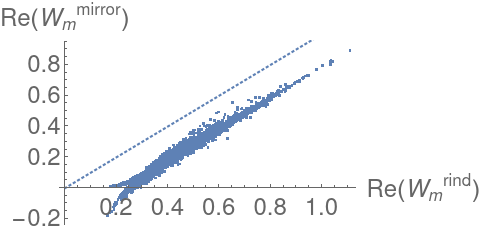}}\,
    \subfloat[$m=0.2$]{\includegraphics[height=1.9cm]{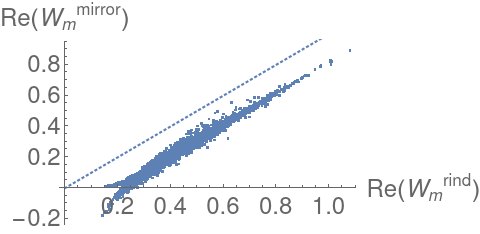}}\,
    \subfloat[$m=0.3$]{\includegraphics[height=1.9cm]{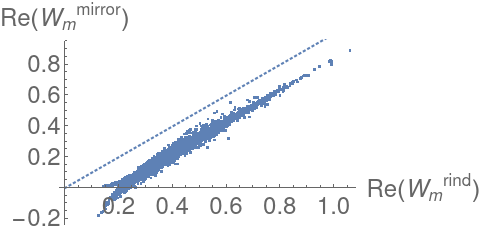}}\,    
    \subfloat[$m=0.4$]{\includegraphics[height=1.9cm]{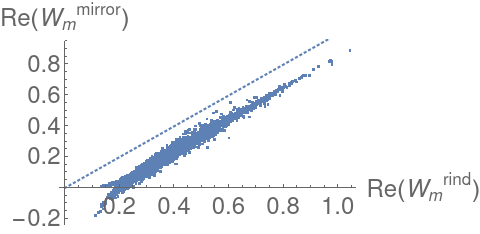}}\,
    \subfloat[$m=1$]{\includegraphics[height=1.9cm]{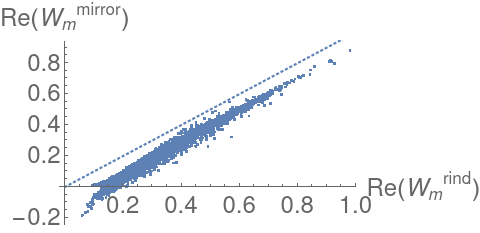}}\,
    \subfloat[$m=2$]{\includegraphics[height=1.9cm]{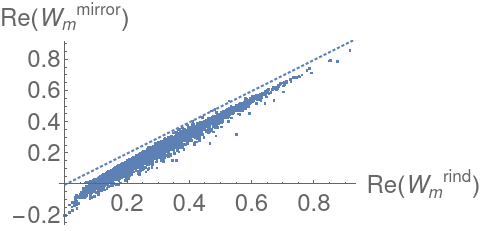}}\,
    \subfloat[$m=5$]{\includegraphics[height=1.9cm]{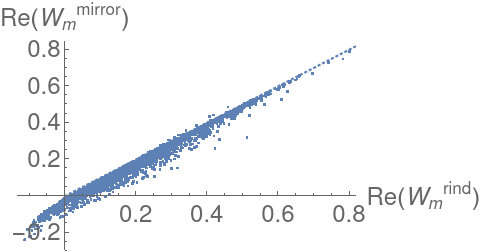}}\,
 \subfloat[$m=8$]{\includegraphics[height=1.9cm]{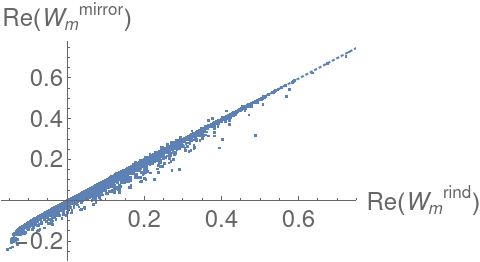}}\,
 \subfloat[$m=10$]{\includegraphics[height=1.9cm]{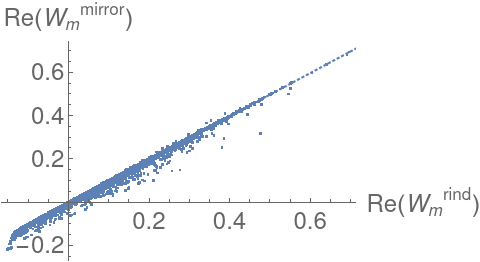}}\,
 \subfloat[$m=12$]{\includegraphics[height=1.9cm]{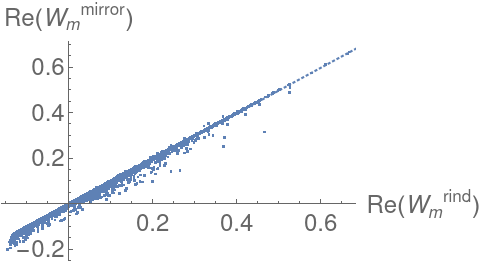}}\,
 \subfloat[$m=15$]{\includegraphics[height=1.9cm]{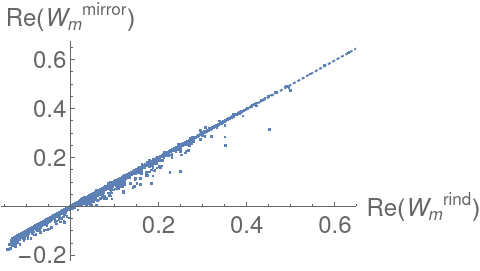}}\,
    \caption{A correlation plot of the real parts of $\wmirrm$ vs $\wrindm$ for the same
    range of masses. For small masses, the correlation is poor but improves
    with mass.}
    \label{fig:rindmirr}
  \end{figure}
 
  \begin{figure}[h]
    \centering
    \subfloat[$m=0.1$]{\includegraphics[height=2.4cm]{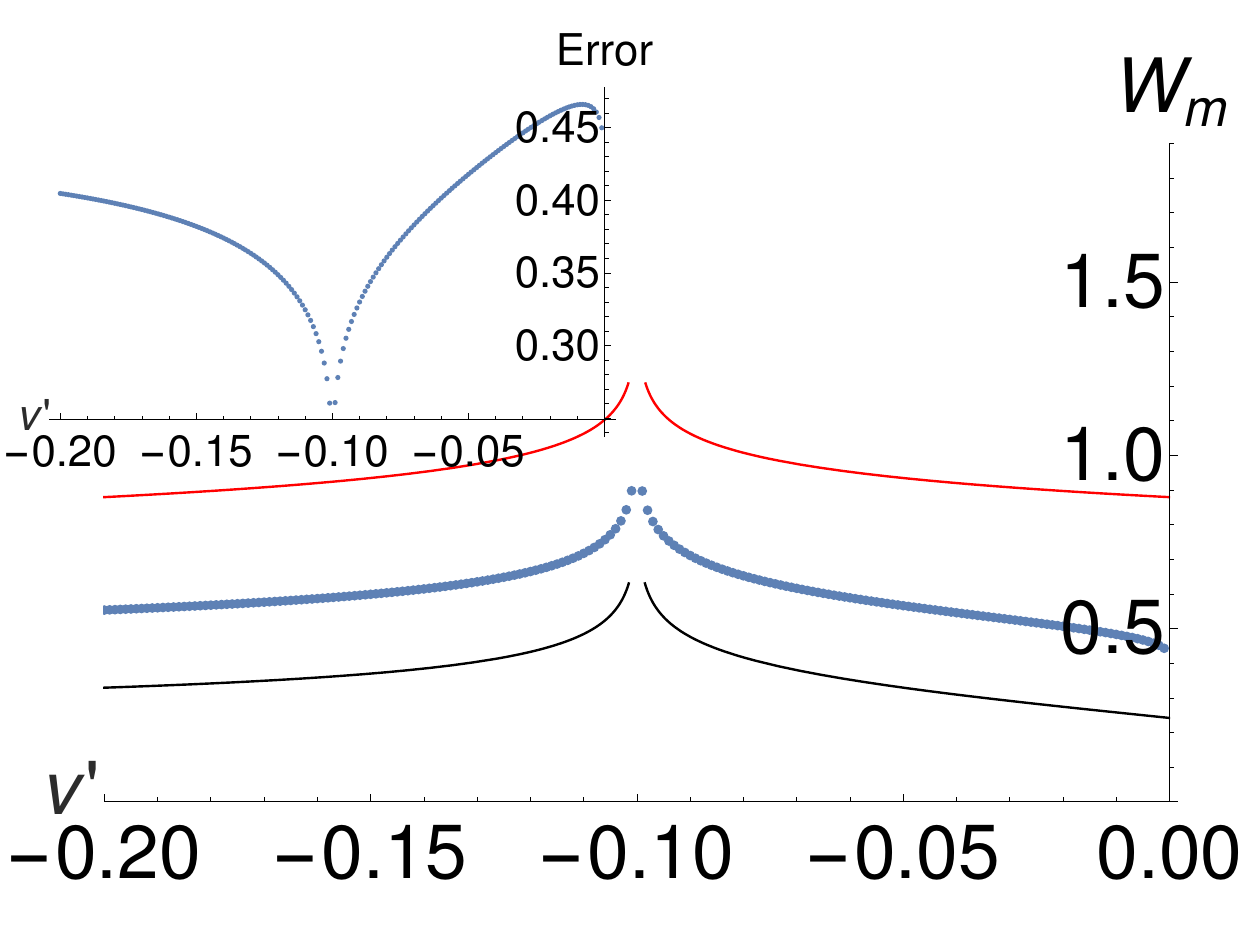}}\hskip 0.1in\,
    \subfloat[$m=0.2$]{\includegraphics[height=2.4cm]{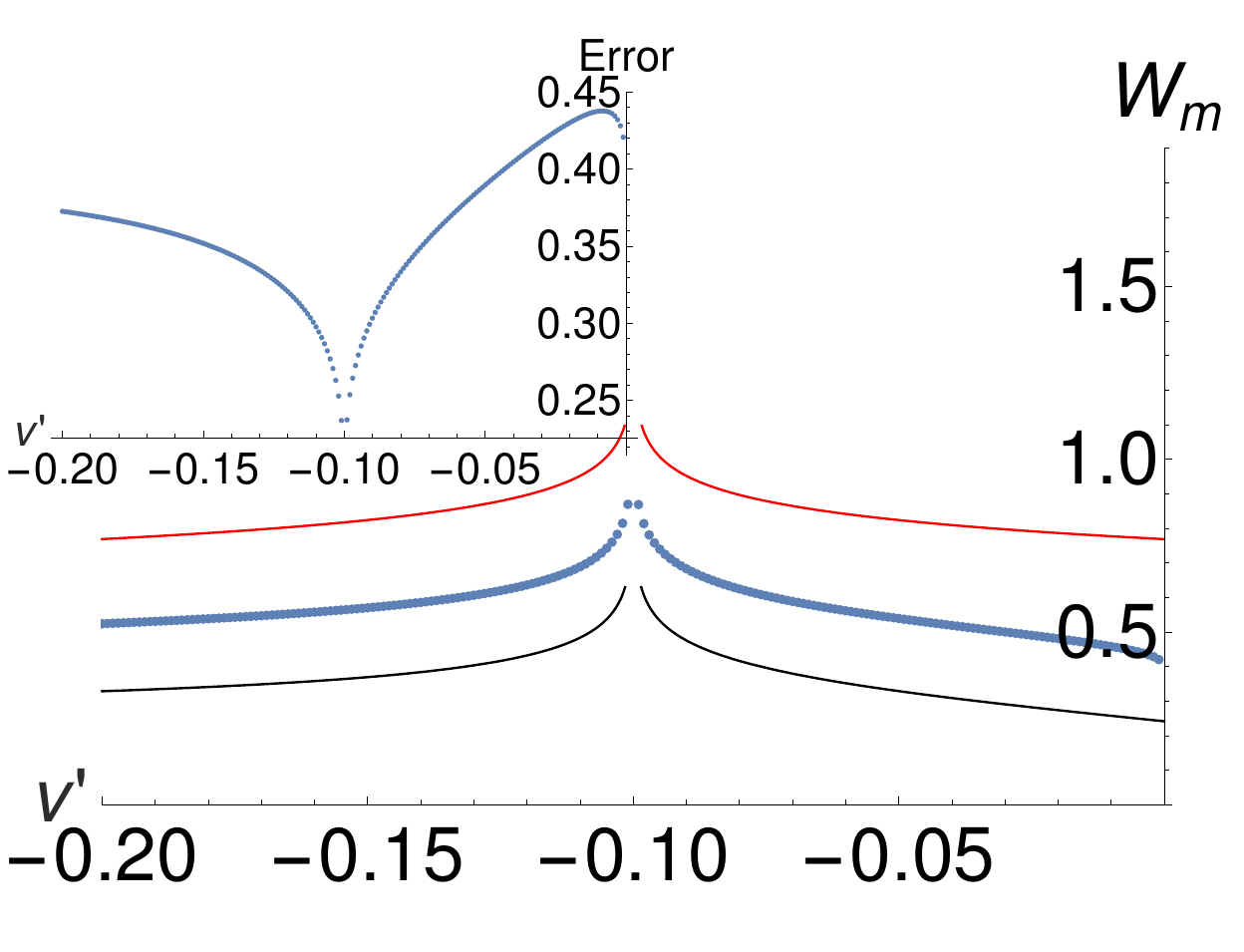}}\hskip 0.1in\,
    \subfloat[$m=0.3$]{\includegraphics[height=2.4cm]{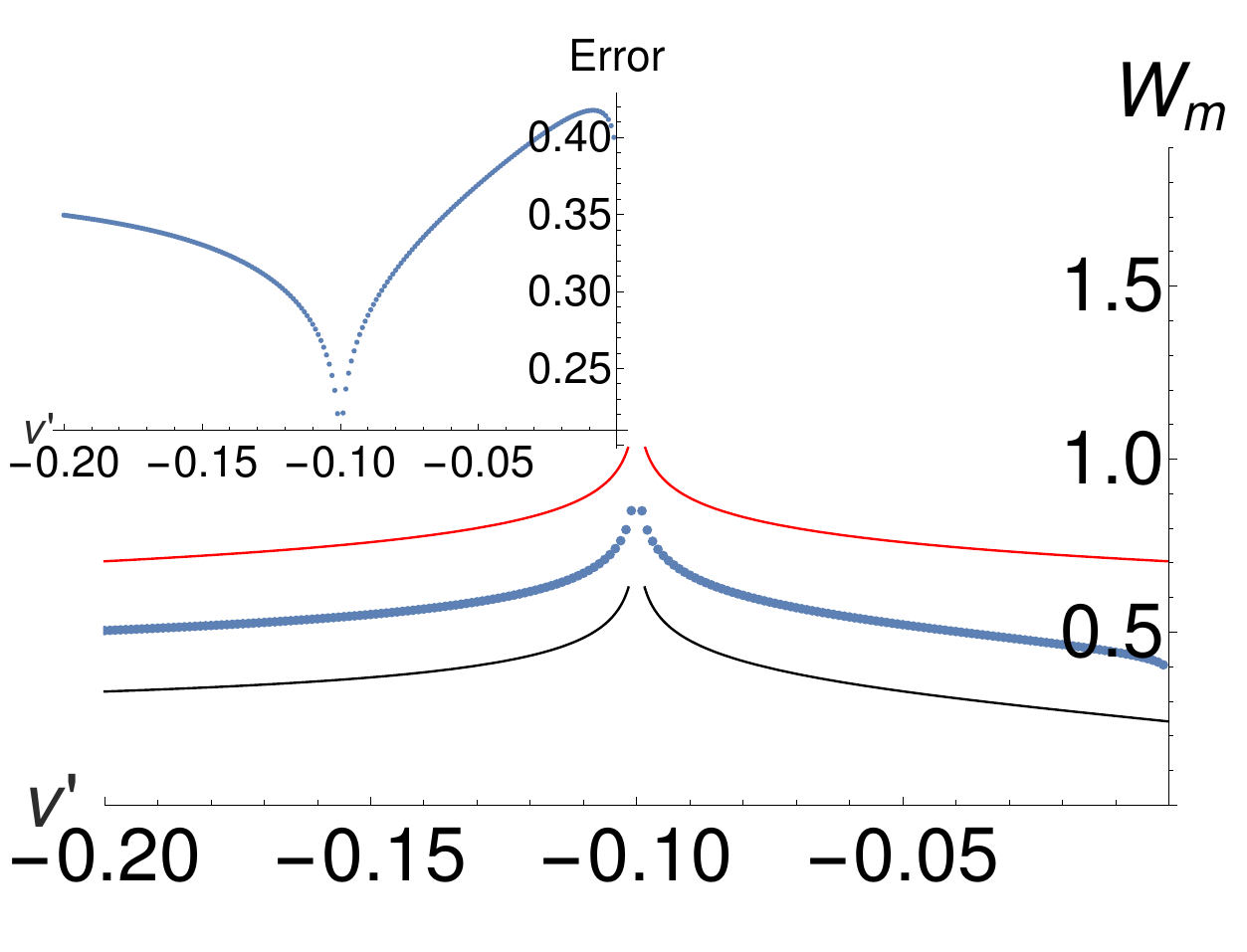}}\hskip 0.1in\,    
    \subfloat[$m=0.4$]{\includegraphics[height=2.4cm]{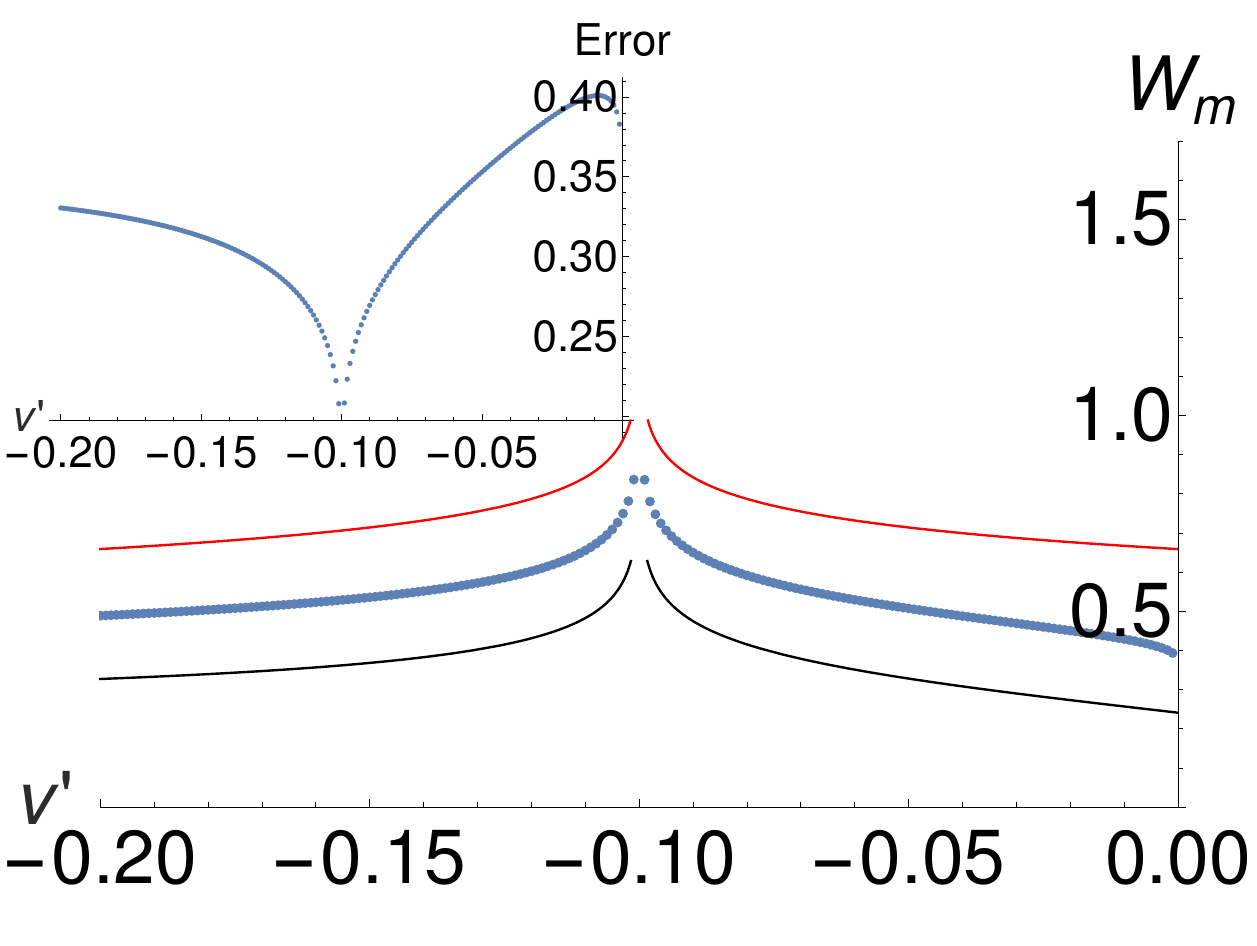}}\hskip 0.1in\,
 \subfloat[$m=1$]{\includegraphics[height=2.4cm]{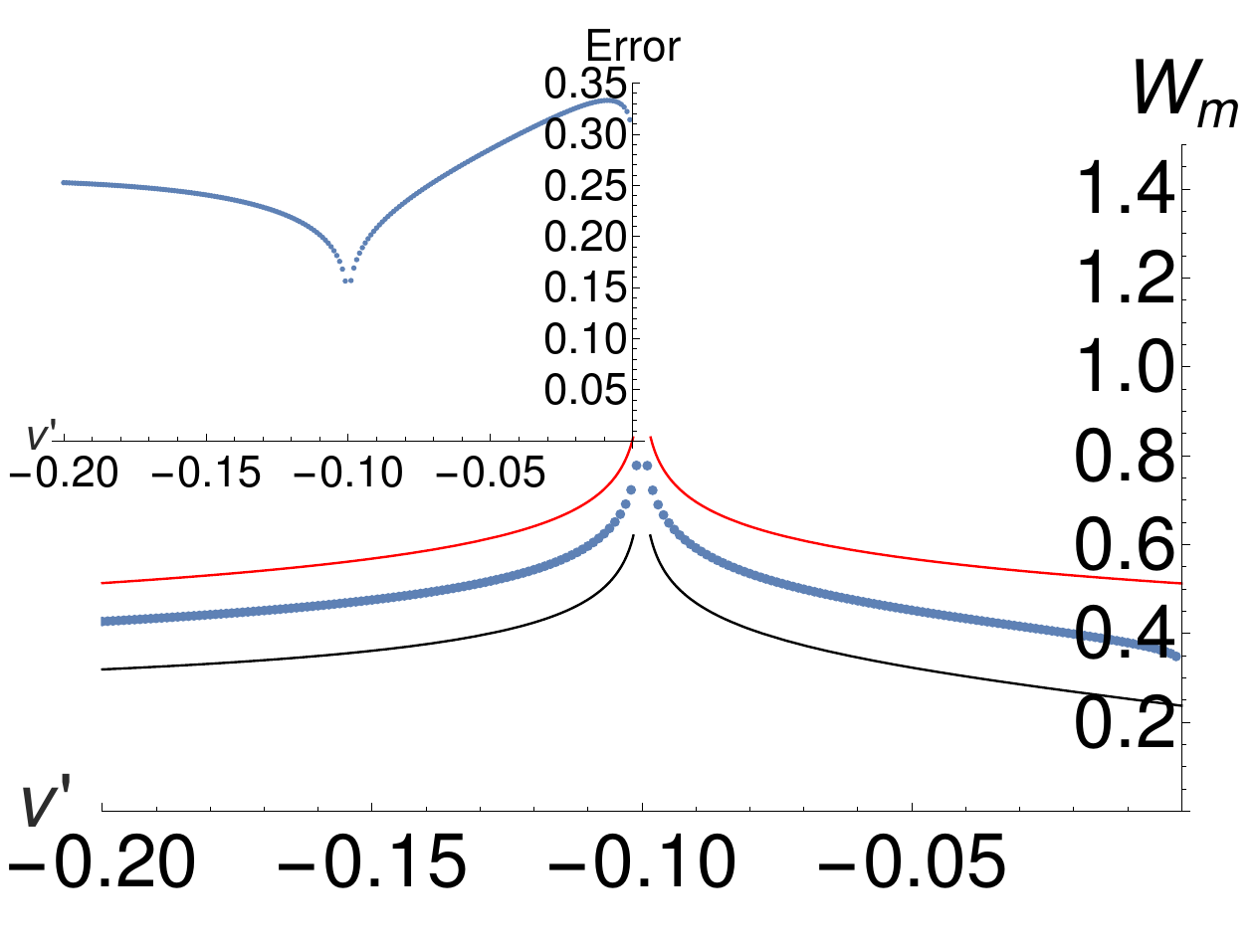}}\hskip 0.1in\,
    \subfloat[$m=2$]{\includegraphics[height=2.4cm]{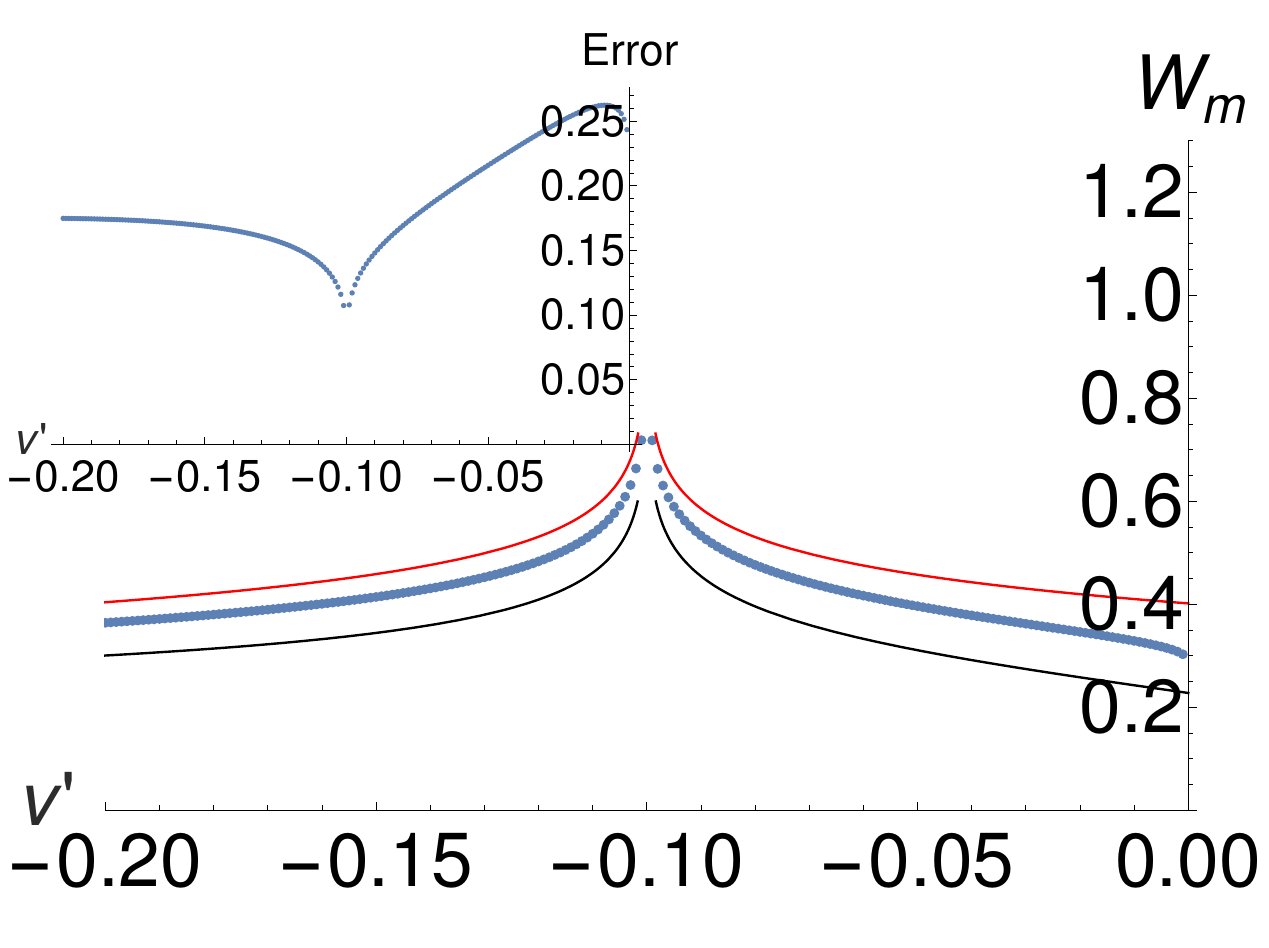}}\hskip 0.1in\,
    \subfloat[$m=5$]{\includegraphics[height=2.4cm]{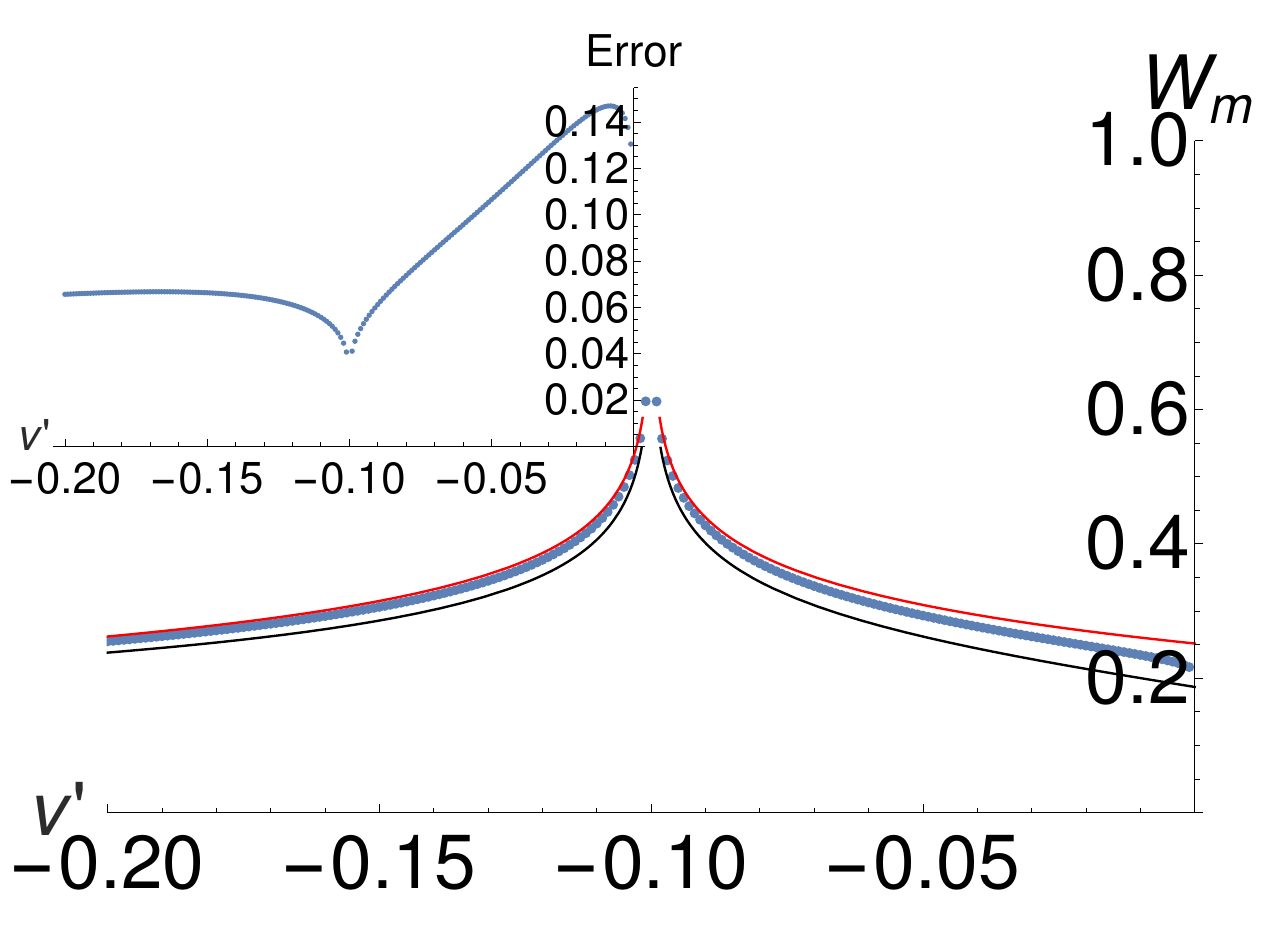}}\hskip 0.1in\,    
    \subfloat[$m=8$]{\includegraphics[height=2.4cm]{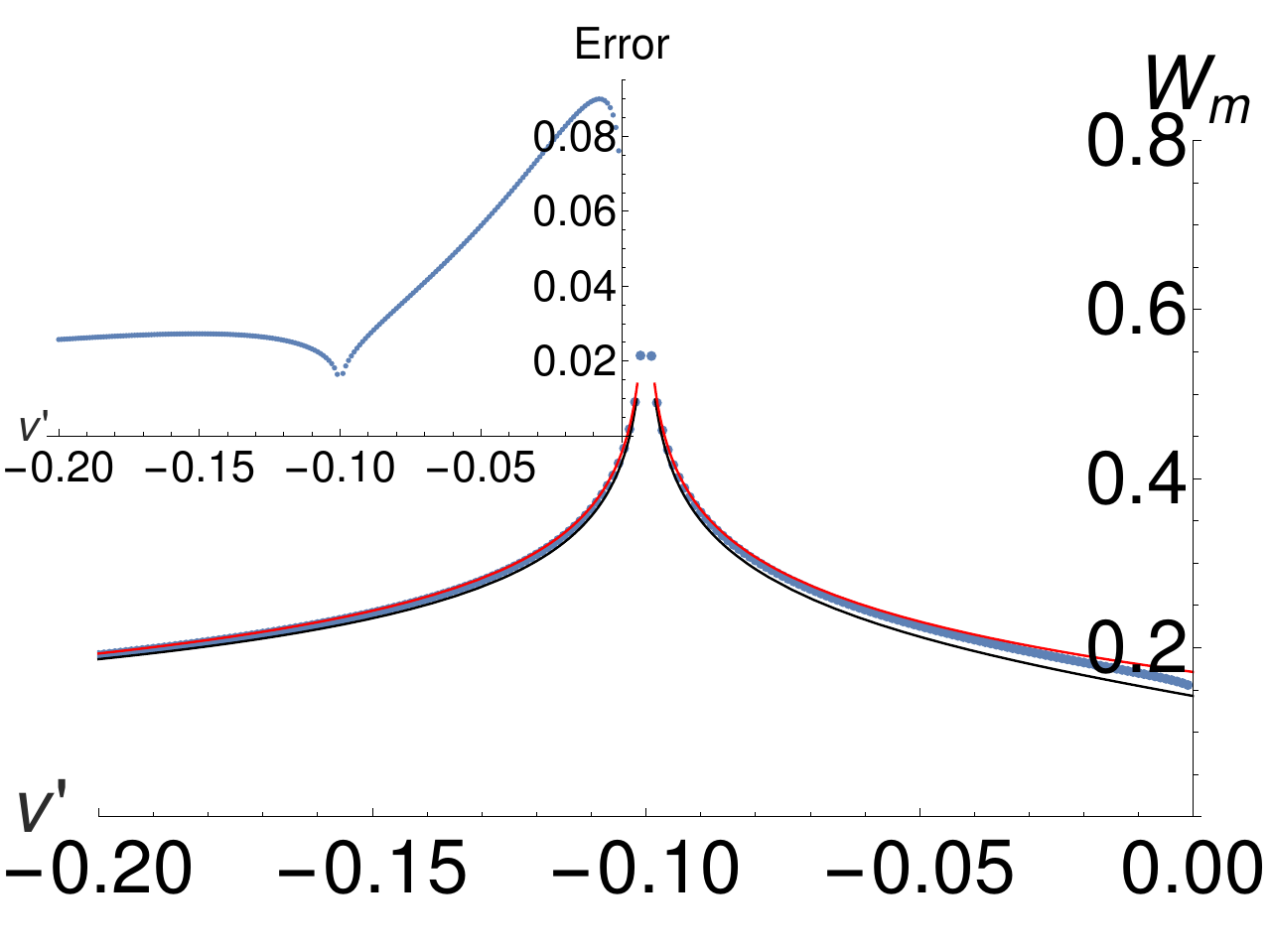}}\hskip 0.1in\,
 \subfloat[$m=10$]{\includegraphics[height=2.4cm]{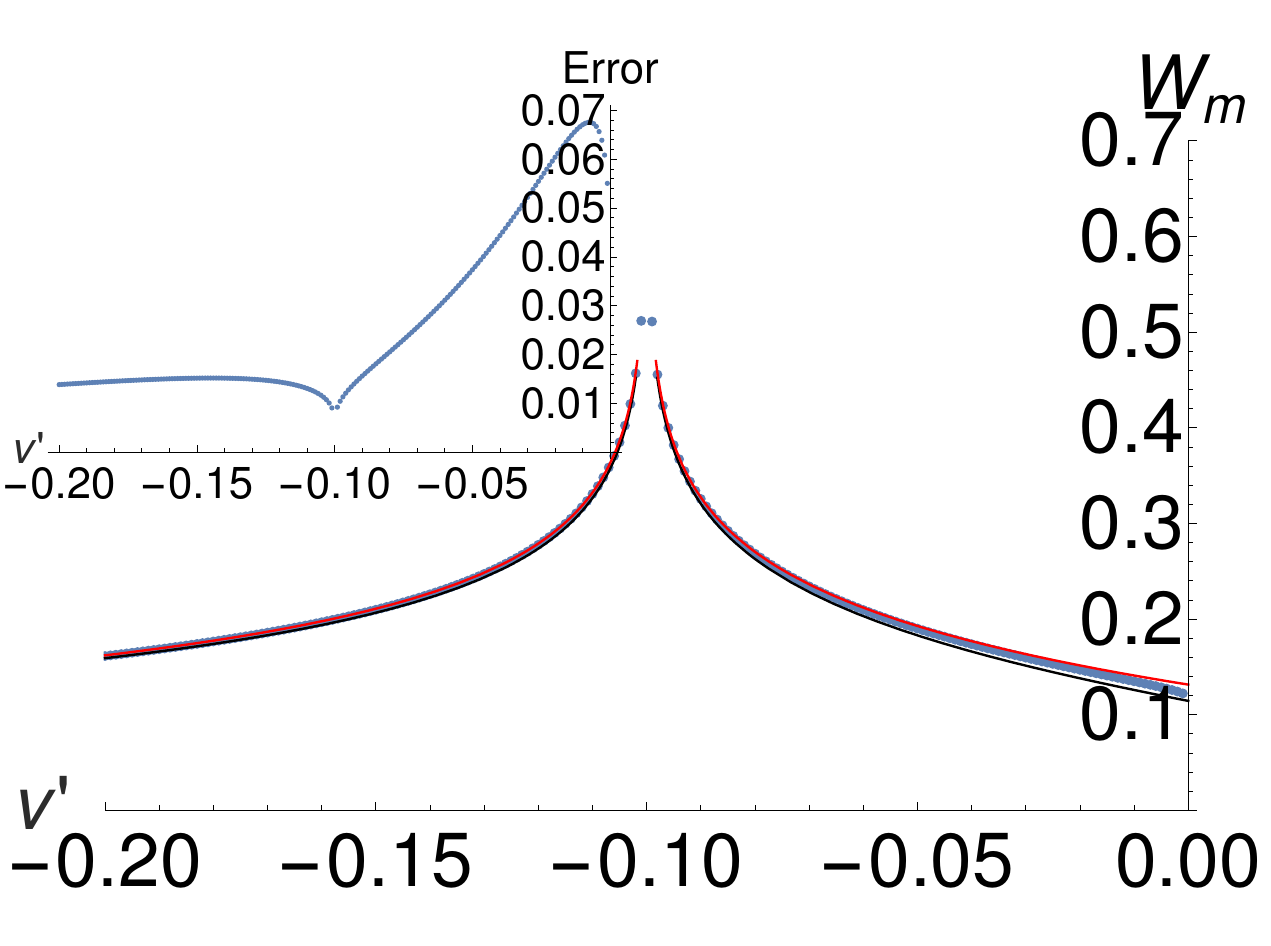}}\hskip 0.1in\,
    \subfloat[$m=12$]{\includegraphics[height=2.4cm]{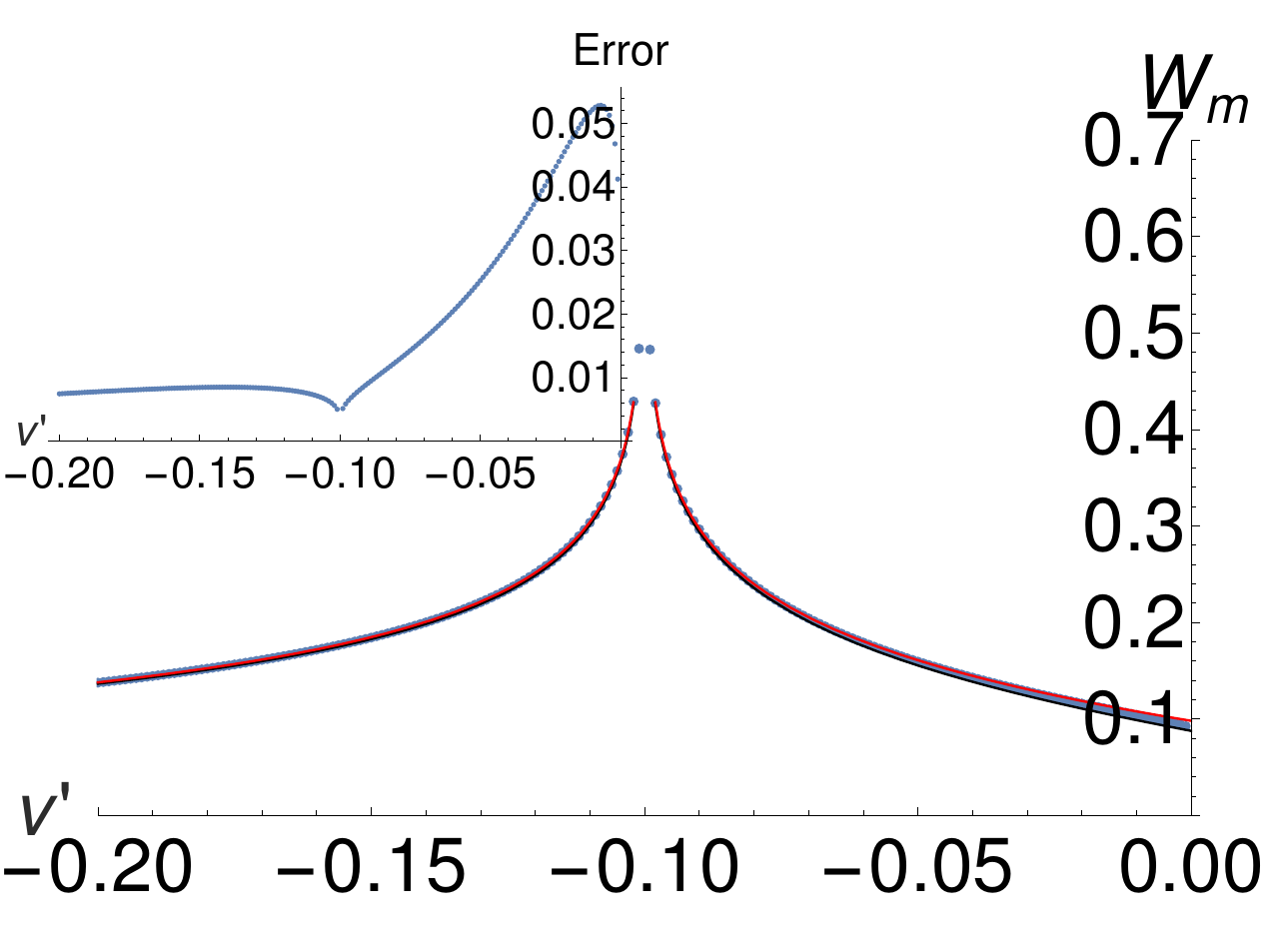}}\hskip 0.1in\,
    \subfloat[$m=15$]{\includegraphics[height=2.4cm]{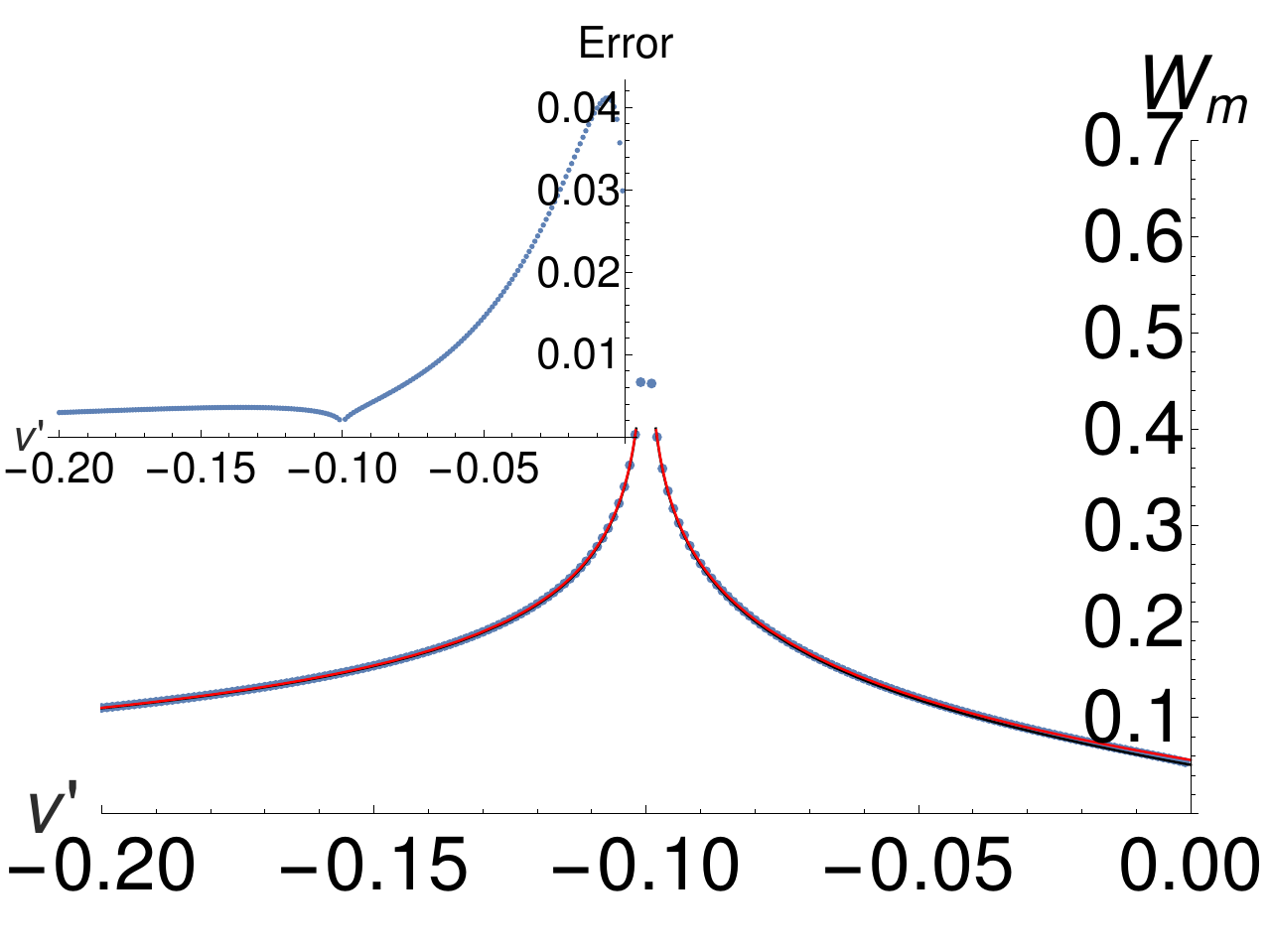}}
     \caption{Real parts of $\wminkm$ (red), $\wrindm$ (blue) and $\wmirrm$ (black) for a pair of points
      $(u=0.1,v=-0.1)$ and $(u'=0.11,v')$ with  varying $v'$. 
      As the mass increases all three converge to a common value.  To make the comparison explicit, the inset figure shows the relative error  between
    the real parts of   $\wrindm$ and $\wmirrm$ as a function of $v'$. } 
\label{fig:mrm}
\end{figure}

\bibliography{reference}
\bibliographystyle{ieeetr}
\end{document}